\documentclass[aps,prx,twocolumn,superscriptaddress]{revtex4-2}
\usepackage[T1]{fontenc}
\usepackage[utf8]{inputenc}
\usepackage{amsfonts}
\usepackage{amsmath}
\usepackage{amssymb}
\usepackage{amsthm}
\usepackage{esint}
\usepackage{graphicx}
\usepackage{bm}
\usepackage{color}
\usepackage{booktabs}
\usepackage{makecell}
\usepackage{xcolor}
\usepackage{mathrsfs}
\usepackage[colorlinks,bookmarks=true,citecolor=blue,linkcolor=red,urlcolor=blue]{hyperref}
\usepackage{appendix}
\usepackage{float}
\usepackage[export]{adjustbox}
\usepackage{ulem}
\usepackage{multirow}

\usepackage{tikz}
\usetikzlibrary{positioning}

\newcommand{\RNum}[1]{\uppercase\expandafter{\romannumeral #1\relax}}

\newcommand{\beq}{\begin{eqnarray} }
\newcommand{\eeq}{\end{eqnarray} }
\newcommand{\Beq}{\begin{eqnarray*} }
\newcommand{\Eeq}{\end{eqnarray*} }

\newcommand{\re}{\mathrm{Re}}

\newcommand{\calH}{\mathcal{H}}

\newcommand{\calE}{\mathcal{E}}

\newcommand{\calO}{\mathcal{O}}
\newcommand{\calP}{\mathcal{P}}
\newcommand{\calQ}{\mathcal{Q}}
\newcommand{\calS}{\mathcal{S}}

\newcommand{\calM}{\mathcal{M}}
\newcommand{\calR}{\mathcal{R}}
\newcommand{\calL}{\mathcal{L}}
\newcommand{\expt}{\mathbb{E}}
\newcommand{\prob}{\mathbb{P}}
\newcommand{\tr}{\mathrm{tr}}
\newcommand{\npow}{{n\text{-}\mathrm{pow}}}
\newcommand{\Onpow}{\langle O\rangle_{\npow}}
\newcommand{\Oth}{\langle O\rangle_{\text{th}}}
\newcommand{\Oeq}{\langle O\rangle_{\text{eq}}}
\newcommand{\rect}{\mathrm{rect}}

\newcommand{\CODE}{Note1}

\newcommand{\absenceIOMs}{shiraishiProofAbsenceLocal2019,fukaiAllLocalConserved2023,chibaProofAbsenceLocal2024a,yamaguchiProofAbsenceLocal2024,shiraishiAbsenceLocalConserved2024,yamaguchiCompleteClassificationIntegrability2024,hokkyoProofAbsenceLocal2025,chibaProofAbsenceLocal2025,parkProofNonintegrabilitySpin12025,parkGraphtheoreticalProofNonintegrability2025}

\newtheorem{theorem}{Theorem}
\newtheorem{proposition}{Proposition}

\renewcommand{\emph}[1]{\textit{#1}}



\usepackage{babel}

\begin{document}
\title{Simple slow operators and quantum thermalization}

\author{Tian-Hua Yang} \email{yangth@princeton.edu} \affiliation{Department of Physics, Princeton University, Princeton, New Jersey 08544, USA}
\author{Ethan Jeon} \affiliation{Department of Physics, Princeton University, Princeton, New Jersey 08544, USA}
\author{Sarang Gopalakrishnan} \email{sgopalakrishnan@princeton.edu} \affiliation{Department of Electrical and Computer Engineering, Princeton University, Princeton, New Jersey 08544, USA}
\author{Dmitry A. Abanin} \email{dabanin@princeton.edu} \affiliation{Department of Physics, Princeton University, Princeton, New Jersey 08544, USA}

\begin{abstract}

We establish a rigorous operator-space condition for the thermalization of typical initial states. We introduce a concept of simple slow operators (SSOs), defined as operators that evolve slowly under Heisenberg-picture dynamics and have significant small-sized components. We show that the non-thermalization of typical initial states drawn from a weakly entangled state ensemble implies the existence of SSOs; equivalently, the absence of SSOs implies that typical initial states thermalize. We establish these results by introducing the concept of an {\it ensemble variance norm} of an operator, defined as the typical magnitude of the expectation value of that operator with respect to states in the ensemble. For low-entanglement ensembles, the norm is related to operator sizes, allowing us to establish a direct link between operator dynamics and thermalization.
\end{abstract}

\maketitle
 
\section{Introduction}

One fundamental assumption of statistical mechanics is that many-body systems reach thermal equilibrium under their own dynamics, a process referred to as thermalization~\cite{pathriaStatisticalMechanics2011,sethnaStatisticalMechanicsEntropy2021a,kardarStatisticalPhysicsParticles2015}. Quantum thermalization is the study of whether and how closed quantum systems thermalize under unitary evolution, a problem that lies at the foundation of quantum statistical physics~\cite{eisertQuantumManybodySystems2015,gogolinEquilibrationThermalisationEmergence2016,dalessioQuantumChaosEigenstate2016b,borgonoviQuantumChaosThermalization2016,moriThermalizationPrethermalizationIsolated2018a,uedaQuantumEquilibrationThermalization2020a}. A rich body of results exists to explain the mechanism of quantum thermalization~\cite{deutschQuantumStatisticalMechanics1991a,srednickiChaosQuantumThermalization1994a,kimTestingWhetherAll2014c,deutschEigenstateThermalizationHypothesis2018a,rigolAlternativesEigenstateThermalization2012,wangEigenstateThermalizationHypothesis2022a,pappalardiEigenstateThermalizationHypothesis2022b,hahnEigenstateCorrelationsEigenstate2024a}, and it has been demonstrated that thermalization should occur in generic quantum many-body systems~\cite{bertoniTypicalThermalizationLowentanglement2025,pilatowsky-cameoQuantumThermalizationMust2025}. As such, it is a particularly interesting question to find classes of quantum systems which avoid thermalization.

Over the decades, physicists have identified numerous systems that behave in non-thermal ways, such as integrable systems~\cite{cauxRemarksNotionQuantum2011a,cassidyGeneralizedThermalizationIntegrable2011a,pozsgayGeneralizedGibbsEnsemble2013,pozsgayCorrelationsQuantumQuenches2014a,woutersQuenchingAnisotropicHeisenberg2014}, many-body localized systems~\cite{nandkishoreManyBodyLocalizationThermalization2015, abaninColloquiumManybodyLocalization2019a, palManybodyLocalizationPhase2010a, serbynLocalConservationLaws2013,husePhenomenologyFullyManybodylocalized2014, voskTheoryManyBodyLocalization2015,chandranConstructingLocalIntegrals2015,rademakerExplicitLocalIntegrals2016}, systems with quantum many-body scars~\cite{turnerWeakErgodicityBreaking2018a,shiraishiSystematicConstructionCounterexamples2017a,moudgalyaPairingHubbardModels2020,serbynQuantumManybodyScars2021,renQuasisymmetryGroupsManyBody2021a,moudgalyaExhaustiveCharacterizationQuantum2024a,moudgalyaQuantumManybodyScars2022b}, and Hilbert space fragmented systems~\cite{salaErgodicityBreakingArising2020,yangHilbertSpaceFragmentationStrict2020a,moudgalyaQuantumManybodyScars2022b,moudgalyaHilbertSpaceFragmentation2022b,balasubramanianGlassyWordProblems2024,chenQuantumFragmentationExtended2024,zhouQuantumHilbertSpace2026}. There are also systems that thermalize but only on long timescales~\cite{bergesPrethermalization2004,rigolBreakdownThermalizationFinite2009,santosOnsetQuantumChaos2010,reimannTypicalityPrethermalization2019,moriThermalizationPrethermalizationIsolated2018a}, sometimes exponentially long~\cite{abaninRigorousTheoryManyBody2017b,abaninEffectiveHamiltoniansPrethermalization2017,yinPrethermalizationLocalRobustness2023a,wurtzEmergentConservationLaws2020,vanovacFinitesizeGeneratorsWeak2024}, meaning that they behave non-thermally over a range of physically relevant timescales, a phenomenon known as prethermalization. It is of particular interest to find a universal characterization of all these non-thermal systems. A strong candidate is the presence of integrals of motion (IOMs), which are operators that commute with the Hamiltonian (\textit{i.e.}, are conserved). It is a common belief that non-thermal behavior can always be attributed to the existence of some form of IOM (apart from quantum many-body scars, which we will discuss in Subsection~\ref{subsec:scars}); furthermore, prethermal behavior is often associated with approximate IOMs. However, to our knowledge, no rigorous result formalizing this intuition exists in the literature. One major challenge in establishing such a connection is that the concept of IOMs is actually more restrictive than ``conserved operators'': projectors onto a Hamiltonian's eigenstates exist as conserved operators in any system, yet they generally do not prevent thermalization and are not considered IOMs.
It is usually implicitly understood that IOMs relevant to thermalization must have certain locality properties. However, the exact extent of this locality constraint has not been spelled out. Moreover, theories of thermalization and infinite-time dynamics are typically based on an eigenstate formulation, where operator locality is further obscured.

In this Article, we introduce the concept of ``simple slow operators'' (SSOs) that expands and makes precise the notion of IOMs with locality properties, and establish a rigorous relationship between the absence of SSOs and thermalization without reference to eigenstates. Our result is based on two intuitive facts: firstly, for generic quantum many-body systems in the Heisenberg picture, local operators experience operator growth~\cite{xuScramblingDynamicsOutofTimeOrdered2024,khemaniOperatorSpreadingEmergence2018d,parkerUniversalOperatorGrowth2019b,gopalakrishnanHydrodynamicsOperatorSpreading2018b,schusterOperatorGrowthOpen2023,robertsOperatorGrowthSYK2018,swingleUnscramblingPhysicsOutoftimeorder2018,swingleMeasuringScramblingQuantum2016b}; secondly, an operator of large operator size would, on average, have a small expectation value on a random weakly entangled state~\cite{qiMeasuringOperatorSize2019,huangPredictingManyProperties2020b,hu2023ClassicalShadow,Ippoliti2023OperatorRelaxation,Hu2023ScalableFlexible,Bertoni2024ShallowShadows,bu2024ClassicalShadows}. From these observations, we have the following rough picture: an operator would have a small expectation value at late time for most physical initial states, unless it overlaps with an operator that both has small-sized components and evades operator growth.

Motivated by this, we define \textit{simple} operators by introducing an \textit{ensemble variance norm}, which characterizes the average magnitude of the expectation value of an operator $O$ with respect to typical states drawn from a given ensemble. When the ensemble consists of random low-entanglement states, this ensemble variance norm exponentially penalizes large-sized operators. We then call an operator ``simple'' if it has a large ensemble variance; it then follows that simple operators are exactly operators that have significant small-sized components (\textit{i.e.}, overlap with local operators). We define \textit{slow} operators to be all operators such that $[H,A]\approx \lambda A$, which will include approximate conserved quantities (corresponding to $\lambda=0$) and approximate spectrum generating algebra (SGA)~\cite{moudgalyaPairingHubbardModels2020,odeaTunnelsTowersQuantum2020,renQuasisymmetryGroupsManyBody2021a,serbynQuantumManybodyScars2021,moudgalyaQuantumManybodyScars2022b,moudgalyaExhaustiveCharacterizationQuantum2024a,bucaUnifiedTheoryLocal2023a}. This definition is motivated by the fact that $[H,A]=\lambda A$ implies $A(t) = e^{i\lambda t} A$, thus avoiding operator growth.

We prove that most initial states drawn from the state ensemble thermalize if the system does not have any SSOs, defined as operators that are both simple and slow. These results hold at infinitely long times as well as at finite times. In the latter case, thermalization on timescale $\tau$ follows from the absence of SSOs that are slow on at least a timescale $\tau$. Equivalently, if typical states do not thermalize, it follows that SSOs exist. We further show that SSOs can be numerically identified via a generalized eigenvalue problem in the operator space. The results can be generalized to Floquet systems straightforwardly.

\begin{figure}[!t]
    \centering
    \includegraphics{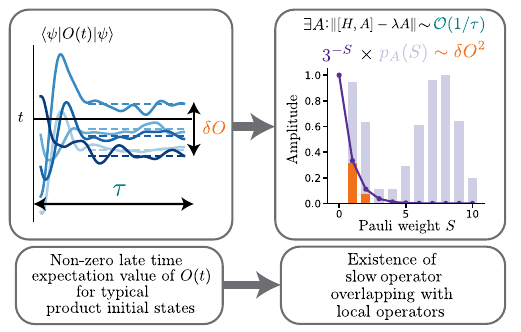}
    \caption{An illustration of the main result of this paper. We demonstrate that the late-time expectation value of any observable $O$ being non-vanishing for most choices of product initial states implies the existence of SSOs, \textit{i.e.}, slow operators that overlap with local operators. In Hamiltonian systems, we choose $O$ to be an observable orthogonalized against $H$ and its powers. The result would imply that if the late-time expectation value is different from the thermal expectation value for most product initial states, then there must exist SSOs that are orthogonal to the Hamiltonian and its powers.}
    \label{fig:intro}
\end{figure}

Our formalism establishes a rigorous connection between the operator-space and state-space viewpoints on quantum thermalization. It bridges three active areas of research---rigorous theorems of equilibration and thermalization~\cite{lindenQuantumMechanicalEvolution2009,shortEquilibrationQuantumSystems2011,shortQuantumEquilibrationFinite2012,eisertQuantumManybodySystems2015,biroliEffectRareFluctuations2010b,moriWeakEigenstateThermalization2016a,bucaUnifiedTheoryLocal2023a,pilatowsky-cameoQuantumThermalizationMust2025,cipolloniEigenstateThermalizationHypothesis2021,erdosEigenstateThermalizationHypothesis2024,huangRandomProductStates2024,bertoniTypicalThermalizationLowentanglement2025},
identifying the existence or proving the absence of conserved operators~\cite{kimSlowestLocalOperators2015a,lianConservedQuantitiesEntanglement2022b,chenGeneralizedGibbsEnsemble2025,moudgalyaSymmetriesGroundStates2024a,crapsUnifiedTheoryLocal2025,pawlowskiFindingLocalIntegrals2025,pawlowskiLocalIntegralsMotion2026,\absenceIOMs},
and the study of operator growth~\cite{xuScramblingDynamicsOutofTimeOrdered2024,khemaniOperatorSpreadingEmergence2018d,parkerUniversalOperatorGrowth2019b,gopalakrishnanHydrodynamicsOperatorSpreading2018b,schusterOperatorGrowthOpen2023,robertsOperatorGrowthSYK2018,swingleUnscramblingPhysicsOutoftimeorder2018,swingleMeasuringScramblingQuantum2016b} and operator size distribution in real-time dynamics~\cite{qiMeasuringOperatorSize2019,nahumOperatorSpreadingRandom2018d,zhangDynamicalTransitionOperator2023}---and enables results to be transferred between these areas. 
We mention a few applications in this regard: (1) the dynamics of operator weights~\cite{qiMeasuringOperatorSize2019,nahumOperatorSpreadingRandom2018d,schusterOperatorGrowthOpen2023,zhangDynamicalTransitionOperator2023} fully encodes the thermalization dynamics of typical initial product states; (2) thermalization implies a lower bound on the growth of out-of-time-ordered correlators (OTOCs)~\cite{swingleMeasuringScramblingQuantum2016b,swingleUnscramblingPhysicsOutoftimeorder2018,xuScramblingDynamicsOutofTimeOrdered2024}; (3) establishing the absence of SSOs by generalizing the method in Refs.~\cite{\absenceIOMs} would suffice to rigorously demonstrate thermalization of most product initial states in certain classes of spin chain models; (4) in the context of Hilbert space fragmented systems~\cite{salaErgodicityBreakingArising2020,yangHilbertSpaceFragmentationStrict2020a,moudgalyaQuantumManybodyScars2022b,moudgalyaHilbertSpaceFragmentation2022b,balasubramanianGlassyWordProblems2024,chenQuantumFragmentationExtended2024,zhouQuantumHilbertSpace2026}, while strongly fragmented systems are guaranteed to be non-thermal, in prior work it was unclear whether such systems must possess IOMs; we have proven that SSOs are guaranteed to exist.

The rest of this work is organized as follows. In Section~\ref{sec:summary}, we give a broad technical overview of our methods and results. In Section~\ref{sec:ensemble-variance-norm}, we introduce the ensemble variance norm, establish its connection to operator size, and give the formal definition of simple operators. In Section~\ref{sec:ensemble-thermalization}, we demonstrate the connection between thermalization and operator growth, and prove the main result that the absence of SSOs implies thermalization. In Section~\ref{sec:slow-operators}, we introduce the $\tau$-$\nu$ plane language of characterizing SSOs and the numerical method to probe them, and present the landscape of SSOs in multiple physical contexts, including chaotic, integrable, and prethermal systems. In Section~\ref{sec:implications}, we present several implications of our results in the context of operator growth, quantum many-body scars, and Hilbert space fragmentation. Finally, in Section~\ref{sec:conclusion}, we summarize our results and discuss several open directions for investigation.

\section{Summary of results}
\label{sec:summary}

\paragraph{Definitions} --- We first define equilibration and thermalization for initial states evolving under unitary dynamics. Equilibration of an initial state means that
after evolving to sufficiently late times, any physical observable $O$ approaches a time-independent expectation value:  $\langle \psi(t) | O | \psi(t) \rangle \rightarrow \Oeq(\psi)$.
Thermalization is a further property that the late-time expectation value $\Oeq(\psi)$ coincides with the thermal expectation value $\Oth(\psi)$, which is usually calculated using either the microcanonical ensemble or the canonical ensemble~\cite{dalessioQuantumChaosEigenstate2016b}. 
To quantify the approach of $\langle \psi(t)|O|\psi(t)\rangle$ to $\Oth$, we define a time-averaged deviation~\cite{shortEquilibrationQuantumSystems2011,shortQuantumEquilibrationFinite2012,pilatowsky-cameoQuantumThermalizationMust2025},
\begin{multline}
D_f(\psi,O;\Oth) \\ = \int_0^{+\infty} \mathrm dt f(t) \left[\langle \psi(t)|O|\psi(t)\rangle-\Oth(\psi)\right]^2,
\end{multline}
where $f(t)$ is a non-negative envelope function that is normalized, $\int_0^{+\infty} f(t) \mathrm  dt=1$. Thermalization over a time window is equivalent to $D_f$ being vanishingly small for an envelope function $f$ that has support over that time window. We will use $D_T$ to denote the commonly used window average corresponding to the envelope function $\frac{1}{T}\Theta(T-t)$, and $D_\infty = \lim_{T\to\infty} D_T$ for the infinite-time average.

Our results pertain to thermalization of typical initial states drawn from a state ensemble. A quantum state ensemble $\mathcal{E}$ is a distribution over quantum states, $(p_\psi, |\psi\rangle)$, where $p_\psi$ is the probability of state $|\psi\rangle$. 
The expectation value over the state ensemble $\mathcal{E}$ is defined as
\begin{equation}
\expt_{\psi\sim\calE} g(\psi) = \sum_\psi p_\psi g(\psi).
\end{equation}
Therefore, we can characterize the thermalization of typical initial states drawn from the ensemble using the ensemble-average deviation,
\begin{equation}
D_f(\calE,O;\Oth) = \expt_{\psi \sim \calE} D_f(\psi,O;\Oth).
\end{equation}

\paragraph{Simple operators via ensemble variance norm} --- Given a state ensemble $\calE$, we define an ensemble variance norm of operator $A$ as follows:
\begin{equation}
\|A\|_{\calE} =  \sqrt{\expt_{\psi\sim\calE} |\langle \psi|A|\psi\rangle|^2}.
\end{equation}
Intuitively, this norm characterizes the typical magnitude of operator expectation values $\langle \psi|A|\psi\rangle$. If $\|A\|_\calE$ is small, it follows that for most states $|\psi\rangle$ drawn from $\calE$, $|\langle \psi|A|\psi\rangle|$ is likely to be small. Specializing to a many-body system of qubits, and choosing $\calE$ to be an ensemble of random product states (RPSs), we find that the norm is related to operator size: decomposing $A$ in the basis of Pauli strings as $A=\sum_P c_P P$, we obtain
\begin{equation}
\|A\|_\text{RPS}^2 = \sum_P |c_P|^2 3^{-|P|}, \label{eq:RPS-norm-intro}
\end{equation}
where $|P|$ is the weight of the Pauli string, defined as the number of non-identity operators in the Pauli string $P$. Noting that $\|A\|^2 = \sum_P |c_P|^2$ for the normalized Frobenius norm $\|A\| = \sqrt{\mathrm{tr}[A^\dagger A]/\mathrm{tr}[I]}$, we see that $\|A\|_\text{RPS}/\|A\|$ exponentially penalizes large operators. A similar phenomenon exists for several other low-entanglement ensembles, and the results can be generalized to qudit systems. In general, $\|A\|_\calE/\|A\|$ being large is equivalent to $A$ having significant small-sized components. We will thus refer to operators with large $\|A\|_\calE/\|A\|$ as ``simple operators''.

\paragraph{Thermalization from absence of SSOs} --- We now state the main results of this Article. As the precise formulation is rather technical, we begin with unpacking an informal statement to motivate the technical ingredients.

\setcounter{theorem}{-1}

\vspace{0.3cm}
\noindent
\fbox
{
\begin{minipage}{0.95\linewidth}
\begin{theorem}[Informal]
\label{theorem-informal}
In a given system, if there exists an observable $O$ that does not thermalize for typical initial states, \textit{i.e.}, $D_\infty(\calE,O;\Oth)$ is large, then the system must possess a non-trivial IOM or SGA that is simple.
\end{theorem}
\end{minipage}
}
\vspace{0.3cm}

The Theorem roughly follows the intuition we just outlined: as operator growth happens generally, $\|O(t)\|_\calE$ should become small at late times, unless it overlaps with some operator that is non-growing and has large ensemble variance norm (\textit{i.e.}, a simple operator). However, there are two gaps between this simplistic intuition and the desired thermalization Theorem. Firstly, $\|O(t)\|_\calE^2$ corresponds to $\langle \psi(t)|O|\psi(t)\rangle^2$, while thermalization requires bounding $\left[\langle \psi(t)|O|\psi(t)\rangle-\Oth(\psi)\right]^2$. Secondly, the conclusion of Theorem~\ref{theorem-informal} requires a ``non-trivial'' qualifier in front of ``simple IOM'', as the identity operator $I$ is a simple IOM, as is the Hamiltonian $H$, even in chaotic systems. In a meaningful statement ``non-thermalization implies existence of an IOM'', the ``IOM'' must exclude these trivial examples.

We find that the solutions to these two problems are actually related. To tackle the second problem, we project out a set of IOMs that we wish to rule out. This includes $I$ and $H$, as well as powers of $H$ up to a certain power $n$. After the projection, the statement ``there exists a simple IOM orthogonal (in the Hilbert-Schmidt inner product) to $\text{span}\{I,H,\dots,H^n\}$'' would no longer be trivial. Correspondingly, we decompose the physical observable $O$ as
\begin{equation}
O = O_\parallel^n + O_\perp^n,
\end{equation}
where $O_\parallel^n$ is the projection of $O$ onto $\text{span}\{I,H,\dots,H^n\}$, which commutes with $H$ by construction. We then apply our na\"ive statement to $O_\perp^n$ and get that $\|O_\perp^n(t)\|_\calE^2$ is small on average if there is no simple IOM or SGA in the projected subspace. Remarkably, this also solves the first problem we just identified: $\|O_\perp^n(t)\|_\calE^2$ is the average of $\left[\langle \psi(t)|O|\psi(t)\rangle-\langle \psi|O_\parallel^n|\psi\rangle\right]^2$. The quantity $\langle \psi|O_\parallel^n|\psi\rangle$ now plays a role similar to $\Oth$. It can be further shown that this  $\langle \psi|O_\parallel^n|\psi\rangle$ is the $n$-th order Taylor-expansion approximation to the thermal expectation value. This brings us to the formal Theorem.

\vspace{0.3cm}
\noindent
\fbox
{
\begin{minipage}{0.95\linewidth}
\begin{theorem}
\label{theorem-inf-time}
Choose an arbitrary observable $O$ and an integer $n\geq 0$, and define
\begin{equation}
\delta = \frac{D_{\infty}(\calE,O;\Onpow)}{\|O\|^2}, \label{eq:def-O-inf}
\end{equation}
where $\Onpow(\psi)=\langle \psi|O_\parallel^n|\psi\rangle$.
Then, there exists a non-zero operator $A$, such that
\begin{align}
[H,A] & =\lambda A \text{ for some }\lambda , \label{eq:therm-intro-inf-2} \\
\mathrm{tr}[AH^k]& = 0 \text{ for } 0 \leq k \leq n, \label{eq:therm-intro-inf-3} \\
\|A\|_\calE^2 & \geq \delta \|A\|^2.
\end{align}
\end{theorem}
\end{minipage}
}
\vspace{0.3cm}

We can see that the parameter $n$ plays two related roles: it rules out a trivial subspace of simple IOMs, and controls the precision with which $\Onpow$ approximates the thermal expectation value. Therefore, although the Theorem holds for any $n$, it only faithfully describes thermalization and the existence (absence) of IOMs when $n$ is taken large enough. The precise notion of ``large enough'' depends on the context and will be detailed in Section~\ref{sec:slow-operators} and Appendices~\ref{sec:n-dependence-nutau} and \ref{sec:appendix-thermal-value}.

With an adequate choice of $n$, the Theorem now states that if the typical initial state does not thermalize, then there must exist an IOM \textit{independent from the trivial ones constructed by taking powers of $H$}, or an SGA, that is simple. The contrapositive says that typical initial states must thermalize if a system has no simple non-trivial IOMs or SGAs. We note that IOMs and SGAs play distinct roles in this case: the absence of SGAs implies equilibration, and the absence of non-trivial IOMs implies that the equilibrium value approaches the thermal expectation value. We note there are physical arguments and numerical evidence (see Appendix~\ref{sec:sga}) indicating that a system possessing simple SGAs likely also possesses simple IOMs; therefore, a semi-formal version of the Theorem with Eq.~\eqref{eq:therm-intro-inf-2} replaced by $[H,A]=0$ would usually hold.

The full version of the Theorem, which generalizes Theorem~\ref{theorem-inf-time} to any finite timescale, reads as follows.

\vspace{0.3cm}
\noindent
\fbox
{
\begin{minipage}{0.95\linewidth}
\begin{theorem}
\label{theorem-finite-time}
Choose an arbitrary observable $O$, an integer $n\geq 0$, and a timescale $\tau>0$, and define
\begin{equation}
\delta = \frac{D_{f_\tau}(\calE,O;\Onpow)}{2\|O\|^2}, \label{eq:intro-therm-bound}
\end{equation}
where $f_\tau(t)= \frac{2}{\pi\tau}(\mathrm{sinc}(t/\tau))^2$. (We use the convention $\mathrm{sinc}(x)=\sin x/x$.)
Then, there exists a non-zero operator $A$, such that
\begin{align}
\|[H,A]-\lambda A\| & \leq \|A\|/\tau \text{ for some }\lambda , \label{eq:therm-intro-1} \\
\mathrm{tr}[AH^k]& = 0 \text{ for } 0 \leq k \leq n, \label{eq:therm-intro-2} \\
\|A\|_\calE^2 & \geq \delta \|A\|^2.
\end{align}
\end{theorem}
\end{minipage}
}
\vspace{0.3cm}

The function $f_\tau$ is roughly supported in the range $t\lesssim \tau$; therefore, $D_{f_\tau}(\calE,O;\Onpow)$ measures the average difference between the expectation value of $O$ and the thermal expectation value over a time window of $\tau$. The theorem thus states that non-thermalization of typical initial states on a timescale $\tau$ necessitates the existence of SSOs --- simple operators that are either non-trivial approximate IOMs or approximate SGAs up to a precision $1/\tau$.

\paragraph{Tradeoff between size and slowness}
--- Theorem~\ref{theorem-finite-time} can be cast in the other direction as follows: if all operators $A$ satisfying Eqs.~\eqref{eq:therm-intro-1} and \eqref{eq:therm-intro-2} have $\|A\|_\calE^2\leq \nu \|A\|^2$ for some $\nu$, then $D_{f_\tau}(\calE,O;\Onpow) \leq 2\nu \|O\|^2$ for all observables $O$. That is, Theorem~\ref{theorem-finite-time} provides a thermalization bound that applies uniformly to all observables, with the tightness controlled by a parameter $\nu_n(\tau)$, defined as the maximal value of $\|A\|_\calE^2/\|A\|^2$ among all operators $A$ that satisfy Eqs.~\eqref{eq:therm-intro-1} and \eqref{eq:therm-intro-2}. The notation $\nu_n(\tau)$ is to emphasize that this quantity depends on the choice of $n$ and $\tau$, while (unlike the $\delta$ in Eq.~\eqref{eq:intro-therm-bound}) it no longer depends on the choice of any specific observable $O$.

We illustrate the physical content of this $\nu_n(\tau)$ as follows. Any operator $A$ can be characterized by two parameters: its \emph{simplicity} $\nu_A=\|A\|_\calE^2 / \|A\|^2$, and its \textit{slowness}, characterized by $\tau_A=\|A\|/\|[H,A]\|$, the timescale on which it evolves. Note that compared to Theorem~\ref{theorem-finite-time}, we have set $\lambda = 0$ in defining $\tau$, \textit{i.e.}, focused our attention on approximate IOMs. All following arguments can be repeated by replacing $[H,A]$ with $[H,A]-\lambda A$ to cover approximate SGAs, although the $\lambda=0$ case carries most physical meaning, as will be demonstrated in Appendix~\ref{sec:sga}.

As $A$ varies, the representative point $(\tau_A,\nu_A)$ moves in a $\tau$-$\nu$ plane. The $\nu_n(\tau)$ function is exactly the upper frontier of the region in which operators orthogonal to $\text{span}\{I,H,\dots,H^n\}$ live. That is, $\nu_n(\tau)$ asks: \textit{how simple is the simplest operator among all operators that are approximately conserved on timescale $\tau$ and orthogonal to the given subspace spanned by powers of $H$?} As $\tau$ increases, the space of attainable operators becomes smaller due to a more stringent constraint on $\|[H,A]\|$; consequently, the simplest operator available becomes less simple, so the function $\nu_n(\tau)$ is monotonically decreasing.
Assuming that $n$ is taken sufficiently large (below we will sometimes omit $n$ from $\nu(\tau)$, with such a choice implicitly understood), the shape of this frontier curve would measure whether SSOs exist on a timescale $\tau$, and is directly related to the thermalization dynamics of the system. For chaotic systems, $\nu(\tau)$ would decay to very small $\nu$ values as $\tau$ increases, and $\nu(\tau=\infty)$ becomes vanishingly small in large systems (provided that $n$ is chosen appropriately), implying exact thermalization in the thermodynamic limit; for integrable systems, simple IOMs exist, thus $\nu(\tau)$ would feature a plateau with $\nu=\Theta(1)$ up to $\tau\to\infty$ (the notation $\Theta(f)$ means on the same order as $f$, while $\calO(f)$ means no larger than the order of $f$); for systems with slow (but not exactly conserved) modes, $\nu(\tau)$ would plateau up to a certain time and then descend again, with the timescale of thermalization encoded in this descent.

Determining this $\nu_n(\tau)$ frontier is equivalent to solving the following optimization problem,
\begin{equation}
\begin{array}{lll}
 \nu_n(\tau) = &\text{maximize} \quad &\|
A\|_\calE^2, \\
&\text{subject to} & \|A\|^2 = 1,  \\
&& \|[H,A]\|^2 \leq 1/\tau^2, \\
&& \mathrm{tr}[AH^k] = 0 \text{ for } 0 \leq k \leq n.
\end{array}
\end{equation}
This is a quadratic optimization problem that can be solved by recasting it as a generalized eigenvalue problem in operator space (see Appendix~\ref{sec:jnr} for details). This allows us to determine $\nu_n(\tau)$ exactly in systems where exact diagonalization (ED) of the Hamiltonian is possible. In larger systems, it is possible to approximate the optimized frontier by using variational ansatzes. In particular, we find that using an ansatz of slow operators constructed by applying a Lorentzian spectral filter to a physically chosen simple operator can often closely approach the frontier. We present the numerical $\nu_n(\tau)$ frontier, exact and variational, in multiple physical contexts in Section~\ref{sec:slow-operators} and Appendix~\ref{sec:numerics}. The code for numerically computing the frontier has been made open-source~\footnote{The code and data producing the figures in this Article are available at \href{https://github.com/ThomasYangth/SSO-Code}{https://github.com/ThomasYangth/SSO-Code}.}.

\tableofcontents

\section{The ensemble variance norm} \label{sec:ensemble-variance-norm}

Denote a quantum state ensemble as $\calE$. The ensemble would contain a set of states $|\psi\rangle$ with an associated probability distribution $p_\psi$. The expectation value over the ensemble is given by
\begin{equation}
\expt_{\psi\sim\calE} g(\psi) = \sum_\psi p_\psi g(\psi).
\end{equation}
The probability distribution $p_\psi$ can be discrete or continuous, and should be normalized to $\sum_\psi p_\psi=1$ in both cases. Traditionally, all the states in an ensemble are required to be normalized, with $\langle \psi|\psi\rangle=1$. An ensemble satisfying this condition is called a normalized ensemble. It is also possible to study unnormalized ensembles~\cite{markMaximumEntropyPrinciple2024a}, where this condition is relaxed.

The expectation value of an operator $A$ with respect to the ensemble is given by
\begin{equation}
\langle A\rangle_\calE = \expt_{\psi \sim \calE} \langle \psi |A|\psi\rangle = \tr (A\rho_\calE),
\end{equation}
where the density matrix, or first moment, of the ensemble is defined as
\begin{equation}
\rho_\calE = \expt_{\psi\sim\calE} |\psi\rangle \langle \psi|.
\end{equation}
We will use notation $\langle A\rangle$ without a subscript to denote normalized trace $\frac{1}{D}\tr A$, where $D=\mathrm{tr}I$ is the dimension of the Hilbert space. We define the following inner product on the space of operators,
\begin{equation}
\langle A,B \rangle_{\calE} = \expt_{\psi\sim\calE} \langle \psi|A^\dagger|\psi\rangle \langle \psi|B|\psi\rangle. \label{eq:cov-inner-product}
\end{equation}
We call this the ensemble covariance inner product for $\calE$, or the $\calE$-inner product. (We are abusing terminology here as the definition of covariances would be $\mathrm{cov}(A,B)=\expt [AB]-\expt[A]\expt [B]$, so our ``covariance inner product'' would be a true covariance only when $\langle A\rangle_\calE=\langle B\rangle_\calE=0$. We choose not to subtract the mean for simplicity, and it is often the case that we do consider operators where such mean is zero.) It is easy to verify that this inner product is bilinear and positive semidefinite. For Hermitian operators $A$ and $B$, this inner product is real symmetric. The inner product induces a norm,
\begin{equation}
\|A\|_{\calE} = \sqrt{\langle A,A\rangle_\calE} = \sqrt{\expt_{\psi\sim\calE} |\langle \psi|A|\psi\rangle|^2}.
\end{equation}
We call this the ensemble variance norm with respect to $\calE$, or the $\calE$-norm for short. This norm measures the fluctuation of the expectation value of $A$ within the ensemble $\calE$, with the concentration bound
\begin{equation}
\prob\left[|\langle \psi|A|\psi\rangle|>\epsilon\right] \leq \frac{\|A\|_\calE^2}{\epsilon^2}.
\end{equation}
Therefore, if $\|A\|_\calE^2$ is small, $\langle \psi|A|\psi\rangle$ would be concentrated around zero for most states in $\calE$.

We will still use $\langle A,B\rangle$ without subscripts to denote the normalized Hilbert-Schmidt inner product $\langle A,B\rangle =\frac{1}{D}\tr[A^\dagger B]=\langle A^\dagger B\rangle$, and $\|A\|=\sqrt{\langle A,A\rangle}$ denotes the normalized Frobenius norm. For each ensemble, we also define the kernel superoperator,
\begin{equation}
\calM_\calE[B] = D \expt_{\psi\sim\calE} \langle \psi|B|\psi\rangle |\psi\rangle\langle \psi|, \label{eq:evn-kernel-def}
\end{equation}
such that
\begin{equation}
\langle A,B\rangle_\calE = \langle A^\dagger \calM_\calE[B]\rangle = \langle A,\calM_\calE[B]\rangle.
\end{equation}
$\calM_\calE$ is always a Hermitian positive semidefinite superoperator under the Hilbert-Schmidt inner product. It can be seen as a representation of the second moment of a quantum state ensemble, an object that has been studied in several contexts within quantum information theory~\cite{divincenzoQuantumDataHiding2002,grossEvenlyDistributedUnitaries2007,tothExtremalPropertiesVariance2013,yuQuantumFisherInformation2013,robertsChaosComplexityDesign2017,hoExactEmergentQuantum2022a,cotlerEmergentQuantumState2023,ferroKickingQuantumFisher2025,fagottiQuantumCoherentThermodynamicsLeaf2026}. Specifically, $\calM_\calE$ plays a central role in quantum shadow tomography~\cite{huangPredictingManyProperties2020b}, where it was shown that the superoperator has closed-form expressions in many cases. Therefore, for a wide range of state ensembles, the covariance inner products and variance norms have simple and meaningful expressions.

As a general setup, we consider quantum systems with $L$ sites, where the sites are qudits with local Hilbert space dimension $d$, and the total Hilbert space dimension is $D=d^L$. We choose a set of Hermitian operators $(\Lambda_i)_{i=1}^{d^2}$ that form a basis of the local operator space, with $\Lambda_1 = I$ being the identity matrix, and the set satisfying orthonormality,
\begin{equation}
\langle \Lambda_i,\Lambda_j\rangle = \frac{1}{d} \mathrm{tr} \Lambda_i \Lambda_j = \delta_{ij}.
\end{equation}
When $d=2$, a natural choice is the set of Pauli matrices. We will abuse terminology to still refer to $\Lambda_i$ as ``Paulis'', and operators such as $\Lambda_{i_1}\otimes \Lambda_{i_2}\otimes \dots \otimes \Lambda_{i_L}$ as ``Pauli strings'' for all $d$. It is easily seen that such Pauli strings form a basis of the operator space on the $D$-dimensional Hilbert space. We will use the symbol $P$ to denote an arbitrary Pauli string. We shall emphasize that although we use terminology such as $L$ and ``string'', our results do not assume any particular geometry of the lattice unless explicitly stated.

\subsection{Haar random states}

The ensemble covariance inner product can be straightforwardly evaluated for Haar random states~\cite{collinsIntegrationRespectHaar2006}, yielding
\begin{equation}
\langle A,B\rangle_\text{Haar} = \frac{1}{D(D+1)}\left(\tr A^\dagger \tr B+\tr \left(A^\dagger B\right)\right).
\end{equation}
Assuming that $A$ is Hermitian and traceless, we have
\begin{equation}
\| A \|_\text{Haar} = \sqrt{\frac{1}{D(D+1)}\tr A^2} = \sqrt{\frac{1}{D+1}\langle A^2\rangle}.
\end{equation}
If $\langle A^2\rangle=\Theta (1)$, which is the case for common observables such as Pauli strings, $\| A \|_\text{Haar}$ would be of the order $D^{-1/2}=d^{-L/2}$, which is exponentially small in system size. This is the well-established phenomenon that the high dimensionality of the Hilbert space leads to an exponential concentration of expectation values, which is a corollary of L\'evy's lemma~\cite{popescuEntanglementFoundationsStatistical2006}, and is responsible for a wide range of phenomena in many-body quantum systems, such as barren plateaus~\cite{mccleanBarrenPlateausQuantum2018}, dynamical quantum typicality~\cite{bartschDynamicalTypicalityQuantum2009,reimannDynamicalTypicalityIsolated2018,heitmannSelectedApplicationsTypicality2020}, and deep thermalization~\cite{hoExactEmergentQuantum2022a,cotlerEmergentQuantumState2023}.

\subsection{Random product states}

A random product state (RPS) is a direct product of independent Haar random states on each site. On one site, we have that
\begin{align}
\langle \Lambda_i,\Lambda_j\rangle_\text{Haar} &
= \frac{1}{d(d+1)}(\mathrm{tr}\Lambda_i\mathrm{tr}\Lambda_j+\mathrm{tr}\Lambda_i\Lambda_j) \nonumber \\
& = \frac{1}{d+1}(\delta_{i,1}\delta_{j,1}d+\delta_{i,j}) \nonumber \\
& = \delta_{i,j} (d+1)^{-S(i)}.
\end{align}
We have introduced the notation $S(i)= \begin{cases} 0, & i =1 \\ 1, & i > 1 \end{cases}$. This shows that the Haar inner product is diagonal for this basis; the diagonal element is $1$ for the identity matrix, and $(d+1)^{-1}$ for all others.  Now, taking the direct product of $L$ copies of this, we would get the random product state inner product on a system of $L$ sites:
\begin{multline}
\langle \Lambda_{i_1}\otimes \Lambda_{i_2}\otimes \dots \otimes \Lambda_{i_L}, \Lambda_{j_1}\otimes \dots \otimes \Lambda_{j_L} \rangle_\text{RPS} \\
= \delta_{i_1j_1}\dots\delta_{i_Lj_L} (d+1)^{-\sum_k S(i_k)}.
\end{multline}
Let us define a superoperator $\calS$ as $\calS (\Lambda_{i_1}\otimes \dots \otimes \Lambda_{i_L}) = \left(\sum_k S(i_k)\right) \Lambda_{i_1}\otimes \dots \otimes \Lambda_{i_L}$. This is a superoperator that is diagonal in the Pauli string basis of the operator space, whose matrix element counts the number of non-identity operators in a string. In compact notation, $\calS(P) = |P| P$, where $P$ denotes an arbitrary Pauli string, and $|P|=\sum_k S(i_k)$ is the weight (or size) of the string. Then, we get the following form for the random product state inner product:
\begin{equation}
\langle A,B\rangle_\text{RPS} = \left\langle A, (d+1)^{-\mathcal S} [B] \right\rangle.
\end{equation}
If we expand operators under the Pauli string basis, $A=\sum_P a_P P$ and $B=\sum_P b_P P$, then
\begin{equation}
\langle A,B\rangle_\text{RPS} = \sum_P (d+1)^{-|P|} a_P^\ast b_P.
\end{equation}
This expression has been known in Refs.~\cite{qiMeasuringOperatorSize2019,Ippoliti2023OperatorRelaxation}. Equivalently, the associated kernel superoperator has the form $\calM_\text{RPS}:A\mapsto \sum_P (d+1)^{-|P|} a_P P$. The form of this superoperator was used in deriving the shadow norm of local Pauli rotations~\cite{huangPredictingManyProperties2020b}.

The RPS-norm penalizes high-weight Pauli strings exponentially: if a Pauli string $P$ has size $S$, then $\|P\|_\text{RPS} = (d+1)^{-S/2}$. Therefore, an operator $O$ having $\Theta(L)$ size would have $\|O\|_\text{RPS} \sim e^{-\Theta(L)}$. If we draw an operator from the entire operator space at random (specifically, choose the coefficients before each Pauli string at random), then the average RPS-norm is
\begin{align}
\overline{\|A\|_\text{RPS}^2} & = \sum_P (d+1)^{-|P|}\overline{|a_P|^2} \nonumber \\
& = \frac{1}{D^2}\sum_{S=0}^L \binom{L}{S} (d^2-1)^S (d+1)^{-S} \nonumber \\
& = \frac{1}{D} = d^{-L}. \label{eq:average-RPS-norm}
\end{align}
This is exponentially small in $L$, which means that operators with $\|O\|_\text{RPS}$ not scaling as $e^{-\Theta(L)}$ are rare in the operator space.

The RPS ensemble is a continuous ensemble, which means that the probability of any one particular state appearing is zero. In practice, it is sometimes preferable to have finite state ensembles. As the RPS ensemble is a Cartesian product of Haar ensembles on each site, if a finite ensemble on a single site gives the same inner product as the Haar ensemble, taking a Cartesian product of copies of this finite ensemble would produce a finite ensemble equivalent to the RPS. In quantum information language, ensembles that give the same inner product as the Haar ensemble are called 2-designs~\cite{ambainisQuantumTdesignsTwise2007a,cotlerEmergentQuantumState2023}. Commonly used finite 2-designs include stabilizer states (when $d$ is a prime power)~\cite{gottesmanHeisenbergRepresentationQuantum1998,koenigHowEfficientlySelect2014,webbCliffordGroupForms2016,zhuMultiqubitCliffordGroups2017b,graydonCliffordGroupsAre2021} and the SIC-POVM set~\cite{renesSymmetricInformationallyComplete2004,zhuSICPOVMsClifford2010,fuchsSICQuestionHistory2017}. In the qubit case $d=2$, the set $\calE_{\text{stab},d=2}=\{|\pm x\rangle ,|\pm y\rangle ,|\pm z\rangle\}$ is a 2-design, so that $(\calE_{\text{stab},d=2})^{\otimes L}$, or the set of stabilizer product states, would give the same ensemble variance norm as the RPS ensemble. We will generally not distinguish between the RPS ensemble and its finite equivalents.

\subsection{Random matrix product states}

For the case of 1D chains, we can generalize RPS to random matrix product states (RMPS) with a fixed bond dimension $\chi$. For the RMPS ensemble, we can also derive an inner product, whose kernel has the form of a matrix product operator (MPO). The full expression is derived in Appendix~\ref{sec:RMPS}. Notably, the ensemble covariance inner product is also diagonal in the Pauli string basis. However, different from the RPS case, the ensemble variance norm is not only determined by the weight of the Pauli $|P|$, but would also depend on the positions of the non-identity operators. Generally, the ensemble variance norm in a large system satisfies
\begin{equation}
(d\chi^2+1)^{-|P|} \leq \|P\|_{\chi\text{-RMPS}}^2\leq \left(\frac{d \chi^2-1}{d^2 \chi^2-1}\right)^{|P|},
\end{equation}
where the lower bound is asymptotically obtained by letting the non-identity operators be as separated as possible, and the upper bound is asymptotically achieved by putting all non-identity operators next to each other. Therefore, on top of penalizing operator size like the RPS-norm does, the RMPS-norm further penalizes ``long-range correlator''-like operators.

\subsection{Thermal dressing}
\label{subsec:thermal-dressing}

Both the RPS and RMPS ensembles have $\rho_\calE \propto I$. Therefore, although there can be states in the ensembles at every energy, the ensemble as a whole should be considered infinite-temperature. To construct a finite-temperature ensemble, consider the following thermal dressing,
\begin{equation}
|\psi\rangle \mapsto \frac{e^{-\beta H/2}}{\sqrt{\langle e^{-\beta H} \rangle}} |\psi\rangle.
\end{equation}
For a given ensemble $\calE$, we define $\widetilde{\calE_\beta}$ as the ensemble obtained by dressing every state in $\calE$ in this fashion. Note that the dressed state is, in general, not normalized. Therefore, $\widetilde{\calE_\beta}$ are unnormalized ensembles. This construction has been used in minimally entangled typical thermal states (METTS)~\cite{whiteMinimallyEntangledTypical2009,stoudenmireMinimallyEntangledTypical2010} and the Scrooge ensemble~\cite{parfionovLazyQuantumEnsembles2006,cotlerEmergentQuantumState2023,markMaximumEntropyPrinciple2024a,liuDeepThermalizationGaussian2024,mcginleyScroogeEnsembleManybody2025,mannaProjectedEnsembleSystem2025,mokNatureStingyUniversality2026}. We discuss more properties of such ensembles in Appendix~\ref{sec:scrooge}. This dressing corresponds to the following dressing of operators,
\begin{equation}
O \mapsto O^\beta = \frac{1}{\langle e^{-\beta H} \rangle} e^{-\frac{1}{2}\beta H} O e^{-\frac{1}{2}\beta H}, \label{eq:thermal-dressing-of-operator}
\end{equation}
which is a common construction in studies of operator dynamics~\cite{parkerUniversalOperatorGrowth2019b,xuScramblingDynamicsOutofTimeOrdered2024}. Then,
\begin{equation}
\rho_{\widetilde{\calE_\beta}} = (\rho_\calE)^\beta,
\end{equation}
and
\begin{equation}
\langle A,B\rangle_{\widetilde{\calE_\beta}} = \langle A^\beta,B^\beta\rangle_{\calE}.
\end{equation}
Note that $I^\beta$ is exactly $D$ times the Gibbs density matrix. Therefore, if $\calE$ has $\rho_\calE = I/D$, then $\rho_{\widetilde{\calE_\beta}}$ would be the Gibbs density matrix at inverse temperature $\beta$. We can therefore construct the thermally dressed versions of RPS and RMPS ensembles. For high temperatures, $e^{-\beta H/2}$ can be Trotterized into a shallow quantum circuit, thus the resulting dressed ensembles are still weakly entangled.

\subsection{Characterization of simple operators}
\label{subsec:simple-operators-characterization}

We call an operator simple if it has a large ensemble variance norm (relative to its Frobenius norm). It is apparent that the class of simple operators would vary based on the ensemble we choose to define the ensemble variance norm. We will discuss the physical characteristics of simple operators based on the RPS ensemble. We expect that many of the properties generalize to other low-entanglement ensembles as well.

Recall that the RPS-norm of an operator decomposed in the Pauli basis as $O=\sum_P c_P P$ is $\|O\|_\text{RPS}^2 = \sum_P |c_P|^2 (d+1)^{-|P|}$. Simple operators as such are obviously not required to be strictly local or $k$-local, and can have support on arbitrarily large Pauli strings. Importantly, these operators need not be quasilocal either (although quasilocal operators are automatically simple). For quasilocal operators, the coefficients of large-support strings are exponentially small; however, simple operators need not obey this constraint. Since the RPS-norm is dominated by the contribution from small $|P|$, the sum of a quasilocal operator and a very long Pauli string would still have a large ensemble variance norm, though it is not quasilocal.

It is also illuminating to compare this with another definition of operator size,
\begin{equation}
O=\sum_P c_P P \mapsto \|O\|_{\calS}^2 = \sum_P |c_P|^2 \cdot |P|.
\end{equation}
This is exactly a weighted average of the weights of the Pauli strings contained in $O$. Such expression appears in the context of depolarization noise~\cite{schusterOperatorGrowthOpen2023,yangIntegralsMotionSlow2025} and OTOCs~\cite{swingleMeasuringScramblingQuantum2016b}. Arguably, $\|\cdot \|_{\calS}$ is more intuitive than $\|\cdot \|_\text{RPS}$: it is dominated by large Pauli strings, instead of small ones. In this case, the sum of a simple operator --- defined as one with small $\|\cdot \|_\calS$ --- and a large Pauli string would no longer be simple. That is, simplicity in terms of $\|\cdot \|_\calS$ measures the absence of large Paulis, while that in terms of $\|\cdot \|_\text{RPS}$ measures the presence of small Paulis. In fact, the class of simple operators in terms of $\|\cdot \|_\text{RPS}$ is provably the largest among categories of operators discussed here. It is intuitive that local, $k$-local and quasilocal operators all have small $\|\cdot \|_\calS$, while
\begin{equation}
\|O\|_\text{RPS}^2 \geq (d+1)^{-\|O\|_\calS^2} 
\end{equation}
due to Jensen's inequality (if $\|O\|=1$). Therefore, every operator with a small $\|O\|_\calS$ must have large $\|O\|_\text{RPS}$, but not vice versa.

We argue that these apparently counterintuitive features of $\|\cdot\|_\text{RPS}$ are essential to describing thermalization. On one hand, assume that an IOM exists such that it is the sum of some low-weight Pauli strings and some high-weight Pauli strings. While traditional measures would disqualify this IOM from being quasilocal or $k$-local, this IOM can certainly prevent thermalization, as the low-weight part of it overlaps with local observables, and these observables would fail to thermalize due to the overlap. Indeed, a recently identified model~\cite{tuAsymptoticallySolvableModel2026} features non-thermalization attributed to this kind of non-quasilocal IOM. Therefore, in any statement like ``systems without simple IOMs thermalize'', it is necessary to make the class of ``simple'' operators broad enough to include this kind of small-plus-large operators. Secondly, we show that there are mathematical conditions on what kind of operator norms can be expressed as an ensemble variance. We prescribe and prove such conditions in Appendix~\ref{sec:gen-op-norm-as-evn}, where it is shown that $\|\cdot\|_\calS$ cannot be expressed as an ensemble variance norm for any pure-state ensemble. These indicate that the properties of $\|\cdot \|_\text{RPS}$ are innate to the problem of ensemble thermalization.

\section{Thermalization bounds}
\label{sec:ensemble-thermalization}

\subsection{Thermalization from operator growth} \label{subsec:thermalization-from-operator-growth}

With the construction above, we can relate the typical thermalization of states in an ensemble to the dynamics of operators. Consider an observable $O$, and assume that the dynamics is generated by a time-independent Hamiltonian $H$. (This formalism can be naturally generalized to the case where the dynamics is given by a Floquet unitary $U$, which we detail in Appendix~\ref{sec:floquet}.) The evolution of operators in Heisenberg picture is
\begin{equation}
O(t) = e^{iHt} O e^{-iHt}.
\end{equation}
To start, assume that the thermal expectation value of $O$ is zero. In this case, saying that most states in an ensemble $\calE$ thermalize is equivalent to saying
\begin{align}
D_f(\calE,O;0) & = \expt_{\psi\sim \calE} D_f(\psi,O;0) \nonumber \\
& = \int_0^{+\infty} f(t) \expt_{\psi\sim \calE}\left[\langle \psi|O(t)|\psi\rangle^2\right] \mathrm dt  \nonumber \\
& = \int_0^{+\infty} f(t) \|O(t)\|_\calE^2 \mathrm dt 
\end{align}
is small. Since $\|O(t)\|_\calE$ is roughly $\exp (-\text{size of }O(t))$, this establishes a direct link between operator growth and thermalization: if $O(t)$ has a large operator size at late times, then the observable $O$ would thermalize for most initial states in the ensemble.

Generally, $\langle \psi|O(t)|\psi\rangle$ would equilibrate not to zero, but to some finite thermal expectation value. On the other side of the equation, $\|O(t)\|_\calE$ may not relax to zero either, as $O$ may overlap with the energy operator $H$, and the overlapping part is conserved in the evolution, meaning that there is a small-sized part of $O(t)$ that survives the evolution. To formulate such generic cases, consider choosing a set of conserved operators $\{Q_i\}$ that capture the non-growing part of $O$. Take the set to be orthonormal under the Hilbert-Schmidt inner product. Then, any operator admits a decomposition,
\begin{align}
O & = \calP[O] + \calQ[O], \\
\calQ[O] & = \sum_i \langle Q_i,O\rangle Q_i, \\
\calP[O] & = O - \sum_i \langle Q_i,O\rangle Q_i \label{eq:calP-definition}.
\end{align}
By definition, $\calP[O(t)] = \calP[O](t)$, and $\calQ[O](t)=\calQ[O(t)]=\calQ[O]$. $\calP[O]$ characterizes the part of $O$ that is not protected by conservation laws. Therefore, if $\{Q_i\}$ is chosen to include enough simple conserved operators, we can reasonably expect $\|\calP[O(t)]\|_\calE$ to be small at late times. In that case, we know that $|\langle \psi|O(t)|\psi\rangle - \langle \psi|\calQ[O]|\psi\rangle|$ is small on average, which means that for most initial states $|\psi\rangle \sim \calE$, we have
\begin{equation}
\langle \psi|O(t)|\psi\rangle  \approx \langle \psi | \calQ[O]|\psi\rangle = \sum_i \langle Q_i,O\rangle \langle \psi|Q_i|\psi\rangle. \label{eq:slow-mode-part-expectation}
\end{equation}
This is the corresponding statement in the generic case: if the part of $O$ orthogonal to known conserved operators grows to large operator size at late times, then the expectation value of $O(t)$ would be given by the projection of $O$ onto the conserved operators. The most natural choice of $\{Q_i\}$ is (a Gram-Schmidt orthonormalization of) the first few powers of the energy operator $\{H^k\}_{k=0}^n$. We specifically denote the projectors as $\calP_n$ and $\calQ_n$ in this case, and $(\calQ_n[O],\calP_n[O])$ would be the $(O_\parallel^n,O_\perp^n)$ defined in Section~\ref{sec:summary}. In particular, $\calQ_n[O]$ would give an approximation to the thermal expectation value, as we will show in the next Subsection. The derivation should generalize naturally to the case where the equilibrium is characterized by the generalized Gibbs ensemble, in which case $\{Q_i\}$ should include low-order polynomials of $H$ and other conserved charges.

\subsection{Powers of $H$ and the thermal expectation value}
\label{sec:thermal-expectation-value}

We now argue that when $\{Q_i\}$ is chosen to be (the orthonormalization of) $\{H^k\}_{k=0}^n$, $O_\parallel^n=\calQ_n[O]$ is an approximation to the thermal expectation value. We will give a semi-quantitative argument showing that $\Onpow(\psi)=\langle \psi|O_\parallel^n|\psi\rangle$ approximates the thermal expectation value. A more thorough discussion is presented in Appendix~\ref{sec:appendix-thermal-value}. 

Given a state $|\psi\rangle$ with energy $E_\psi = \langle \psi|H|\psi\rangle$, the thermal expectation value can be taken as the microcanonical-ensemble expectation value
\begin{equation}
\langle O\rangle_\text{MC}(\psi) = \frac{1}{N_{E_\psi}} \sum_{|E-E_\psi|<\epsilon} \langle E|O|E\rangle,
\end{equation}
where $|E\rangle$ are the eigenstates of $H$, and $N_{E_\psi}$ counts the number of eigenvalues in a shell around $E_\psi$ with width $\epsilon$. Alternatively, one could also use the Gibbs density matrix,
\begin{equation}
\langle O\rangle_\text{Gibbs}(\psi) = \langle O\rangle_{\beta_\psi},
\end{equation}
where the Gibbs expectation value at a given inverse temperature is defined as
\begin{equation}
\quad \langle O\rangle_\beta = \frac{\tr[e^{-\beta H}O]}{\tr[e^{-\beta H}]},
\end{equation}
and the inverse temperature of a specific state, $\beta_\psi$, is determined by $\langle H\rangle_{\beta_\psi}=E_\psi$. These two constructions are equivalent in the thermodynamic limit~\cite{brandaoEquivalenceStatisticalMechanical2015,tasakiLocalEquivalenceCanonical2018} up to $\calO(L^{-1/2})$ for local operators $O$.

Without assuming anything about thermalization, the long-time expectation value of an operator will always be given by the diagonal ensemble, 
\begin{equation}
O_\text{diag} = \sum_E \langle E|O|E\rangle |E\rangle \langle E|,
\end{equation}
where $|E\rangle$ denote eigenstates of $H$ satisfying $H|E\rangle=E|E\rangle$. For simplicity, we have written this assuming that the spectrum of $H$ is non-degenerate; however, this is not an essential assumption, and the argument can be naturally generalized to cases where the eigenvalues are degenerate. The microcanonical expectation value can be written as the following operator,
\begin{equation}
O_\text{MC} = \sum_E \left(\frac{1}{N_{E}} \sum_{|E^\prime-E|<\epsilon} \langle E^\prime|O|E^\prime\rangle\right) |E\rangle \langle E|.
\end{equation}
Both operators are diagonal in the energy eigenbasis, and have the form $O=\sum_E f(E) |E\rangle \langle E|$. While $f_\text{diag}(E)=\langle E|O|E\rangle$ could be any function in principle, $f_\text{MC}(E)$ is a sliding-window average of $f_\text{diag}$, hence it will always be smooth. The eigenstate thermalization hypothesis (ETH)~\cite{deutschQuantumStatisticalMechanics1991a,deutschEigenstateThermalizationHypothesis2018a,dalessioQuantumChaosEigenstate2016b} can be naturally understood in this context: the (diagonal) ETH assumes that $f_\text{diag}(E)$ is a smooth function, implying that its smoothing is equal to itself, hence $O_\text{MC}=O_\text{diag}$, and it follows that the late-time expectation value of $O$ is given by the microcanonical ensemble.

In our construction, $O_\parallel^n=\sum_{k=0}^n c_k H^k$, which gives the diagonal entries $f_{\npow}(E)=\sum_{k=0}^n c_k E^k$. This is an order-$n$ polynomial. Moreover, as the coefficients $c_k$ are obtained via projecting $O$ onto $H^k$, $f_{\npow}(E)$ would be a ``polynomial fitting'' of $f_\text{diag}(E)$; therefore, it should be similar to the smoothing of $f_\text{diag}(E)$, which is $f_\text{MC}(E)$. This can be made more rigorous following a Taylor-expansion argument as in Ref.~\cite{wangEigenstateThermalizationHypothesis2025}. By definition,
\begin{equation}
\left\langle H^\ell, O_\parallel^n - O \right\rangle = 0, \, \forall 0\leq \ell \leq n.
\end{equation}
In line with Ref.~\cite{wangEigenstateThermalizationHypothesis2022a}, this implies
\begin{equation}
\left.\frac{\partial^\ell}{\partial \beta^\ell} \frac{\tr[e^{-\beta H}(O_\parallel^n-O)]}{\tr[e^{-\beta H}]}\right|_{\beta=0} = 0,  \,\forall 0\leq \ell \leq n. \label{eq:cluster-poly-equal}
\end{equation}
Equivalently, $\langle O\rangle_\beta$ and $\langle O_\parallel^n\rangle_\beta$ are equal near $\beta=0$ up to order $\calO(|\beta|^{n+1})$. This implies that $f_{\npow}(E)$ and $f_\text{MC}(E)$ are similar near $E=0$ (which is the energy corresponding to $\beta=0$, assuming $H$ is traceless), with the error being $\calO(|E/L|^{n+1})$. Therefore, the value $\Onpow(\psi)=\langle \psi|O_\parallel^n|\psi\rangle$ serves as a good approximation to the thermal expectation value for large enough $n$. More detailed discussions of the accuracy of this approximation are given in Appendix~\ref{sec:appendix-thermal-value}.

\subsection{Thermalization from absence of SSOs} \label{subsec:thermalization-and-slow-operators}

We are now equipped to establish our main result: thermalization of most initial states in an ensemble $\calE$ is implied by the absence of SSOs, where ``simple'' is defined with respect to that ensemble.
To simplify the notation, we will consider an operator $\tilde O$ that is already orthogonal against all known conserved charges $\{Q_i\}$. We expand the operator $\tilde O$ in the eigenbasis of the Liouvillian $[H,\cdot]$ as
\begin{equation}
\tilde O(t) = \sum_\lambda \tilde O_\lambda e^{i\lambda t}, \label{eq:O-eigen-decomposition}
\end{equation}
where each $\tilde O_\lambda$ satisfies $[H,\tilde O_\lambda]=\lambda \tilde O_\lambda$. Since $[H,\cdot]$ is a Hermitian superoperator under the Hilbert-Schmidt inner product, the set $\{\tilde O_\lambda\}$ is orthogonal, and $\|\tilde O(t)\|^2 = \|\tilde O\|^2 = \sum_\lambda \|\tilde O_\lambda\|^2$. Note that we \textit{do not} assume the eigenvalues of $[H,\cdot]$ are non-degenerate.

The time-averaged $\calE$-norm can be calculated as
\begin{multline}
\int_0^{+\infty} f(t)\mathrm dt \|\tilde O(t)\|_{\calE}^2 \\
=  \sum_\lambda \|\tilde O_\lambda\|_{\calE}^2
 + \sum_{\lambda \neq \lambda^\prime} \int_0^{+\infty} f(t)\mathrm dt
e^{i(\lambda-\lambda^\prime)t}
\langle \tilde O_{\lambda^\prime},\tilde O_\lambda \rangle_{\calE}.
\end{multline}
In the infinite-time limit (\textit{i.e.}, take $f(t)=\Theta(T-t)/T$, then send $T\to\infty$), the second term vanishes.
Therefore,
\begin{align}
\lim_{T\to\infty} & \frac{1}T\int_0^{T} \|\tilde O(t)\|_{\calE}^2 \mathrm dt \nonumber \\
 & = \sum_{\lambda}\|\tilde O_\lambda\|_{\calE}^2 \nonumber \\
 & = \sum_{\lambda} \|\tilde O_\lambda\|^2 \left\| \frac{\tilde O_\lambda}{\|\tilde O_\lambda\|} \right\|_{\calE}^2 \nonumber \\
 & \leq \|\tilde O\|^2 \max_\lambda \left\| \frac{\tilde O_\lambda}{\|\tilde O_\lambda\|} \right\|_{\calE}^2 \label{eq:infinite-time-thermalization-bound-derivation},
\end{align}
where in the third line we multiplied and divided each term by $\|\tilde O_\lambda\|^2$. Therefore, if the left-hand side of this inequality---namely, the time-averaged deviation---is equal to $\delta \|\tilde O\|^2$, it must be the case that $\max_\lambda \left\| \frac{\tilde O_\lambda}{\|\tilde O_\lambda\|} \right\|_{\calE}^2 \geq \delta$. Suppose the maximum is achieved for $\lambda = \lambda^*$. Then, we arrive at the following result: if
\begin{equation}
\frac{1}{\|\tilde O\|^2}\lim_{T\to\infty} \frac{1}T\int_0^{T} \|\tilde O(t)\|_{\calE}^2 \mathrm dt = \delta, \label{eq:inf-therm-bound-in-derivation}
\end{equation}
there must exist an operator, which we explicitly construct as $A \equiv \frac{\tilde O_{\lambda^*}}{\|\tilde O_{\lambda^*}\|}$, satisfying the following properties:
\begin{align}
\|A\| & = 1, \label{eq:finite-time-slow-op-1} \\
[H,A]& = \lambda^\ast A, \label{eq:finite-time-slow-op-2} \\
\mathrm{tr}[AQ_i]& = 0 \,\, \forall Q_i, \label{eq:finite-time-slow-op-3} \\
\|A\|_\calE^2 & \geq \delta \label{eq:finite-time-slow-op-4}.
\end{align}
In general, we could be considering an observable $O$ that is not orthogonal to $\{Q_i\}$. In this case, take $\tilde O=\calP[O]$, and $\|\tilde O(t)\|_\calE$ would measure the average deviation between $\langle \psi(t)| O|\psi(t)\rangle$ and $\langle \psi|\calQ[ O]|\psi\rangle$. When $\{Q_i\}$ is chosen as $\{H^k\}_{k=0}^n$, this implies
\begin{equation}
D_\infty(\calE,O;\Onpow) =\lim_{T\to\infty} \frac{1}T\int_0^{T} \|\tilde O(t)\|_{\calE}^2 \mathrm dt
\end{equation}
Further noting $\|O\| \geq \|\tilde O\|$ establishes Theorem~\ref{theorem-inf-time}. Thus, if the system does not thermalize, simple non-trivial IOMs or SGAs must exist. Equivalently, if we posit that the maximum $\|A\|_\calE^2$ for operators satisfying Eqs.~\eqref{eq:finite-time-slow-op-1} to \eqref{eq:finite-time-slow-op-3} is $\nu$, then the time-and-ensemble-averaged deviation in Eq.~\eqref{eq:inf-therm-bound-in-derivation} is upper-bounded by $\nu$. If $\nu$ decreases to zero as system size is increased, it follows that most initial states thermalize in the thermodynamic limit. We shall emphasize that $\nu$ depends not only on the Hamiltonian $H$ and the choice of the ensemble $\calE$, but also on the choice of the set $\{Q_i\}$.

The results obtained above require all the oscillating terms $e^{i(\lambda-\lambda^\prime)t}$ to average to zero. This would happen on a timescale of $1/|\lambda-\lambda^\prime|$, inverse of the smallest gap spacing, that is expected to be exponentially long for quantum many-body systems~\cite{shortEquilibrationQuantumSystems2011,shortQuantumEquilibrationFinite2012,pilatowsky-cameoQuantumThermalizationMust2025}. Therefore, it would be of interest to derive results that are applicable at shorter, physically relevant timescales. In Appendix~\ref{sec:thermalization-in-finite-time}, we prove Theorem~\ref{theorem-finite-time}, which we state here in the contrapositive:
\begin{equation}
\int_{0}^{+\infty} \|\calP[O](t)\|_\calE^2 \frac{2}{\pi} \frac{\sin^2 (t/\tau)}{t^2/\tau} \mathrm dt \leq 2 \nu(\tau) \|O\|^2, \label{eq:O-t-thermalization-finite-time-bound}
\end{equation}
given that for all operators $A$ satisfying
\begin{align}
\|A\| & = 1,  \label{eq:finite-time-slowop-cond-1} \\
\|[H,A] - \lambda A \|& \leq 1/\tau \text{ for some }\lambda , \label{eq:finite-time-slowop-cond-2} \\
\mathrm{tr}[AQ_i]& = 0\,\,\forall Q_i,
\end{align}
the following inequality holds:
\begin{equation}
\|A\|_\calE^2 \leq \nu(\tau). \label{eq:finite-time-slowop-cond-4}
\end{equation}
The envelope function $f_\tau(t)=\frac{2}{\pi} \frac{\sin^2 (t/\tau)}{t^2/\tau}$ is normalized and roughly measures the time interval $t \lesssim \tau$: one can see that $\frac{2}{\pi} \frac{\sin^2 (t/\tau)}{t^2/\tau}\approx \frac{2}{\pi\tau}$ if $t \lesssim \tau$, and for $t \gg \tau$ it vanishes as $\tau/t^2$. The $\nu$ here, which now also depends on $\tau$, gives an upper bound on how well the system has relaxed on a timescale of $\tau$. The $\tau\to\infty$ limit, $\nu(\infty)$, reduces to the $\nu$ defined in the infinite-time case. We note that Eq.~\eqref{eq:O-t-thermalization-finite-time-bound} is a factor of $2$ looser than Eq.~\eqref{eq:infinite-time-thermalization-bound-derivation} in this limit, which arises due to a technical reason laid out in Appendix~\ref{sec:thermalization-in-finite-time}.

We can also switch back to a more conventional time-average. For any $T<\pi\tau$,
\begin{equation}
\frac{1}{T} \int_0^{T} \|\calP[O(t)]\|_\calE^2 \mathrm dt < \frac{\pi T/\tau}{(\sin (T/\tau))^2} \nu(\tau) \|O\|^2. \label{eq:O-t-thermalization-finite-time-bound-with-T}
\end{equation}
Assume that $T\ll \tau$, the right-hand side simplifies to
\begin{equation}
\frac{1}{T} \int_0^{T} \|\calP[O(t)]\|_\calE^2 \mathrm dt < \pi \frac{\tau \nu(\tau)}{T} \|O\|^2. \label{eq:finite-time-therm-one-over-T}
\end{equation}

\subsection{Thermalization at finite temperature}
\label{subsec:thermalization-at-finite-temperature}

The thermalization bounds derived in the previous Subsection, Eq.~\eqref{eq:O-t-thermalization-finite-time-bound}, work for any ensemble $\calE$, including RPS, RMPS, and any finite-temperature ensemble, provided that ``simple operators'' are defined accordingly; \textit{i.e.}, the same ensemble over which the deviation is averaged in Eq.~\eqref{eq:O-t-thermalization-finite-time-bound} should also be used in Eq.~\eqref{eq:finite-time-slowop-cond-4}. However, the fact that the ensemble variance norm is (upper-bounded by) the exponential of negative operator size is only proven for RPS and RMPS, both infinite-temperature ensembles; meanwhile, it is less clear how to characterize operators that have a large ensemble variance norm for finite-temperature ensembles. Fortunately, we find that for certain choices of the finite-temperature ensemble, a thermalization bound at infinite temperature can be transferred to the finite-temperature case. That is, the absence of SSOs as defined by the RPS- or RMPS-norm also implies the thermalization of certain finite-temperature ensembles. We discuss two choices of the finite-temperature ensemble that allow this transfer: (1) stabilizer product states within an energy shell~\cite{huangRandomProductStates2024} and (2) thermally dressed R(M)PS ensembles~\cite{whiteMinimallyEntangledTypical2009,mcginleyScroogeEnsembleManybody2025}.

First, consider the set of stabilizer product states in a given energy shell. Ref.~\cite{huangRandomProductStates2024} proposed the set
\begin{equation}
\text{RPS}_\beta^{\text{MC}} = \{|\psi\rangle \in \text{RPS}| |\langle \psi|H|\psi\rangle-E(\beta)| =o(L) \},
\end{equation}
where $o(L)$ means the quantity divided by $L$ vanishes in the limit of large $L$, $E(\beta)$ is the energy corresponding to inverse temperature $\beta$, and RPS should be understood as the stabilizer product state set. It was proven that
\begin{equation}
\frac{|\text{RPS}_\beta^{\text{MC}}|}{|\text{RPS}|} \sim e^{-c\beta L}
\end{equation}
at small enough $\beta$. As a consequence, the ensemble variance norm of the state ensemble that uniformly draws from $\text{RPS}_\beta^{\text{MC}}$ satisfies
\begin{align}
\| O\|_{\text{RPS}_\beta^{\text{MC}}}^2 & = \frac{1}{|\text{RPS}^{\text{MC}}_\beta|} \sum_{|\psi\rangle \in \text{RPS}^{\text{MC}}_\beta} |\langle \psi|O|\psi\rangle|^2 \nonumber \\
& \leq \frac{1}{|\text{RPS}^{\text{MC}}_\beta|} \sum_{|\psi\rangle \in \text{RPS}} |\langle \psi|O|\psi\rangle|^2 \nonumber \\
& = \frac{|\text{RPS}|}{|\text{RPS}^{\text{MC}}_\beta|} \frac{1}{|\text{RPS}|} \sum_{|\psi\rangle \in \text{RPS}} |\langle \psi|O|\psi\rangle|^2 \nonumber \\
&\sim e^{c\beta L} \|O\|_\text{RPS}^2.
\end{align}
Therefore, if
\begin{equation}
\int_{0}^{+\infty}  \|\calP[O(t)]\|_\text{RPS}^2 \frac{2}{\pi} \frac{\sin^2 (t/\tau)}{t^2/\tau} \mathrm dt = \nu,
\end{equation}
this would imply that
\begin{equation}
\int_{0}^{+\infty}  \|\calP[O(t)]\|_{\text{RPS}^{\text{MC}}_\beta}^2 \frac{2}{\pi} \frac{\sin^2 (t/\tau)}{t^2/\tau} \mathrm dt \lesssim e^{c\beta L} \nu.
\end{equation}
If $\nu$ is $e^{-\Theta(L)}$, the right-hand side would also be $e^{-\Theta(L)}$ at small $\beta$. Therefore, the infinite-temperature bound implies a bound at high but finite temperature.

One could also consider the dressed ensemble $\widetilde{\calE_\beta}$ based on an infinite-temperature ensemble $\calE$, the RPS or RMPS. We show that a thermalization bound at infinite temperature similarly implies a bound that holds for high but finite temperatures. Since Eq.~\eqref{eq:O-t-thermalization-finite-time-bound} holds for any operator $O$, we can insert $O^\beta$ (as defined in Eq.~\eqref{eq:thermal-dressing-of-operator}) to get
\begin{equation}
\int_{0}^{+\infty} \|\calP[O^\beta(t)]\|_\calE^2 \frac{2}{\pi} \frac{\sin^2 (t/\tau)}{t^2/\tau} \mathrm dt < 2 \nu(\tau) \|O^\beta\|^2.
\end{equation}
We show that this implies (see Appendix~\ref{sec:appendix-thermal-value-finite-temp})
\begin{equation}
\int_{0}^{+\infty} \|\calP_\beta[O(t)]\|_{\widetilde{\calE_\beta}}^2 \frac{2}{\pi} \frac{\sin^2 (t/\tau)}{t^2/\tau} \mathrm dt < 2 \nu(\tau) \|O^\beta\|^2.
\end{equation}
Here, the $\nu(\tau)$ is the same $\nu$ as in $\calE$. $\calP_\beta$ is a projection against finite-temperature-dressed $\{H^k\}_{k=0}^n$, and thus $\calQ_\beta[O] = O-\calP_\beta[O]$ gives a Taylor expansion of the thermal expectation value around a non-zero energy density (see Appendix~\ref{sec:appendix-thermal-value-finite-temp} for details). $\|O^\beta\|$ can be bounded by
\begin{equation}
\|O^\beta\| = \frac{\|e^{-\frac{\beta H}{2}} Oe^{-\frac{\beta H}{2}} \|}{\langle e^{-\beta H} \rangle} \leq \frac{e^{-\beta E_\mathrm{min}}}{\langle e^{-\beta H}\rangle} \|O\| \lesssim e^{\calO(\beta L)} \|O\|.
\end{equation}
Therefore, similar to the previous case, a bound in the infinite temperature case implies a bound at high but finite temperature, assuming $\nu$ is $e^{-\Theta(L)}$.

\section{Simple slow operators} \label{sec:slow-operators}

In this Section, we will discuss the physical content of SSOs in multiple settings, ranging from chaotic to prethermal and integrable systems. We make two brief comments before proceeding. First, we reiterate that the concept of the simplicity of an operator depends on the choice of a state ensemble $\calE$; while the formalism developed should hold for general ensembles, all the physical discussions and numerical calculations are performed specifically with $\calE=\text{RPS}$ in qubit systems. Second, while the concept of SSOs involved in the thermalization theorems derived in Subsection~\ref{subsec:thermalization-and-slow-operators} includes both simple approximate IOMs and simple approximate SGAs, we will exclusively study SSOs as simple approximate IOMs. We justify this by noting that SGAs are rare in quantum many-body systems, and physically relevant SSOs are often IOMs. A detailed discussion of the role of SGAs, including physical and numerical evidence for sidelining them, is presented in Appendix~\ref{sec:sga}.

\subsection{The $\tau$-$\nu$ plane of operators}

Given any operator $A$, normalized to have unit Frobenius norm, $\|A\|=1$, we can use two parameters $(\tau,\nu)$ to characterize it, defined as $1/\tau^2 = \|[H,A]\|^2$ and $\nu = \|A\|_\text{RPS}^2$. They represent the slowness and simplicity of an operator, respectively. We study the set of representative points $(\tau,\nu)$ of all possible operators that are orthogonal to a given set $\{H^k\}_{k=0}^n$,
\begin{equation}
\calR_n = \left\{\left(1/\|[H,A]\|, \|A\|_\text{RPS}^2\right) : \|A\|=1, A \perp \{H^k\}_{k=0}^n \right\}.
\end{equation}
We have used $A\perp\{Q_i\}$ as a shorthand for $\tr[AQ_i]=0,\forall Q_i$, and we explicitly emphasize that the shape of the region depends on the choice of $\{Q_i\}$, or $n$ here. This construction is related to the joint numerical range (JNR) of the two matrices~\cite{barvinokCourseConvexity2002}. We used $\tau=1/\sqrt{\|[H,A]\|^2}$, since it has a clear physical meaning of timescale. It is worth noting that using $1/\tau^2$ will render $\calR$ an actual JNR, which has nicer mathematical properties. We discuss this in detail in Appendix~\ref{sec:jnr}.

\begin{figure*}[!t]
    \centering
    \includegraphics{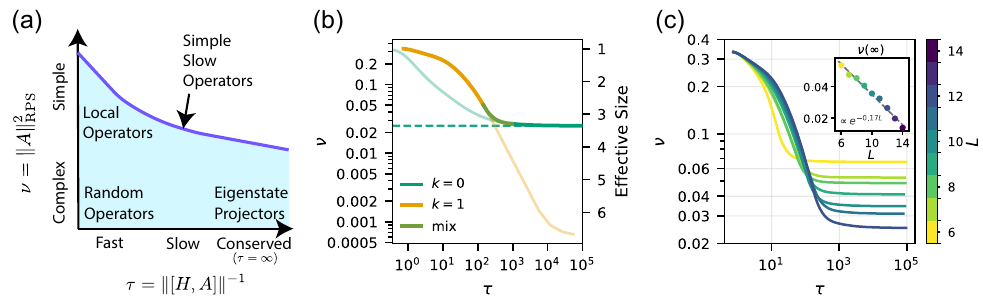}
    \caption{(a) Schematic demonstration of the $\tau$-$\nu$ plane characterization of operators. (b) The numerical $\nu(\tau)$ frontier for a chaotic mixed-field Ising model $H=\sum_{i} JZ_i Z_{i+1} + h_x X_i + h_z Z_i$ with $J=1$, $h_x=0.905$, and $h_z=0.809$, on an $L=12$ chain with PBC, with $n=4$. The effective operator size $(-\log_3\nu)$ is shown alongside $\nu$. Color of the line indicates the momentum sector in which the operator lies, with $k=qL/2\pi$, where $q$ is the lattice momentum. The dimmed curves probe the simplest operator \textit{within the momentum sector} at a given $\tau$. Horizontal dashed line shows $\nu(\tau=\infty)$. (c) The same $\nu(\tau)$ frontier as in (b) for MFIM chains of different sizes $L$, with the inset showing the decay of $\nu(\infty)$ with respect to $L$.}
    \label{fig:nutau-mfim}
\end{figure*}

As shown in Fig.~\ref{fig:nutau-mfim}(a), any operator can be represented as a point in this $\tau$-$\nu$ plane. A randomly drawn operator from the entire operator space would have small $\nu$, since it would mainly consist of large Pauli strings (see Eq.~\eqref{eq:average-RPS-norm}), and small $\tau$, since it is not conserved and would have a large commutator with the Hamiltonian. For any system, there are always operators with either large $\nu$ or large $\tau$. To get large $\nu$, we can pick single-site operators (properly orthogonalized against the powers of $H$), such that $\nu=\Theta(1)$; however, such operators would generally not be conserved, hence still have small $\tau$. On the other hand, to get large $\tau$, we can choose eigenstate projectors, which are rigorously conserved, with $\tau = \infty$; yet eigenstate projectors are likely to be very complicated in the Pauli basis and thus have very small $\nu$. Whether or not operators with both large $\nu$ and large $\tau$ exist is not guaranteed. These are exactly the operators that lie on (or near) the upper boundary of $\calR_n$, or the \textit{Pareto frontier} for $\nu$ and $\tau$. These operators are SSOs, and their existence or absence determines the thermalization properties of the system.

Reading from this angle, Eqs.~\eqref{eq:finite-time-slowop-cond-1} to \eqref{eq:finite-time-slowop-cond-4} are looking for operators that have the largest $\nu$ for a given $\tau$. (To be precise, the actual $\tau$ of an operator should be defined as $1/\tau^2 = \inf_{\lambda} \|[H,A]-\lambda A\|^2$ in line with Eq.~\eqref{eq:finite-time-slowop-cond-2}; we have implicitly assumed that this infimum is achieved at $\lambda=0$ due to the physical unlikeliness of simple SGAs, with detailed justifications given in Appendix~\ref{sec:sga}.) Whether such operators exist corresponds to the solution of the following optimization problem, 
\begin{equation}
\begin{array}{lll}
\nu_n(\tau) =&\text{maximize} \quad & \|
A\|_\text{RPS}^2, \\
&\text{subject to} & \|A\|^2 = 1,  \\
&& \mathrm{tr}[AH^k] = 0, \, 0\leq k \leq n, \\
&& \|[H,A]\|^2 \leq 1/\tau^2.
\end{array} \label{eq:slow-op-sigma-eq}
\end{equation}
Intuitively, this defines the simplest operator that is approximately conserved at a timescale $\tau$ (and independent of the selected powers of $H$). This maximum, denoted as $\nu_n(\tau)$, maps out the frontier of the region $\calR_n$. The problem can be solved with Lagrange multipliers,
\begin{equation}
\frac{\partial}{\partial A}\left[ \|
A\|_\text{RPS}^2 - \lambda_1 \|[H,A]\|^2 - \lambda_2  \|A\|^2\right] = 0,
\end{equation}
which can be cast into a generalized eigenvalue problem in operator space (see Appendix~\ref{sec:jnr} for details),
\begin{equation}
\calM_\text{RPS}[A] = \mu \left(A+\theta^2[H,[H,A]]\right),\label{eq:Lagrange-mult-eigenvalue-problem}
\end{equation}
where $\calM_\text{RPS}$ is the kernel of the RPS-norm, defined in Eq.~\eqref{eq:evn-kernel-def}, which acts on Pauli strings as $\calM_\text{RPS}[P] = 3^{-|P|}P$. This generalized eigenvalue problem should be solved subject to the constraint $\mathrm{tr}[AH^k] = 0$ for $0\leq k \leq n$. For each $\theta$, we find the largest generalized eigenvalue $\mu(\theta)$ and the corresponding eigenoperator $A(\theta)$. It can be shown that $A(\theta)$ is the solution to the optimization in Eq.~\eqref{eq:slow-op-sigma-eq} with $\tau =\tau(\theta) = 1/\|[H,A(\theta)]\|$, with the maximal value being $\nu = \nu(\theta) = \|A(\theta)\|_\text{RPS}^2$. Therefore, as we vary $\theta$, we can get the frontier curve $\nu_n(\tau)$ parameterized by $\theta$. The properties of JNRs guarantee that the function $\tau(\theta)$ is invertible in general situations, hence by varying $\theta$ we are guaranteed to get the full curve. This provides a concrete numerical method to map out the $\nu_n(\tau)$ frontier for any given system.

The constrained optimization problem, or the derived generalized eigenvalue problem, can only be exactly solved when an exact diagonalization (ED) of the Hamiltonian is possible. At beyond-ED system sizes, while the exact frontier is generally not computable, we can choose an ansatz and make variational approximations to the frontier. By the nature of the optimization problem, for any normalized operator $O$, we can compute $\tau_O$ and $\nu_O$, and it would be guaranteed that $\nu(\tau_O) \geq \nu_O$. Therefore, if there is a family of operators $O_s$, the curve $(\tau_{O_s}, \nu_{O_s})$ parametrized by $s$ would give a variational lower bound to the $\nu(\tau)$ frontier. Such trial operators $O$ can either be obtained by constructing physical slow operators by hand, or solving the optimization problem within a restricted ansatz.

\subsection{Structure of the $\nu(\tau)$ frontier on a finite chaotic spin chain}
\label{subsec:nutau-tfim}

To give a schematic picture of the structure of the $\nu(\tau)$ frontier, we present in Fig.~\ref{fig:nutau-mfim}(b) the numerically computed frontier curves for a chaotic mixed-field Ising model (MFIM) on a finite-length chain with periodic boundary conditions (PBC), with the parameter $n$ adequately (see discussions below) chosen. As expected, the curve starts at $\nu\approx 0.3$ at very small $\tau$, where there is effectively no commutator constraint and one may just choose some (combinations of) single-site Paulis (whose $\nu=1/3$) that satisfy the orthogonal-to-$H^k$ constraint. As $\tau$ increases, such single-site operators need to include some dressings to satisfy the commutator constraint, and the $\nu$ of the simplest operator decreases. In the opposite limit, at $\tau=\infty$, $\nu(\infty)$ looks for the simplest non-trivial operator that is exactly conserved. A chaotic model has no such IOMs, and operators that commute with $H$ have to be constructed using eigenstate projectors; such operators are generally complex in the Pauli basis and have small $\nu(\infty)$. 

For a spatially homogeneous Hamiltonian on a PBC chain, there is an exact translation symmetry, and we can talk about momentum sectors. We will use an integer $k = qL/2\pi$ to label momentum sectors, where $q$ is the momentum. As shown in Fig.~\ref{fig:nutau-mfim}(b), for the large-$\tau$ part of the curve, the simplest operator lies in the $k=0$ sector, consistent with the fact that eigenstate projectors in the MFIM are translationally invariant. By contrast, the simplest operator at the small-to-intermediate-$\tau$ part lies in the $k=1$ sector, where the longest-wavelength hydrodynamic modes live. At some intermediate $\tau$, the $k=1$ branch gives way to the $k=0$ branch. Notably, there is an interval near this crossover $\tau$ where the simplest operator is a superposition of an operator in the $k=0$ sector and one in the $k=1$ sector. (See Appendix~\ref{sec:jnr} for a discussion of the nature of such superposition.)

Figure~\ref{fig:nutau-mfim}(c) shows how this $\nu(\tau)$ frontier changes as the system size $L$ is varied. We can observe that in the small-$\tau$ regime, the larger-$L$ curve has a larger $\tau$ at the same $\nu$. This is expected, as the longest-wavelength hydrodynamic mode becomes slower as $L$ is increased. On the other hand, for large $\tau$, the $\nu$ value decreases significantly as $L$ is increased. $\nu(\infty)$ is shown to decay with $L$ in an exponential fashion, indicating that the simplest conserved operator in this system has extensive operator size. With the thermalization bound, this exponential scaling of $\nu(\infty)$ would imply that the average deviation of the real-time expectation value of observables from the thermal expectation value is exponentially small in system size, and thus vanishes in the thermodynamic limit.

These examples illustrate the general structure of such $\nu(\tau)$ frontiers and the physical messages they convey. At very small $\tau$, $\nu$ is expected to be around $1/3$, regardless of the underlying Hamiltonian. At intermediate $\tau$, the simplest operator corresponds to the slowly relaxing physical modes; for chaotic spin chains, it is usually the longest-wavelength hydrodynamic mode. At large $\tau$, the curve flattens towards the $\tau=\infty$ result, which probes the simplest non-$H$-power operator that commutes with the Hamiltonian, and reveals whether a genuine simple IOM exists in the system. We note that due to the existence of hydrodynamic modes, the $\tau$ at which $\nu(\tau)$ approaches $\nu(\infty)$ would diverge in the thermodynamic limit; more precisely, it should be the Thouless timescale, as we will demonstrate in Subsection~\ref{subsec:nutau-and-thermalization-timescale}.

Before closing this Subsection, we briefly comment on the role of the parameter $n$. All powers of $H$ are conserved operators, and for any finite value of $n$, one can use $H^{n+1}$ to construct a conserved operator $K_{n+1}$ that is orthogonal to the tower $\{H^k\}_{k=0}^n$, whose size is at most $m(n+1)$ if $H$ is $m$-local. Therefore, if $n$ is too small, this $K_{n+1}$ will overshadow the other operators discussed above; in particular, it would impose $\nu_n(\infty) \geq 3^{-m(n+1)}$, a lower bound that is system-size independent. Therefore, for small $n$, we expect to see $K_{n+1}$ on the $\nu_n(\tau)$ frontier (at least at large $\tau$); in this case, $\nu_n(\tau)$ would depend sensitively on $n$. On the other hand, we observe that above a certain value of $n$, $\nu_n(\tau)$ would stabilize, indicating that the operator lying on the frontier is no longer $K_{n+1}$. Therefore, when we say ``choosing an adequately large value of $n$'', we mean choosing $n$ above this threshold $n_\ast$ so that the SSO we obtain on the frontier is independent of $K_{n+1}$. A semi-quantitative estimation indicates that the threshold value $n_\ast \propto L$ (see Appendix~\ref{sec:n-dependence-nutau}). In Fig.~\ref{fig:nutau-mfim}(c), we have chosen $n$ to be above $n_\ast$ so that $\nu_n(\infty)$ decays exponentially with $L$; we will be sticking to this choice in all numerical results presented in this Article. It is possible, however, to choose $n<n_\ast$ and still get a vanishing $\nu_n(\infty)$ in the thermodynamic limit. In such cases, $\nu_n(\infty) \approx \|K_{n+1}\|_\text{RPS}^2 \sim r^n$, where $r<1$ is a constant (see Appendices~\ref{sec:n-dependence-nutau} and \ref{sec:rps-Kn}); therefore, choosing $n \sim \frac{\log\epsilon}{\log r}$ suffices to obtain thermalization bounds of any fixed accuracy $\epsilon$, and choosing $n=\Theta(\log L)$ makes $\nu$ polynomially small in $L$. It is only when discussing finite-temperature thermalization, in which case $\nu_n(\infty)$ is \textit{required} to vanish exponentially with $L$, that we have to choose $n=\Theta(L)$.

\begin{figure*}[!t]
    \centering
    \includegraphics{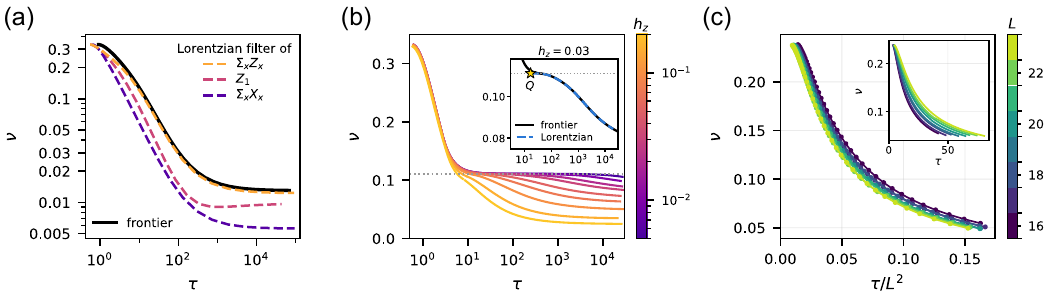}
    \caption{(a) The $\nu(\tau)$ curve of a KIM $U=e^{-ih_x\sum_x X_x} e^{-i \sum_x\left(J Z_x Z_{x+1} + h_z Z_x\right)}$ with $J=1$, $h_x=0.905$, $h_z=0.809$ on an $L=12$ PBC chain, compared with trajectories of the Lorentzian filter applied on some local or extensive base operators. The base operator $\sum_x Z_x$ almost exactly reproduces the frontier. (b) The $\nu(\tau)$ frontier for a family of MFIMs with $J=1$ and $h_x=0.9$, with varying $h_z$, on a PBC chain with $L=12$, within the translationally invariant sector, with $n=4$. The $h_z=0$ model is integrable, showing a $\nu(\tau)$ plateau that extends indefinitely. The curves at finite but small $h_z$ overlap with the $h_z=0$ curve at small $\tau$ and deviates at large $\tau$. The inset shows that the frontier can be reproduced by modulating the bare slow operator $Q = \sum_x Y_x Z_{x+1}-Z_x Y_{x+1}$ (with $\|[H,Q]\| \sim h_z \|Q\|$) with the Lorentzian filter. (c) The $\nu(\tau)$ trajectory of the Lorentzian modulation of the hydrodynamic mode $H_{q=\frac{2\pi}{L}}$ for the MFIM with $J=1$, $h_x=-1.05$, $h_z=0.5$ on finite-size PBC chains. The curves are shown to roughly collapse by rescaling $\tau$ with $L^2$, consistent with the fact that such systems relax on a Thouless timescale proportional to $L^2$.  }
    \label{fig:thermdyn}
\end{figure*}

\subsection{The $\nu(\tau)$ frontier in a Krylov chain model and the Lorentzian filter}

The $\nu(\tau)$ frontier of the MFIM displays two sets of physics: at small $\tau$, it reflects the slow relaxation associated with hydrodynamics, and at large $\tau$, it reflects the thermalization of chaotic systems. In this Subsection, we further illustrate the $\nu(\tau)$ frontier of chaotic systems by choosing systems without hydrodynamics. We present the solution for a Krylov chain effective model, which approximates the physics of the large-$q$ Sachdev-Ye-Kitaev (SYK) model~\cite{robertsOperatorGrowthSYK2018}. We extract several points of physical intuitions from this solution, and compare it against the exact numerical $\nu(\tau)$ frontier of a Floquet kicked Ising model (KIM)~\cite{lakshminarayanMultipartiteEntanglementOnedimensional2005,akilaParticletimeDualityKicked2016,bertiniEntanglementSpreadingMinimal2019} on a finite chain.

The Krylov chain model considers a subspace of operators spanned by an orthonormal basis $\{O_0,O_1,\dots\}$ which is closed with respect to the Hamiltonian commutator $[H,\cdot]$ and the RPS-norm kernel, whose matrix elements are given by
\begin{align}
\langle O_m,[H,O_n]\rangle & = \begin{cases}
\sqrt{(n+1)(n+2)}, & m=n+1;\\
\sqrt{n(n+1)}, & m=n-1; \\
0, & \text{otherwise},
\end{cases} \\
\langle O_m,\calM_\text{RPS}[O_n]\rangle & = \delta_{mn} e^{-\mu n}/3.
\end{align}
This model is similar (though not rigorously identical) to the Krylov subspace spanned by a single Majorana operator in the large-$q$ SYK model~\cite{parkerUniversalOperatorGrowth2019b}, and the parameter $\mu=(q-2)\log 3$.

As $q\to \infty$, we have $\mu \to \infty$, and $\langle O_m,\calM_\text{RPS}[O_n]\rangle = \delta_{m0}\delta_{n0}/3$. In this case, the $\nu(\tau)$ frontier of this model admits an exact solution (see Appendix~\ref{sec:krylov-frontier}). The operator on the frontier is obtained by applying a Lorentzian filter on the base operator $O_0$,
\begin{equation}
O_\theta =\int_{-\infty}^{+\infty} \frac{e^{-\frac{|t|}{\theta}}}{2\theta} O_0(t) \mathrm dt,
\end{equation}
where $O_0(t) = e^{i[H,\cdot] t}[O_0]$. This is called a Lorentzian filter as the time-integration kernel $\frac{e^{-\frac{|t|}{\theta}}}{2\theta}$ is the Fourier transform of $(1+\theta^2 \omega^2)^{-1}$. The resulting $\nu(\tau)$ frontier has the form $\nu \approx 2/(3\tau)$ in the large-$\tau$ limit.

While this Krylov chain model is a simplified ideal model, we take away two lessons from it that carry over to more general models. Firstly, for a chaotic system without conservation laws in the thermodynamic limit, $\nu(\tau)$ would decay as $1/\tau$ in the large-$\tau$ limit, and thus vanishes as $\tau\to \infty$. This $1/\tau$ behavior can be intuitively seen from the thermalization bound. Consider $D_{f_\tau}(\text{RPS},O;\Onpow)$. For a local observable $O$, it is expected that in $\Theta(1)$ time, $O(t)$ would still have an $\Theta(1)$ component on the low-weight subspace, so that $\|O(t)\|_\text{RPS}^2 \sim \Theta(1)$; meanwhile, $f_\tau(t)$ is $\Theta(1/\tau)$ at $t \sim \Theta(1)$, which would imply that $D_{f_\tau}(\text{RPS},O;\Onpow)$ is at least on the order of $1/\tau$; therefore, $\nu(\tau)$ decays at most as fast as $\Theta(1/\tau)$. This scaling is also present in several previous thermalization bounds~\cite{shortQuantumEquilibrationFinite2012,garcia-pintosEquilibrationTimeScales2017}. Secondly, the SSO on the $\nu(\tau)$ frontier can be obtained from Lorentzian-filtering a physical ``base operator''. Such modulation exactly reproduces the frontier in the Krylov chain model and, as shown in Fig.~\ref{fig:thermdyn}(a), approximates the frontier of a Floquet KIM very well with an appropriate choice of the base operator.

\subsection{$\nu(\tau)$ and thermalization timescale}
\label{subsec:nutau-and-thermalization-timescale}

We now discuss how the $\nu(\tau)$ language describes slow relaxation. Intuitively, $\nu(\tau)$ should be able to detect the thermalization timescale of these systems: $\nu(\tau)$ should remain large for $\tau \lesssim \tau_\ast$, where $\tau_\ast$ is the thermalization timescale of the system, and decay as $\tau$ is further increased. We exemplify this in two settings: an integrable system with a weak integrability breaking term, and the hydrodynamic modes in a chaotic spin chain. To make the Subsection concise, we reserve more detailed numerical results to Appendix~\ref{sec:numerics}.

In Fig.~\ref{fig:thermdyn}(b), we show a family of MFIMs with weak longitudinal fields $h_z$. At $h_z=0$, the model reduces to the TFIM, where simple IOMs exist and a corresponding $\nu(\tau)$ plateau extends to $\tau=\infty$. As a weak $h_z$ is tuned on, the curve at small $\tau$ remains essentially unperturbed, while that at large $\tau$ deviates downwards. With larger $h_z$, the $\tau$ at which the deviation happens becomes smaller, indicating faster thermalization. We further identify the SSO lying on the $\nu(\tau)$ frontier curve as the Lorentzian-filtering of the simplest IOM of the TFIM. Notably, such filtering makes the slow operator parametrically more conserved: a bare IOM $Q$ of the TFIM would have $\|[H,Q]\| \sim h_z$, corresponding to a lower bound on the thermalization time on the scale of $h_z^{-1}$~\cite{mohanQuantumSpeedLimits2022,carabbaQuantumSpeedLimits2022,garcia-pintosUnifyingQuantumClassical2022}; however, realistic systems often relax on longer timescales~\cite{parkerUniversalOperatorGrowth2019b,bartschEstimationEquilibrationTime2024a,moriThermalizationPrethermalizationIsolated2018a,mallayyaPrethermalizationThermalizationIsolated2019a,mallayyaPrethermalizationThermalizationFermis2021,abaninRigorousTheoryManyBody2017b,suraceWeakIntegrabilityBreaking2023b}. The filtered operator would have $\|[H,Q_\theta]\| \sim h_z^2$ (while $\|Q_\theta\|_\text{RPS}$ remains large), correctly reproducing the Golden-rule thermalization timescale $\tau_\ast \sim h_z^{-2}$~\cite{moriThermalizationPrethermalizationIsolated2018a,mallayyaPrethermalizationThermalizationIsolated2019a,mallayyaPrethermalizationThermalizationFermis2021} (see Appendix~\ref{sec:numerics}).

We can also observe a similar phenomenon in chaotic spin chains. Although general chaotic systems are not considered ``prethermal'', for spin chains with local interactions, slow modes also exist in the system, which are the hydrodynamic modes
\begin{equation}
H_q = \sum_x e^{iqx} h_x,
\end{equation}
where $h_x$ are the local energy densities such that $\sum_x h_x = H$. When $q$ is small, $\|[H,H_q]\| \sim q$, and we can treat $H_q$ in the same fashion as approximate conserved quantities in prethermal systems. For the longest-wavelength mode, $q \sim L^{-1}$, giving a timescale $\tau_\ast \sim q^{-2} \sim L^2$, which is exactly the Thouless time~\cite{varmaEnergyDiffusionErgodic2017,dymarskyMechanismMacroscopicEquilibration2019,dymarskyBoundEigenstateThermalization2022b}. Therefore, we expect the $\nu(\tau)$ frontier of a chaotic spin chain to remain at $\nu \sim \Theta(1)$ for $\tau \sim L^2$, propped up by the spectral-filtering of the longest-wavelength hydrodynamic mode.

Since the $L$ required to make $\frac{2\pi}{L}$ reasonably small is beyond the size where ED is possible, we present the $\nu(\tau)$ curves obtained variationally by directly modulating $H_{q=2\pi/L}$ with the Lorentzian. (Note that on PBC chains, this operator is automatically orthogonal to any power of $H$ since they lie in different momentum sectors.) In Fig.~\ref{fig:thermdyn}(c), the $\nu(\tau)$ curves obtained in this way for MFIM chains of different lengths $L$ indeed demonstrate a collapse upon rescaling $\tau$ by $L^2$, consistent with the scaling of the Thouless time.

\section{Implications}
\label{sec:implications}

In this Section, we discuss several implications of our results, in the context of Hilbert space fragmentation (Subsection~\ref{subsec:HSF}), weak ergodicity breaking (Subsection~\ref{subsec:scars}), and operator growth (Subsection~\ref{subsec:operator-growth}).

\subsection{Hilbert space fragmentation}
\label{subsec:HSF}

Hilbert space fragmentation (HSF) is a strong form of ergodicity breaking where the Hilbert space shatters into exponentially many Krylov subspaces that cannot be connected by the Hamiltonian dynamics. Using the commutant algebra language~\cite{moudgalyaHilbertSpaceFragmentation2022b}, Hilbert space fragmentation occurs for a family of models $H=\sum_i c_i h_i$, where $h_i$ are local terms, if the Hilbert space decomposes into an exponential number of Krylov subspaces, $\calH = \bigoplus_\alpha \calH_\alpha$, where each $\calH_\alpha$ is a representation of the algebra generated by $\{h_i\}$; equivalently, $h_i \calH_\alpha \subset \calH_\alpha$. Most known models that experience HSF are classically fragmented, meaning that each $\calH_\alpha$ is spanned by several computational basis product states.

In systems experiencing classical HSF, product initial states cannot explore the full Hilbert space. Therefore, from the state-space perspective, such models do not thermalize. However, the picture from the operator-space perspective is not clear. One key feature of HSF systems is that Krylov subspaces cannot be explained by the quantum numbers of traditional IOMs alone. In certain models, there are families of $k$-local conserved quantities that can explain the Krylov sectors~\cite{moudgalyaHilbertSpaceFragmentation2022b}, while in others such quantities have not been explicitly constructed~\cite{brighiHilbertSpaceFragmentation2023a,chenQuantumFragmentationExtended2024}. Generally, the ``IOMs'' that control the Krylov sector decomposition are encoded in the commutant algebra~\cite{moudgalyaHilbertSpaceFragmentation2022b}. However, it was unclear whether or not the commutant algebra must contain local 
(or quasilocal / $k$-local) operators. Our results demonstrate that for any model experiencing the HSF, as long as the majority of product initial states do not thermalize (which is guaranteed to be true in the case of strong fragmentation), then SSOs must exist. In particular, if the diagonal ensemble expectation value differs from the thermal expectation value for most initial states, then simple IOMs must exist. (See Appendix~\ref{sec:sga} for the basis of the latter argument.)

\subsection{Weak ergodicity breaking}
\label{subsec:scars}

It is apparent that our discussion of thermalization has excluded weak ergodicity breaking in quantum many-body scarred systems, where most initial states will thermalize, yet a few low-entanglement initial states, known as the scar states, fail to~\cite{turnerWeakErgodicityBreaking2018a,shiraishiSystematicConstructionCounterexamples2017a,moudgalyaPairingHubbardModels2020,serbynQuantumManybodyScars2021,renQuasisymmetryGroupsManyBody2021a,moudgalyaExhaustiveCharacterizationQuantum2024a,moudgalyaQuantumManybodyScars2022b}. The number of scar states is usually exponentially small compared to the number of all product (or matrix product) states; therefore, the non-thermal nature of scar states would only contribute an exponentially small amount to $D_f(\calE,O;\Oth)$. Correspondingly, low-rank operators generally have exponentially small ensemble variance norm, as we prove in Appendix~\ref{sec:rps-norm-rank} for the RPS-norm. Hence, our formalism would fail to distinguish a fully thermal system from a system with weak ergodicity breaking. We argue that this fact is inevitable. On one hand, an operator supported on $k$ sites acts as the identity on the remaining $L-k$ sites of the system, and therefore must have rank at least $d^{L-k}$, which is exponentially large. Therefore, a formalism centered around simple operators will always translate into a statement about a collection of an exponential number of states in state space. On the other hand, it is known that no algorithm exists that can determine whether or not thermalization occurs for a particular initial state~\cite{shiraishiUndecidabilityQuantumThermalization2021,devulapalliComplexityThermalizationFinite2025}. Therefore, there might not exist any universal characterization of thermalization that can capture weak ergodicity breaking.

\subsection{Thermalization and operator size distribution}
\label{subsec:operator-growth}

Our results establish a direct connection between thermalization and the dynamics of operator size distribution. A series of works~\cite{schusterOperatorGrowthOpen2023,zhangDynamicalTransitionOperator2023} have studied the quantity $p(S,t)$, defined as $p_O(S,t) = \sum_{|P|=S} |c_P(t)|^2$ for an operator expanded in the Pauli string basis as $O(t)=\sum_P c_P(t) P$. The relationship between $p(S,t)$ and OTOCs is well-established~\cite{xuScramblingDynamicsOutofTimeOrdered2024}:
\begin{equation}
\sum_{x,i} \|[O(t),\Lambda_i^{(x)}]\|^2= 2d^2 \sum_S p_O(S,t)S,
\end{equation}
where $\{\Lambda_i\}$ is the basis of the operator space, and the superscript $(x)$ indicates the operator at site $x$. This is proportional to the square of the $\|\cdot \|_\calS$ norm introduced in Subsection~\ref{subsec:simple-operators-characterization}. Meanwhile, the ensemble variance norm can also be expressed in terms of $p(S,t)$ as~\cite{qiMeasuringOperatorSize2019}
\begin{equation}
\|O(t)\|_\text{RPS}^2 = \sum_S (d+1)^{-S} p_O(S,t).
\end{equation}
We have shown that thermalization is equivalent to the fact that for any operator $O$, the part of it orthogonal to the powers of $H$, $\calP[O(t)]$, has a vanishing ensemble variance norm at late times. Formally,
\begin{multline}
\expt_{\psi\sim \text{RPS}} \left[\langle \psi|O(t)|\psi\rangle - \Onpow(\psi)\right]^2\\
= \|\calP[O(t)]\|_\text{RPS}^2=\sum_S (d+1)^{-S} p_{\calP [O]}(S,t).
\end{multline}
This establishes a rigorous link between operator growth and thermalization: the average deviation of the expectation value of $O$ from the thermal expectation value can be directly read off the operator size distribution $p_{\calP[O]}(S,t)$. Furthermore, by Jensen's inequality, this also implies that
\begin{equation}
\sum_{x,i} \|[\calP[O](t),\Lambda_i^{(x)}]\|^2 \geq 2d^2 \frac{\log \|\calP[O(t)]\|_\text{RPS}^{-2}}{\log (d+1)},
\end{equation}
assuming $\|\calP[O]\|=1$. Therefore, thermalization provides a lower bound on the growth of OTOCs. Note that since Jensen's inequality holds only in one direction, the converse is not true: there is no guarantee that the growth of OTOCs would lead to thermalization of most product states.

\section{Conclusion and outlook}
\label{sec:conclusion}

We have established a rigorous relationship between the thermalization of typical initial states drawn from a general state ensemble $\mathcal{E}$ and the absence of SSOs with respect to that ensemble, quantitatively characterized by the $\nu(\tau)$ curve. 
Our main result is that if typical states from ensemble $\mathcal{E}$ fail to thermalize on timescale $\tau$, there exist SSOs that are approximately conserved on that timescale.
When applied to the specific ensemble of random product states, these SSOs are operators that have large weight on short Pauli strings, so the $\nu(\tau)$ curve is a generalization of the ``slowest local operator'' construction in Ref.~\cite{kimSlowestLocalOperators2015a}. This result can also be generalized to high but finite temperatures. The central apparatus, the ensemble variance norm, bears similarity to the Mix-Norm in classical ergodicity theory~\cite{mathewMultiscaleMeasureMixing2005}.
By relating the dynamics of specific state ensembles to that of operators, this result leads to two interesting implications that were not previously known: (i)~that thermalization of typical product initial states implies the growth of OTOCs, and (ii)~that strong Hilbert space fragmentation implies the existence of SSOs.

Our results rigorously reduce quantum thermalization to an operator space property, namely the $\nu(\tau)$ frontier function, which can in principle be determined by diagonalizing a superoperator. While explicit diagonalization is limited to small systems, there are encouraging signs that the determination of $\nu(\tau)$ can be done in larger systems using algebraic methods in operator space. For instance, a recently developed technique of utilizing operator commutator relations to rigorously prove absence of IOMs and SGAs with finite support has found success in a wide range of spin chain models~\cite{\absenceIOMs}. These techniques do not yet apply to operators that are simple in the RPS-norm, but an extension to this case is potentially within reach.

While our results already establish thermalization bounds for most physical initial states at finite time, the bound can potentially be further tightened in two directions: towards shorter time and lower temperature. On thermalization timescale, as our bound is uniform across observables, it will always capture the slowest mode that exists in the system, which can, for example, overestimate the timescale of local thermalization in systems with hydrodynamics. Besides, we have shown that our thermalization bound over a timescale of $\tau$ is at most as tight as $1/\tau$, which fails to characterize the fast exponential thermalization of chaotic systems, in particular the Ruelle-Pollicott resonances~\cite{prosenRuelleResonancesQuantum2002,prosenRuelleResonancesKicked2004,moriLiouvilliangapAnalysisOpen2024b,znidaricMomentumdependentQuantumRuellePollicott2024,jacobySpectralGapsLocal2025a,zhangThermalizationRatesQuantum2025,yoshimuraTheoryIrreversibilityQuantum2025a,duhRuellePollicottResonancesDiffusive2026}. Besides, the bound for finite-temperature ensembles is exponentially worse than the infinite-temperature bound. A better bound could be attainable by using a more natively finite-temperature ensemble. A promising starting point is provided by the stabilizer product state decomposition of the high-temperature Gibbs state~\cite{bakshiHighTemperatureGibbsStates2024a}, which was adopted in Ref.~\cite{pilatowsky-cameoQuantumThermalizationMust2025}, and recently generalized to lower temperatures~\cite{bakshiEntanglementQuantumSpin2026}. This requires working out a characterization of simple operators under these ensembles, which we leave as an interesting direction for future work.

\section*{Acknowledgments}

T.-H. Y. thanks Andrew G. Green, Boris Hanin, Wen Wei Ho, Qi Camm Huang, David A. Huse, J. Alexander Jacoby, Biao Lian, Liang-Hong Mo, Yi-Ting Tu, Weijun Wu, and Hironobu Yoshida for insightful discussions. Numerical calculations in this work are performed using the QuSpin library~\cite{weinbergQuSpinPythonPackage2017,weinbergQuSpinPythonPackage2019}. This work was partially supported by the Brown Investigator Award. 

\section*{Data availability}

The code and data supporting the findings of this Article are publicly available~\cite{\CODE}.

\section*{AI disclosure}

AI tools, mainly Claude Opus 4.5, 4.6 and 4.8, have been employed to write the code for numerical computation and derive technical results. GPT-6 Astra has been used to improve the manuscript.

\appendix

\section{Properties of the ensemble variance norm} \label{sec:norm-properties}

In this Appendix, we discuss certain properties of the ensemble variance norm.

The ensemble variance norm derives from an inner product, and is non-negative. Therefore, using the condition $\|A+\lambda B\|_\calE^2 \geq 0$ for all $\lambda$, we can establish the Cauchy-Schwarz inequality,
\begin{equation}
|\langle A,B\rangle_\calE| \leq \|A\|_\calE \|B\|_\calE.
\end{equation}
As a result, the triangle inequality also holds:
\begin{equation}
\|A+B\|_\calE \leq \|A\|_\calE + \|B\|_\calE.
\end{equation}

We will prove two specific properties that will be useful in the discussion of SGAs.

\begin{proposition}
For any operators $A$ and $B$,
\begin{equation}
\|AB\|_\calE^2 \leq \|AA^\dagger \|_\calE \|B^\dagger B\|_\calE.
\end{equation}
\label{prop:AB-product}
\end{proposition}
\begin{proof}
We have
\begin{equation}
|\langle \psi|AB|\psi\rangle| \leq \| A^\dagger|\psi\rangle \| \|B|\psi\rangle\|,
\end{equation}
which is
\begin{equation}
|\langle \psi|AB|\psi\rangle|^2 \leq \langle \psi| AA^\dagger |\psi\rangle \langle \psi|B^\dagger B|\psi\rangle. \label{eq:AB-expt-AAd-BBd}
\end{equation}
Therefore,
\begin{equation}
\|AB\|_\calE^2 \leq \langle AA^\dagger, B^\dagger B\rangle_\calE.
\end{equation}
By the Cauchy-Schwarz inequality, this gives
\begin{equation}
\|A B\|_\calE^2 \leq \|AA^\dagger\|_\calE \|B^\dagger B\|_\calE.
\end{equation}
\end{proof}

\begin{proposition}
For any operator $A$,
\begin{equation}
\|A\|_\calE \leq \sqrt{D\|\tilde \rho_\calE\|_\infty}\|A\|, 
\end{equation}
where $\|\tilde \rho_\calE\|_\infty$ is the largest eigenvalue of
\begin{equation}
\tilde \rho_\calE = \calM_\calE[I/D] = \expt_{\psi\sim\calE} \langle \psi|\psi\rangle |\psi\rangle \langle \psi|.
\end{equation}
If $\calE$ is a normalized ensemble, then $\tilde\rho_\calE = \rho_\calE$.
\label{prop:A-EV-HS}
\end{proposition}
\begin{proof}
Inserting $B=I$ to Eq.~\eqref{eq:AB-expt-AAd-BBd}, we have
\begin{equation}
|\langle \psi|A|\psi\rangle|^2 \leq \langle \psi|\psi\rangle \langle \psi| AA^\dagger |\psi\rangle.
\end{equation}
Taking average on both sides,
\begin{equation}
\|A\|_\calE^2 \leq \mathrm{tr} \left[AA^\dagger \tilde\rho_\calE \right].
\end{equation}
Therefore,
\begin{equation}
\|A\|_\calE^2  \leq \|\tilde \rho_\calE\|_\infty \mathrm{tr}[AA^\dagger] = D \|\tilde \rho_\calE\|_\infty \|A\|^2.
\end{equation}
\end{proof}

In particular, if the ensemble is normalized, Proposition~\ref{prop:AB-product} implies $\|A\|_\calE^2 \leq \min\left(\|AA^\dagger\|_\calE, \|A^\dagger A\|_\calE\right)$; if we also have $\rho_\calE=\tilde \rho_\calE = I/D$, Proposition~\ref{prop:A-EV-HS} implies $\|A\|_\calE \leq \|A\|$. A version of the latter result was known in Ref.~\cite{schusterOperatorGrowthOpen2023}.

\section{Operator norms representable as ensemble variance} \label{sec:gen-op-norm-as-evn}

In this Appendix, we investigate the general question: among all possible norms on operators, which can be represented as the ensemble variance norm for some ensemble?

Firstly, the norm must be induced by an inner product. This is equivalent to the norm satisfying the parallelogram law. This rules out norms like the operator norm. Furthermore, it is known that no two different inner products can induce the same norm on a complex vector space~\cite{rudinFunctionalAnalysis1991}. Therefore, if a norm is representable as an ensemble variance norm, it must be induced by an inner product that is representable as an ensemble covariance inner product.

Notice that
\begin{multline}
\langle A,B\rangle_\calE = \expt_{\psi\sim \calE} \langle \psi|A^\dagger|\psi\rangle \langle \psi|B|\psi\rangle \\ = \tr\left[(A^\dagger \otimes B) \left(\expt_{\psi\sim \calE} |\psi\rangle \langle \psi| \otimes |\psi\rangle \langle \psi|\right) \right].
\end{multline}
Define
\begin{equation}
\rho^{(2)}_\calE = \expt_{\psi\sim \calE} \left[ |\psi\rangle \langle \psi| \otimes |\psi\rangle \langle \psi| \right],
\end{equation}
which is an operator on the double Hilbert space. This is also called the second moment of the state ensemble $\calE$~\cite{cotlerEmergentQuantumState2023}. Any inner product on the operator space can also be uniquely associated with a $\rho^{(2)}$. For example, the Hilbert-Schmidt inner product
\begin{equation}
\langle A,B\rangle = \frac{1}{D} \tr [A^\dagger B]
\end{equation}
is associated with
\begin{equation}
\rho^{(2)}_\text{HS} = \frac{1}{D} \sum_{i,j}|i\rangle\langle j| \otimes |j\rangle \langle i| = \frac{1}{D}\text{SWAP}.
\end{equation}
The inner product whose norm measures operator size,
\begin{equation}
\langle A,B\rangle_S = \sum_P a_P^\ast b_P |P|,
\end{equation}
is associated with
\begin{equation}
\rho^{(2)}_S = \frac{1}{D^2} \sum_P |P| P\otimes P.
\end{equation}

\begin{proposition}
A second-moment density matrix $\rho^{(2)}$ can be realized by an ensemble if and only if $\rho^{(2)}$ is positive semidefinite, separable, and supported on the symmetric subspace.
\end{proposition}
\begin{proof}
The ``only if'' part is obvious, since $|\psi\rangle \langle \psi| \otimes |\psi\rangle \langle \psi|$ is positive semidefinite and supported on the symmetric subspace, and $\expt_{\psi\sim\calE}$ is a classical mixture of them.

Conversely, if $\rho^{(2)}$ is positive semidefinite and separable, it can be represented as
\begin{equation}
\rho^{(2)} = \sum_k \lambda_k |\alpha_k\rangle \langle \alpha_k| \otimes |\beta_k\rangle \langle \beta_k|.
\end{equation}
As $\rho^{(2)}$ is positive semidefinite, being supported on the symmetric subspace means that $\langle \psi|\rho^{(2)}|\psi\rangle=0$ for any $|\psi\rangle$ that is anti-symmetric. This translates to
\begin{equation}
 \sum_k \lambda_k |\langle \psi|\alpha_k \otimes \beta_k\rangle|^2 =0,
\end{equation}
which requires that $\langle \psi|\alpha_k\otimes \beta_k\rangle=0$ for any $k$ and any anti-symmetric $|\psi\rangle$. This is only possible if $|\alpha_k\rangle =|\beta_k\rangle$. This proves the result.
\end{proof}

In particular, in qubit systems, we can design a necessary condition for a Pauli-diagonal inner product,
\begin{equation}
\langle P,Q\rangle = f(P) \delta_{P,Q},
\end{equation}
to be representable as an ensemble covariance. The associated second moment is
\begin{equation}
\rho^{(2)} = \frac{1}{D^2}\sum_P f(P) P\otimes P.
\end{equation}
We notice that the operators $P\otimes P$ can be simultaneously diagonalized in the generalized Bell basis~\cite{wernerAllTeleportationDense2001},
\begin{equation}
|\beta_P\rangle = \frac{1}{\sqrt{D}} (P\otimes I) \sum_i |i\rangle \otimes |i\rangle.
\end{equation}
Given the action of $P$ on the computational basis, one can verify that
\begin{equation}
P\otimes P\sum_i |i\rangle \otimes |i\rangle = (-1)^{n_Y(P)} \sum_i |i\rangle \otimes |i\rangle,
\end{equation}
where $n_Y(P)$ is the number of $Y$-operators in $P$. This also implies that
\begin{equation}
\text{SWAP} |\beta_P\rangle = (-1)^{n_Y(P)} |\beta_P\rangle.
\end{equation}
Therefore,
\begin{equation}
P\otimes P |\beta_Q\rangle = (-1)^{n_Y(P)+s(P,Q)} |\beta_Q\rangle,
\end{equation}
with $s$ defined as $PQ=(-1)^{s(P,Q)} QP$. Therefore,
\begin{equation}
\sum_P f(P)P\otimes P = \sum_Q g(Q) |\beta_Q\rangle \langle \beta_Q|,
\end{equation}
where
\begin{equation}
g(Q) = \sum_P (-1)^{n_Y(P)+s(P,Q)} f(P).
\end{equation}
Therefore, for the inner product to be representable by an ensemble, it is necessary (not sufficient) that:
\begin{enumerate}
    \item $g(Q) \geq 0$ for all $Q$,
    \item $g(Q)=0$ whenever $n_Y(Q)$ is odd.
\end{enumerate}

We will verify that two known inner products, namely the Haar and RPS inner products, satisfy this condition.

For the Haar inner product, $f(I)=1$, and $f(P)=(2^L+1)^{-1}$ for all other Pauli strings. In this case,
\begin{equation}
g(Q) = \frac{2^L + \sum_P (-1)^{n_Y(P)+s(P,Q)}}{2^L+1}.
\end{equation}
We find that $(-1)^{n_Y(P)+s(P,Q)}$ admits a multiplicative decomposition. Therefore,
\begin{equation}
\sum_P (-1)^{n_Y(P)+s(P,Q)} = \prod_{i=1}^L \sum_{P_i\in \{I,X,Y,Z\}} (-1)^{\delta_{P_i,Y} + s(P_i,Q_i)}.
\end{equation}
The single-site sum can be easily evaluated as
\begin{equation}
\sum_{P_i\in \{I,X,Y,Z\}} (-1)^{\delta_{P_i,Y} + s(P_i,Q_i)} = \begin{cases}
2, & Q_i \in \{I,X,Z\}; \\
-2, & Q_i = Y.
\end{cases}
\end{equation}
Therefore,
\begin{equation}
g(Q) = \frac{2^L}{2^L+1} \left(1+(-1)^{n_Y(Q)}\right).
\end{equation}
This indeed satisfies the two conditions.

For the RPS inner product, $f(P) = 3^{-|P|}$. We find that
\begin{equation}
g(Q) = \prod_{i=1}^L \sum_{P_i\in \{I,X,Y,Z\}} (-1)^{\delta_{P_i,Y} + s(P_i,Q_i)} 3^{-1+\delta_{P_i,I}},
\end{equation}
where the sum is
\begin{multline}
\sum_{P_i\in \{I,X,Y,Z\}} (-1)^{\delta_{P_i,Y} + s(P_i,Q_i)} 3^{-1+\delta_{P_i,I}}  \\= \begin{cases}
\frac{4}{3}, & Q_i\in \{I,X,Z\}; \\
0, & Q_i = Y.
\end{cases}
\end{multline}
Therefore,
\begin{equation}
g(Q) = \begin{cases}
\left(\frac{4}{3}\right)^L, & n_Y(Q) = 0; \\
0, & n_Y(Q) > 0.
\end{cases}
\end{equation}
This also satisfies the two conditions.

By contrast, if we consider $f(P)=|P|$, we get
\begin{multline}
g(Q) = \sum_{\{P_i\}} \left(\sum_i |P_i|\right) \prod_{i=1}^L (-1)^{\delta_{P_i,Y} + s(P_i,Q_i)} \\
= \left. \frac{\partial}{\partial t}  \prod_{i=1}^L \sum_{P_i\in \{I,X,Y,Z\}} (-1)^{\delta_{P_i,Y} + s(P_i,Q_i)} e^{t|P_i|} \right|_{t=0} \\
= \left. \frac{\partial}{\partial t}  \left(1+e^t\right)^{L-n_Y(Q)} \left(1-3e^t\right)^{n_Y(Q)} \right|_{t=0}  \\
= (-1)^{n_Y(Q)}2^{L-1} (L+2n_Y(Q)).
\end{multline}
Obviously, $g(Q)$ does not vanish for odd $n_Y(Q)$. This implies that the associated second moment is not supported on the symmetric subspace, and therefore cannot be represented as an ensemble variance norm.

\section{Ensemble variance norm for random matrix product states}
\label{sec:RMPS}

We derive the ensemble variance norm for random matrix product states of a fixed bond dimension. Our derivation is based on the averaging of Haar unitary ensemble matrices~\cite{collinsIntegrationRespectHaar2006}, which has been extensively used in the study of random unitary circuits~\cite{skinnerMeasurementInducedPhaseTransitions2019,nahumQuantumEntanglementGrowth2017a,nahumOperatorSpreadingRandom2018d,fisherRandomQuantumCircuits2023} and similar studies of random matrix product states~\cite{garneroneStatisticalPropertiesRandom2010,garneroneTypicalityRandomMatrix2010,garneroneGeneralizedQuantumMicrocanonical2013,rolandiExtensiveRenyiEntropies2020,haferkampEmergentStatisticalMechanics2021,chenMagicRandomMatrix2024,dowlingFreeIndependenceUnitary2025}.

Let us consider an MPS under periodic boundary conditions of the following form,
\begin{equation}
|\psi\rangle = \begin{tikzpicture}[
    tensor/.style={rectangle, draw, thick, minimum size=0.3cm},
    baseline=(A1.center),
    scale=0.7,
    every node/.style={font=\normalsize}
]

\node[tensor] (A1) at (0.5, 0) {$A_1$};
\node[tensor] (A2) at (2.5, 0) {$A_2$};
\node[tensor] (An) at (6.0, 0) {$A_L$};

\draw[thick] (A1.east) -- (A2.west) node[midway, above, yshift=0.1cm] {$s_1$};
\draw[thick] (A2.east) -- ++(0.5,0) node[above,yshift=0.1cm] {$s_2$};
\node[right=0.5cm of A2] (dots1) {$\dots$};
\draw[thick] (An.west) -- ++(-0.5,0) node[above,yshift=0.1cm,xshift=-0.1cm] {$s_{L-1}$};
\draw[thick] (A1.west) -- ++(-0.5,0) node[left] {$s_0$};
\draw[thick] (An.east) -- ++(0.5,0) node[right] (sn) {$s_L$}
;
\node[above=0.4cm of sn] {$s_0$};
\node[above=0.2cm of sn, xshift=0.2cm, rotate=90] {$=$};

\draw[thick] (A1.north) -- ++(0,0.5) node[above] {$i_1$};
\draw[thick] (A2.north) -- ++(0,0.5) node[above] {$i_2$};
\draw[thick] (An.north) -- ++(0,0.5) node[above] {$i_L$};
\node[right=0.5cm of A2,yshift=0.9cm] (dots2) {$\dots$};
\end{tikzpicture}.
\end{equation}
Assume that the bond indices $s_k$ have dimension $\chi$, and the site indices $i_k$ have dimension $d$. We also assume that the blocks are in normal form, with
\begin{equation}
\sum_{i_n,s_n} (A_n)_{s_{n-1} i_n s_n} (A_n)_{s_{n-1}^\prime i_n s_n}^\ast = \delta_{s_{n-1},s_{n-1}^\prime}.
\end{equation}
Equivalently, we can envision $A_n$ as part of the unitary matrix $U_n$, whose size is $q \times q$ with $q=d\chi$, where the index can be written as $(i,s)$ with $i=1,\dots,d$ and $s=1,\dots,\chi$, and let $(A_n)_{s_{n-1} i_n s_n} = (U_n)_{(1,s_{n-1}),(i_n,s_n)}$. An MPS that has this form is called unitarily embeddable. We consider the set of unitarily embeddable RMPSs, where the randomness is characterized by individually sampling each $U_n$ from the Haar ensemble. Consider the formulae~\cite{collinsIntegrationRespectHaar2006}
\begin{equation}
\overline{U_{ij} U^\ast_{i^\prime, j^\prime}} = \frac{1}{q} \delta_{ii^\prime}\delta_{jj^\prime},
\end{equation}
and
\begin{align}
&\overline{U_{i_1j_1} U_{i_2j_2} U^\ast_{i_1^\prime, j_1^\prime} U^\ast_{i_2^\prime, j_2^\prime}} \nonumber \\ 
= & \frac{1}{q^2-1} \bigg[
\delta_{i_1i_1^\prime}\delta_{i_2i_2^\prime}\delta_{j_1j_1^\prime} \delta_{j_2j_2^\prime}
+ \delta_{i_1i_2^\prime}\delta_{i_2i_1^\prime}\delta_{j_1j_2^\prime} \delta_{j_2j_1^\prime} \nonumber \\
&
- \frac{1}{q} \delta_{i_1i_1^\prime}\delta_{i_2i_2^\prime} \delta_{j_1j_2^\prime} \delta_{j_2j_1^\prime}
- \frac{1}{q} \delta_{i_1i_2^\prime}\delta_{i_2i_1^\prime} \delta_{j_1j_1^\prime} \delta_{j_2j_2^\prime} \bigg].
\end{align}
This gives 
\begin{equation}
\begin{tikzpicture}[
    tensor/.style={rectangle, draw, thick, minimum size=0.3cm},
    baseline=(current bounding box.center),
    scale=0.7,
    every node/.style={font=\normalsize}
]
\node[tensor] (A) at (0, -1.7) {$A$};
\node[tensor] (As) at (0, 1.7) {$A^\ast$};

\draw[thick] (A.west) -- ++(-0.5,0) node[left] {$r$};
\draw[thick] (A.east) -- ++(0.5,0) node[right] {$s$};
\draw[thick] (As.west) -- ++(-0.5,0) node[left] {$r^\prime$};
\draw[thick] (As.east) -- ++(0.5,0) node[right] {$s^\prime$};
\draw[thick] (A.north) -- ++(0,0.5) node[above] {$i$};
\draw[thick] (As.south) -- ++(0,-0.5) node[below] {$i^\prime$};

\end{tikzpicture}
=\frac{1}{d\chi}
\begin{tikzpicture}[
    tensor/.style={rectangle, draw, thick, minimum size=0.3cm},
    baseline=(current bounding box.center),
    scale=0.7,
    every node/.style={font=\normalsize}
]
\node (r) at (-2, -1.7) {$r$};
\node (s) at (2, -1.7) {$s$};
\node (rp) at (-2, 1.7) {$r^\prime$};
\node (sp) at (2, 1.7) {$s^\prime$};
\node (i) at (0, -1) {$i$};
\node (ip) at (0, 1) {$i^\prime$};
\draw[thick] (r) to[in=0, out=0] (rp);
\draw[thick] (s) to[in=180, out=180] (sp);
\draw[thick] (i) to[in=-90, out=90] (ip);
\end{tikzpicture}.
\end{equation}
Average over $A$ ($U$) is understood on the left-hand side. Therefore,
\begin{equation}
\expt_{\psi\sim\chi\text{-RMPS}} \left[ |\psi\rangle \langle \psi| \right] = \prod_{n} \frac{1}{d} \delta_{i_ni_n^\prime} = \frac{I}{D}.
\end{equation}
The second moment is more complicated, with
\begin{widetext}
\begin{align}
\begin{tikzpicture}[
    tensor/.style={rectangle, draw, thick, minimum size=0.3cm},
    baseline=(current bounding box.center),
    scale=0.7,
    every node/.style={font=\normalsize}
]
\node[tensor] (A) at (0, -1.7) {$A$};
\node[tensor] (As) at (0, 1.7) {$A^\ast$};
\draw[thick] (A.west) -- ++(-0.5,0) node[left] {$r_1$};
\draw[thick] (A.east) -- ++(0.5,0) node[right] {$s_1$};
\draw[thick] (As.west) -- ++(-0.5,0) node[left] {$r_1^\prime$};
\draw[thick] (As.east) -- ++(0.5,0) node[right] {$s_1^\prime$};
\draw[thick] (A.north) -- ++(0,0.5) node[above] {$i_1$};
\draw[thick] (As.south) -- ++(0,-0.5) node[below] {$i_1^\prime$};
\end{tikzpicture}
\begin{tikzpicture}[
    tensor/.style={rectangle, draw, thick, minimum size=0.3cm},
    baseline=(current bounding box.center),
    scale=0.7,
    every node/.style={font=\normalsize}
]
\node[tensor] (A) at (0, -1.7) {$A$};
\node[tensor] (As) at (0, 1.7) {$A^\ast$};
\draw[thick] (A.west) -- ++(-0.5,0) node[left] {$r_2$};
\draw[thick] (A.east) -- ++(0.5,0) node[right] {$s_2$};
\draw[thick] (As.west) -- ++(-0.5,0) node[left] {$r_2^\prime$};
\draw[thick] (As.east) -- ++(0.5,0) node[right] {$s_2^\prime$};
\draw[thick] (A.north) -- ++(0,0.5) node[above] {$i_2$};
\draw[thick] (As.south) -- ++(0,-0.5) node[below] {$i_2^\prime$};
\end{tikzpicture}
= & \frac{1}{q^2-1} \left[
\begin{tikzpicture}[
    baseline=(current bounding box.center),
    scale=0.7,
    every node/.style={font=\normalsize}
]
\node (i1p) at (0, 1) {$i_1^\prime$};
\node (i1) at (0, -1) {$i_1$};
\node (i2p) at (0.5, 1) {$i_2^\prime$};
\node (i2) at (0.5, -1) {$i_2$};
\draw[thick] (i1) to (i1p);
\draw[thick] (i2) to (i2p);
\end{tikzpicture}
\begin{tikzpicture}[
    baseline=(current bounding box.center),
    scale=0.7,
    every node/.style={font=\normalsize}
]
\node (r1p) at (0, 1.5) {$r_1^\prime$};
\node (r2p) at (0, -0.5) {$r_2^\prime$};
\node (r1) at (0, 0.5) {$r_1$};
\node (r2) at (0, -1.5) {$r_2$};
\draw[thick] (r1p) to[in=0, out=0] (r1);
\draw[thick] (r2p) to[in=0, out=0] (r2);
\end{tikzpicture}
\begin{tikzpicture}[
    baseline=(current bounding box.center),
    scale=0.7,
    every node/.style={font=\normalsize}
]
\node (s1p) at (0, 1.5) {$s_1^\prime$};
\node (s2p) at (0, -0.5) {$s_2^\prime$};
\node (s1) at (0, 0.5) {$s_1$};
\node (s2) at (0, -1.5) {$s_2$};
\draw[thick] (s1p) to[in=180, out=180] (s1);
\draw[thick] (s2p) to[in=180, out=180] (s2);
\end{tikzpicture}
+
\begin{tikzpicture}[
    baseline=(current bounding box.center),
    scale=0.7,
    every node/.style={font=\normalsize}
]
\node (i1p) at (0, 1) {$i_1^\prime$};
\node (i1) at (0, -1) {$i_1$};
\node (i2p) at (0.5, 1) {$i_2^\prime$};
\node (i2) at (0.5, -1) {$i_2$};
\draw[thick] (i1) to (i2p);
\draw[thick] (i2) to (i1p);
\end{tikzpicture}
\begin{tikzpicture}[
    baseline=(current bounding box.center),
    scale=0.7,
    every node/.style={font=\normalsize}
]
\node (r1p) at (0, 1.5) {$r_1^\prime$};
\node (r2p) at (0, -0.5) {$r_2^\prime$};
\node (r1) at (0, 0.5) {$r_1$};
\node (r2) at (0, -1.5) {$r_2$};
\draw[thick] (r1p) to[in=0, out=0] (r2);
\draw[thick] (r2p) to[in=0, out=0] (r1);
\end{tikzpicture}
\begin{tikzpicture}[
    baseline=(current bounding box.center),
    scale=0.7,
    every node/.style={font=\normalsize}
]
\node (s1p) at (0, 1.5) {$s_1^\prime$};
\node (s2p) at (0, -0.5) {$s_2^\prime$};
\node (s1) at (0, 0.5) {$s_1$};
\node (s2) at (0, -1.5) {$s_2$};
\draw[thick] (s1p) to[in=180, out=180] (s2);
\draw[thick] (s2p) to[in=180, out=180] (s1);
\end{tikzpicture}
\right] \nonumber \\
&-\frac{1}{q(q^2-1)}\left[
\begin{tikzpicture}[
    baseline=(current bounding box.center),
    scale=0.7,
    every node/.style={font=\normalsize}
]
\node (i1p) at (0, 1) {$i_1^\prime$};
\node (i1) at (0, -1) {$i_1$};
\node (i2p) at (0.5, 1) {$i_2^\prime$};
\node (i2) at (0.5, -1) {$i_2$};
\draw[thick] (i1) to (i1p);
\draw[thick] (i2) to (i2p);
\end{tikzpicture}
\begin{tikzpicture}[
    baseline=(current bounding box.center),
    scale=0.7,
    every node/.style={font=\normalsize}
]
\node (r1p) at (0, 1.5) {$r_1^\prime$};
\node (r2p) at (0, -0.5) {$r_2^\prime$};
\node (r1) at (0, 0.5) {$r_1$};
\node (r2) at (0, -1.5) {$r_2$};
\draw[thick] (r1p) to[in=0, out=0] (r2);
\draw[thick] (r2p) to[in=0, out=0] (r1);
\end{tikzpicture}
\begin{tikzpicture}[
    baseline=(current bounding box.center),
    scale=0.7,
    every node/.style={font=\normalsize}
]
\node (s1p) at (0, 1.5) {$s_1^\prime$};
\node (s2p) at (0, -0.5) {$s_2^\prime$};
\node (s1) at (0, 0.5) {$s_1$};
\node (s2) at (0, -1.5) {$s_2$};
\draw[thick] (s1p) to[in=180, out=180] (s1);
\draw[thick] (s2p) to[in=180, out=180] (s2);
\end{tikzpicture}
+
\begin{tikzpicture}[
    baseline=(current bounding box.center),
    scale=0.7,
    every node/.style={font=\normalsize}
]
\node (i1p) at (0, 1) {$i_1^\prime$};
\node (i1) at (0, -1) {$i_1$};
\node (i2p) at (0.5, 1) {$i_2^\prime$};
\node (i2) at (0.5, -1) {$i_2$};
\draw[thick] (i1) to (i2p);
\draw[thick] (i2) to (i1p);
\end{tikzpicture}
\begin{tikzpicture}[
    baseline=(current bounding box.center),
    scale=0.7,
    every node/.style={font=\normalsize}
]
\node (r1p) at (0, 1.5) {$r_1^\prime$};
\node (r2p) at (0, -0.5) {$r_2^\prime$};
\node (r1) at (0, 0.5) {$r_1$};
\node (r2) at (0, -1.5) {$r_2$};
\draw[thick] (r1p) to[in=0, out=0] (r1);
\draw[thick] (r2p) to[in=0, out=0] (r2);
\end{tikzpicture}
\begin{tikzpicture}[
    baseline=(current bounding box.center),
    scale=0.7,
    every node/.style={font=\normalsize}
]
\node (s1p) at (0, 1.5) {$s_1^\prime$};
\node (s2p) at (0, -0.5) {$s_2^\prime$};
\node (s1) at (0, 0.5) {$s_1$};
\node (s2) at (0, -1.5) {$s_2$};
\draw[thick] (s1p) to[in=180, out=180] (s2);
\draw[thick] (s2p) to[in=180, out=180] (s1);
\end{tikzpicture}
\right]. \label{eq:A-second-moment}
\end{align}
\end{widetext}
We can see that the configuration with respect to $r$ or $s$ has only two possibilities: it is either one of
\begin{equation}
|\sigma\rangle = \begin{tikzpicture}[
    baseline=(current bounding box.center),
    scale=0.7,
    every node/.style={font=\normalsize}
]
\node (s1p) at (0, 1.5) {};
\node (s2p) at (0, -0.5) {};
\node (s1) at (0, 0.5) {};
\node (s2) at (0, -1.5) {};
\draw[thick] (s1p) to[in=0, out=0] (s1);
\draw[thick] (s2p) to[in=0, out=0] (s2);
\end{tikzpicture}, \quad
|\tau\rangle = \begin{tikzpicture}[
    baseline=(current bounding box.center),
    scale=0.7,
    every node/.style={font=\normalsize}
]
\node (s1p) at (0, 1.5) {};
\node (s2p) at (0, -0.5) {};
\node (s1) at (0, 0.5) {};
\node (s2) at (0, -1.5) {};
\draw[thick] (s1p) to[in=0, out=0] (s2);
\draw[thick] (s2p) to[in=0, out=0] (s1);
\end{tikzpicture}. \label{eq:sigma-and-tau}
\end{equation}
We note that when contracting on neighboring sites, we have
\begin{equation}
\langle \sigma | \sigma \rangle = \langle \tau | \tau \rangle  = \chi^2, 
\end{equation}
\begin{equation}
\langle \sigma | \tau \rangle  = \chi.
\end{equation}
We can construct a two-dimensional representation,
\begin{equation}
|\sigma \rangle  = \left(\sqrt{\frac{\chi(\chi+1)}{2}}, \sqrt{\frac{\chi(\chi-1)}{2}}\right)^T, \label{eq:transfer-sigma}
\end{equation}
\begin{equation}
|\tau\rangle  = \left(\sqrt{\frac{\chi(\chi+1)}{2}}, -\sqrt{\frac{\chi(\chi-1)}{2}}\right)^T. \label{eq:transfer-tau}
\end{equation}
We let kets denote the configuration on $r$, and bras denote the configuration on $s$, then, Eq.~\eqref{eq:A-second-moment} can be simplified as
\begin{widetext}
\begin{equation}
\begin{tikzpicture}[
    tensor/.style={rectangle, draw, thick, minimum size=0.3cm},
    baseline=(current bounding box.center),
    scale=0.7,
    every node/.style={font=\normalsize}
]
\node[tensor] (A) at (0, -1.7) {$A$};
\node[tensor] (As) at (0, 1.7) {$A^\ast$};
\draw[thick] (A.west) -- ++(-0.5,0) node[left] {$r_1$};
\draw[thick] (A.east) -- ++(0.5,0) node[right] {$s_1$};
\draw[thick] (As.west) -- ++(-0.5,0) node[left] {$r_1^\prime$};
\draw[thick] (As.east) -- ++(0.5,0) node[right] {$s_1^\prime$};
\draw[thick] (A.north) -- ++(0,0.5) node[above] {$i_1$};
\draw[thick] (As.south) -- ++(0,-0.5) node[below] {$i_1^\prime$};
\end{tikzpicture}
\begin{tikzpicture}[
    tensor/.style={rectangle, draw, thick, minimum size=0.3cm},
    baseline=(current bounding box.center),
    scale=0.7,
    every node/.style={font=\normalsize}
]
\node[tensor] (A) at (0, -1.7) {$A$};
\node[tensor] (As) at (0, 1.7) {$A^\ast$};
\draw[thick] (A.west) -- ++(-0.5,0) node[left] {$r_2$};
\draw[thick] (A.east) -- ++(0.5,0) node[right] {$s_2$};
\draw[thick] (As.west) -- ++(-0.5,0) node[left] {$r_2^\prime$};
\draw[thick] (As.east) -- ++(0.5,0) node[right] {$s_2^\prime$};
\draw[thick] (A.north) -- ++(0,0.5) node[above] {$i_2$};
\draw[thick] (As.south) -- ++(0,-0.5) node[below] {$i_2^\prime$};
\end{tikzpicture}
 =\frac{1}{q^2-1} \left[
\begin{tikzpicture}[
    baseline=(current bounding box.center),
    scale=0.7,
    every node/.style={font=\normalsize}
]
\node (i1p) at (0, 1) {$i_1^\prime$};
\node (i1) at (0, -1) {$i_1$};
\node (i2p) at (0.5, 1) {$i_2^\prime$};
\node (i2) at (0.5, -1) {$i_2$};
\draw[thick] (i1) to (i1p);
\draw[thick] (i2) to (i2p);
\end{tikzpicture}
|\sigma\rangle \langle \sigma|
+
\begin{tikzpicture}[
    baseline=(current bounding box.center),
    scale=0.7,
    every node/.style={font=\normalsize}
]
\node (i1p) at (0, 1) {$i_1^\prime$};
\node (i1) at (0, -1) {$i_1$};
\node (i2p) at (0.5, 1) {$i_2^\prime$};
\node (i2) at (0.5, -1) {$i_2$};
\draw[thick] (i1) to (i2p);
\draw[thick] (i2) to (i1p);
\end{tikzpicture}
|\tau\rangle \langle \tau|
\right] -\frac{1}{q(q^2-1)}\left[
\begin{tikzpicture}[
    baseline=(current bounding box.center),
    scale=0.7,
    every node/.style={font=\normalsize}
]
\node (i1p) at (0, 1) {$i_1^\prime$};
\node (i1) at (0, -1) {$i_1$};
\node (i2p) at (0.5, 1) {$i_2^\prime$};
\node (i2) at (0.5, -1) {$i_2$};
\draw[thick] (i1) to (i1p);
\draw[thick] (i2) to (i2p);
\end{tikzpicture}
|\tau\rangle \langle \sigma|
+
\begin{tikzpicture}[
    baseline=(current bounding box.center),
    scale=0.7,
    every node/.style={font=\normalsize}
]
\node (i1p) at (0, 1) {$i_1^\prime$};
\node (i1) at (0, -1) {$i_1$};
\node (i2p) at (0.5, 1) {$i_2^\prime$};
\node (i2) at (0.5, -1) {$i_2$};
\draw[thick] (i1) to (i2p);
\draw[thick] (i2) to (i1p);
\end{tikzpicture}
|\sigma\rangle \langle \tau|
\right].
\end{equation}
\end{widetext}
We can also write down the 2-by-2 matrix representations of the $|\sigma\rangle$ and $|\tau\rangle$ involved, to get
\begin{equation}
|\sigma \rangle \langle \sigma | = \begin{pmatrix}\frac{\chi(\chi+1)}{2} & \frac{\chi\sqrt{\chi^2-1}}{2} \\ \frac{\chi\sqrt{\chi^2-1}}{2} &  \frac{\chi(\chi-1)}{2}\end{pmatrix},
\end{equation}
\begin{equation}
|\tau \rangle \langle \tau | = \begin{pmatrix}\frac{\chi(\chi+1)}{2} & -\frac{\chi\sqrt{\chi^2-1}}{2} \\ -\frac{\chi\sqrt{\chi^2-1}}{2} &  \frac{\chi(\chi-1)}{2}\end{pmatrix},
\end{equation}
\begin{equation}
|\tau \rangle \langle \sigma | = \begin{pmatrix} \frac{\chi(\chi+1)}{2} &\frac{\chi\sqrt{\chi^2-1}}{2} \\ -\frac{\chi\sqrt{\chi^2-1}}{2} &-\frac{\chi(\chi-1)}{2} \end{pmatrix},
\end{equation}
\begin{equation}
|\sigma\rangle \langle \tau| = \left(|\tau\rangle \langle \sigma|\right)^T.
\end{equation}

Now, consider stacking two Pauli strings into the middle. We notice that
\begin{widetext}
\noindent
\begin{minipage}{.5\textwidth}
\begin{equation}
\begin{tikzpicture}[
    tensor/.style={rectangle, draw, thick, minimum size=0.3cm},
    baseline=(current bounding box.center),
    scale=0.7,
    every node/.style={font=\normalsize}
]
\node[tensor,
label={[label distance=3mm]north:$i_1^\prime$},
label={[label distance=3mm]south:$i_1$}] (l1) at (0,0) {$\Lambda_a$};
\coordinate (l1n) at (0.75,0) {};
\draw[thick] (l1.north) to[out=90,in=90,looseness=2] (l1n);
\draw[thick] (l1.south) to[out=-90,in=-90,looseness=2] (l1n);
\node[tensor,
label={[label distance=3mm]north:$i_2^\prime$},
label={[label distance=3mm]south:$i_2$}] (l2) at (1.5,0) {$\Lambda_b$};
\coordinate (l2n) at (2.25,0) {};
\draw[thick] (l2.north) to[out=90,in=90,looseness=2] (l2n);
\draw[thick] (l2.south) to[out=-90,in=-90,looseness=2] (l2n);
\end{tikzpicture}
= \tr \Lambda_a \tr \Lambda_b = d^2 \delta_{a,1}\delta_{b,1}, \label{eq:RMPS-site-1} 
\end{equation}
\end{minipage}
\begin{minipage}{.5\textwidth}
\begin{equation}
\begin{tikzpicture}[
    tensor/.style={rectangle, draw, thick, minimum size=0.3cm},
    baseline=(current bounding box.center),
    scale=0.7,
    every node/.style={font=\normalsize}
]
\node[tensor,
label={[label distance=3mm]north:$i_1^\prime$},
label={[label distance=3mm]south:$i_1$}] (l1) at (0,0) {$\Lambda_a$};
\node[tensor,
label={[label distance=3mm]north:$i_2^\prime$},
label={[label distance=3mm]south:$i_2$}] (l2) at (1.5,0) {$\Lambda_b$};
\draw[thick] (l1.north) to[out=90,in=-90,looseness=2] (l2.south);
\draw[thick] (l2.north) to[out=90,in=-90,looseness=2] (l1.south);
\end{tikzpicture}
= \tr \Lambda_a \Lambda_b = d \delta_{a,b}. \label{eq:RMPS-site-2}
\end{equation}
\end{minipage}

From this, we immediately see that the inner product of two different Pauli strings is zero, since Eqs.~\eqref{eq:RMPS-site-1} and \eqref{eq:RMPS-site-2} both vanish unless $a=b$. We would only have to consider two cases: $a=b=1$ and $a=b\neq 1$. We would have, respectively,
\begin{equation}
\begin{tikzpicture}[
    tensor/.style={rectangle, draw, thick, minimum size=0.3cm},
    baseline=(current bounding box.center),
    scale=0.7,
    every node/.style={font=\normalsize}
]
\node[tensor] (A) at (0, -1.5) {$A$};
\node[tensor] (As) at (0, 1.5) {$A^\ast$};
\node[tensor] (L) at (0,0) {$I$};
\draw[thick] (A.west) -- ++(-0.5,0) node[left] {$r_1$};
\draw[thick] (A.east) -- ++(0.5,0) node[right] {$s_1$};
\draw[thick] (As.west) -- ++(-0.5,0) node[left] {$r_1^\prime$};
\draw[thick] (As.east) -- ++(0.5,0) node[right] {$s_1^\prime$};
\draw[thick] (A.north) -- (L.south);
\draw[thick] (As.south) -- (L.north);
\end{tikzpicture}
\begin{tikzpicture}[
    tensor/.style={rectangle, draw, thick, minimum size=0.3cm},
    baseline=(current bounding box.center),
    scale=0.7,
    every node/.style={font=\normalsize}
]
\node[tensor] (A) at (0, -1.5) {$A$};
\node[tensor] (As) at (0, 1.5) {$A^\ast$};
\node[tensor] (L) at (0,0) {$I$};
\draw[thick] (A.west) -- ++(-0.5,0) node[left] {$r_2$};
\draw[thick] (A.east) -- ++(0.5,0) node[right] {$s_2$};
\draw[thick] (As.west) -- ++(-0.5,0) node[left] {$r_2^\prime$};
\draw[thick] (As.east) -- ++(0.5,0) node[right] {$s_2^\prime$};
\draw[thick] (A.north) -- (L.south);
\draw[thick] (As.south) -- (L.north);
\end{tikzpicture}
 =
 \begin{pmatrix}
     \cfrac{(d+1) (\chi+1)}{2 (\chi d+1)} & \cfrac{(d-1) \sqrt{\chi^2-1}}{2 (\chi d+1)} \\
 \cfrac{(d-1) \sqrt{\chi^2-1}}{2 (\chi d-1)} & \cfrac{(d+1) (\chi-1)}{2 (\chi d-1)}
 \end{pmatrix} =: B^{11}, \label{eq:RMPS-MPO-1}
 \end{equation}
\begin{equation}
\begin{tikzpicture}[
    tensor/.style={rectangle, draw, thick, minimum size=0.3cm},
    baseline=(current bounding box.center),
    scale=0.7,
    every node/.style={font=\normalsize}
]
\node[tensor] (A) at (0, -1.5) {$A$};
\node[tensor] (As) at (0, 1.5) {$A^\ast$};
\node[tensor] (L) at (0,0) {$\Lambda_a$};
\draw[thick] (A.west) -- ++(-0.5,0) node[left] {$r_1$};
\draw[thick] (A.east) -- ++(0.5,0) node[right] {$s_1$};
\draw[thick] (As.west) -- ++(-0.5,0) node[left] {$r_1^\prime$};
\draw[thick] (As.east) -- ++(0.5,0) node[right] {$s_1^\prime$};
\draw[thick] (A.north) -- (L.south);
\draw[thick] (As.south) -- (L.north);
\end{tikzpicture}
\begin{tikzpicture}[
    tensor/.style={rectangle, draw, thick, minimum size=0.3cm},
    baseline=(current bounding box.center),
    scale=0.7,
    every node/.style={font=\normalsize}
]
\node[tensor] (A) at (0, -1.5) {$A$};
\node[tensor] (As) at (0, 1.5) {$A^\ast$};
\node[tensor] (L) at (0,0) {$\Lambda_a$};
\draw[thick] (A.west) -- ++(-0.5,0) node[left] {$r_2$};
\draw[thick] (A.east) -- ++(0.5,0) node[right] {$s_2$};
\draw[thick] (As.west) -- ++(-0.5,0) node[left] {$r_2^\prime$};
\draw[thick] (As.east) -- ++(0.5,0) node[right] {$s_2^\prime$};
\draw[thick] (A.north) -- (L.south);
\draw[thick] (As.south) -- (L.north);
\end{tikzpicture}
 = \begin{pmatrix}
     \cfrac{\chi +1}{2(\chi d+1)} & -\cfrac{\sqrt{\chi ^2-1}}{2(\chi d+1)} \\
 -\cfrac{\sqrt{\chi ^2-1}}{2(\chi d-1)} & \cfrac{\chi -1}{2(\chi d-1)}
 \end{pmatrix} =: B^{aa},\quad a \neq 1. \label{eq:RMPS-MPO-2}
\end{equation}
\end{widetext}
We have treated the left-hand sides of Eqs.~\eqref{eq:RMPS-MPO-1} and \eqref{eq:RMPS-MPO-2} as transfer matrices from $(s_1,s_2,s_1^\prime,s_2^\prime)$ to $(r_1,r_2,r_1^\prime,r_2^\prime)$, where the $s$- and $r$-space are encoded with a two-dimensional representation given by Eqs.~\eqref{eq:sigma-and-tau}, \eqref{eq:transfer-sigma}, and \eqref{eq:transfer-tau}. Combining all the results above, we get the MPO representation desired,
\begin{multline}
\langle \Lambda_{a_1}\dots \Lambda_{a_L}, \Lambda_{b_1}\dots \Lambda_{b_L}\rangle_{\chi\text{-RMPS}} \\
= \mathrm{tr} \left[ B^{a_1b_1} B^{a_2b_2}\dots B^{a_Lb_L}\right]. \label{eq:B-sum}
\end{multline}
$B^{ab}$ for $a\neq b$ are understood as zero.

It would be convenient to make a change of basis such that $B^{11}$ is diagonal. The result is
\begin{align}
B^{11} & = \begin{pmatrix} 1 & \\ & \frac{d(\chi^2-1)}{\chi^2d^2 -1} \end{pmatrix}, \\
B^{aa} & = \begin{pmatrix} \frac{1}{\chi^2 d+1} & \frac{d\chi(\chi-1)}{(\chi d+1)(\chi^2 d+1)} \\ \frac{d\chi(\chi+1)}{(\chi d-1)(\chi^2 d+1)} & \frac{d^2 \chi^2 (\chi^2-1)}{(\chi^2 d+1)(\chi^2d^2 -1)} \end{pmatrix},\quad a\neq 1.
\end{align}
It is notable that $B^{aa}$ is rank-1, with
\begin{equation}
B^{aa} = \frac{
\begin{pmatrix}d\chi -1 \\ d\chi (\chi+1)\end{pmatrix}
\begin{pmatrix}d\chi +1 & d\chi(\chi-1)\end{pmatrix}
}{(d\chi^2+1)(d^2\chi^2-1)} .
\end{equation}
This means that
\begin{equation}
\left(B^{aa}\right)^n = \left(\frac{d\chi^2-1}{d^2\chi^2-1}\right)^{n-1} B^{aa}.
\end{equation}
Therefore, consider an operator $P$ of the form
\begin{equation}
P = I^{\otimes m_1} \Lambda^{\otimes n_1} I^{\otimes m_2} \Lambda ^{\otimes n_2}\dots I^{\otimes m_k} \Lambda ^{\otimes n_k},
\end{equation}
where $\Lambda$ are any non-identity Paulis, then
\begin{align}
\|P\|&_{\chi\text{-RMPS}}^2 = \left(\frac{d \chi^2-1}{d^2 \chi^2-1}\right)^{\sum_{j=1}^k n_j}  \nonumber \\ & \times \prod_{j=1}^{k} \frac{d^2\chi^2(\chi^2-1) \left(\frac{d(\chi^2-1)}{d^2 \chi^2-1}\right)^{m_j}+d^2\chi^2-1}{d^2\chi^4 - 1}.
\end{align}
We see that $\sum_{j=1}^k n_j$ is the Pauli weight $|P|$. Therefore, the first part of the expression is an exponential penalization of $|P|$, similar to the RPS-norm. The second part, however, makes the RMPS-norm dependent on additional information, namely $\{m_j\}$. It is easy to show that the factors in the second part are all less than unity, implying that
\begin{equation}
\|P\|_{\chi\text{-RMPS}}^2 \leq \left(\frac{d \chi^2-1}{d^2 \chi^2-1}\right)^{|P|}. \label{eq:RMPS-norm-upper-bound}
\end{equation}

More specifically, as $\sum_j m_j$ is fixed, from a convexity argument, we can see that given $|P|$, $\|P\|_{\chi\text{-RMPS}}^2$ is maximized when $m_j$ is consolidated in one $m_1=L-|P|$, in which case
\begin{align}
\|P\|&_{\chi\text{-RMPS}}^2 =  \left(\frac{d \chi^2-1}{d^2 \chi^2-1}\right)^{|P|} \nonumber \\ & \times  \frac{d^2\chi^2(\chi^2-1) \left(\frac{d(\chi^2-1)}{d^2 \chi^2-1}\right)^{L-|P|}+d^2\chi^2-1}{d^2\chi^4 - 1}.
\end{align}
In the limit where $L\gg |P|$, this becomes
\begin{equation}
\|P\|_{\chi\text{-RMPS}}^2 = \frac{1}{d\chi^2+1}\left(\frac{d \chi^2-1}{d^2 \chi^2-1}\right)^{|P|-1}.
\end{equation}
Note that this expression assumed $P\neq I$. This implies that the upper bound Eq.~\eqref{eq:RMPS-norm-upper-bound} is asymptotically tight. On the other hand, the multiplicative factor becomes small when $m_j$ is evenly distributed. In the worst case, assume that $k=|P|$ and $m_j\gg 1$ for each $j$, then,
\begin{equation}
\|P\|_{\chi\text{-RMPS}}^2 \approx (d\chi^2+1)^{-|P|}.
\end{equation}
This coincides with the upper bound at $\chi=1$, while being much smaller at $\chi>1$.

Before closing this Appendix, we make a comment about the choice of the RMPS ensemble. The unitarily embeddable RMPS ensemble has one shortcoming that at finite system sizes, states drawn from such an ensemble are not strictly normalized; however, the fluctuation of $\langle \psi|\psi\rangle$ around $1$ vanishes exponentially as $L$ is increased~\cite{haferkampEmergentStatisticalMechanics2021,chenMagicRandomMatrix2024}. Therefore, we can approximately see this ensemble as a normalized ensemble in large enough systems. Our derivation can also be naturally generalized to random MPS defined as constructing an open-chain MPS with a certain center of orthogonality, then sampling the left- and right-canonical tensors from Haar average. Normalization of the states would be guaranteed in this setup. However, it has been shown that this choice leads to a biased distribution over the manifold of MPSs, while the true unbiased RMPS distribution requires non-local correlations between the tensors, complicating the calculations~\cite{leonticaUnbiasedMeasureMatrix2025}. Nonetheless, we expect the qualitative behavior of the RMPS-norm to be consistent in all these setups.

\section{Properties of thermally dressed ensembles} \label{sec:scrooge}

In this Appendix, we will discuss some properties of the thermally dressed ensembles. Thermal dressing has been extensively studied in the context of Scrooge ensembles~\cite{jozsaLowerBoundAccessible1994,parfionovLazyQuantumEnsembles2006,mokNatureStingyUniversality2026,mannaProjectedEnsembleSystem2025,mcginleyScroogeEnsembleManybody2025,cotlerEmergentQuantumState2023,liuDeepThermalizationGaussian2024,markMaximumEntropyPrinciple2024a}, which are a generalization of the Haar ensemble in the case where the first moment, or the density matrix, of the ensemble is not a maximally mixed state. Given a density matrix $\rho$, the Scrooge ensemble takes a Haar random state $|\psi\rangle$ and deforms it to $\sqrt{\rho}|\psi\rangle$. In our case, we wish to construct typical physical initial states at a certain temperature, instead of typical random states; therefore, we would wish to thermally dress a weakly entangled ensemble, such as the RPS or RMPS, rather than the Haar ensemble. However, the properties discussed below do not depend on the ``base ensemble'', which we will keep general and denote by $\calE$. We will consider constructing ensembles whose first moment is equal to some Gibbs density matrix $\rho_\beta$; the constructions naturally generalize to targeting any density matrix.

There are two versions of the Scrooge ensemble discussed in the literature, the normalized and the unnormalized. Both constructions generalize to the dressing on arbitrary base ensembles. The normalized version has
\begin{equation}
\calE_\beta:\begin{cases}
|\Phi\rangle = \frac{\sqrt{\rho_\beta} |\psi\rangle}{\sqrt{\langle \psi|\rho_\beta|\psi\rangle}}, \\
p(\Phi) = \langle \psi|D\rho_\beta|\psi\rangle  p_\calE(\psi).
\end{cases}
\end{equation}
In this ensemble, $\langle \Phi|\Phi\rangle=1$. We can see that
\begin{equation}
\rho_{\calE_\beta} =\expt_{\psi\sim\calE}\langle \psi|D\rho_\beta|\psi\rangle \frac{\sqrt{\rho_\beta} |\psi\rangle\langle\psi|\sqrt{\rho_\beta} }{\langle \psi|\rho_\beta|\psi\rangle} = \rho_\beta,
\end{equation}
provided that $\rho_\calE = I/D$. On the other hand, the unnormalized version is defined by
\begin{equation}
\widetilde{\calE_\beta}:\begin{cases}
|\Phi\rangle = \sqrt{D\rho_\beta} |\psi\rangle, \\
p(\Phi) =  p_\calE(\psi).
\end{cases}
\end{equation}
In this case we would also have $\rho_{\widetilde{\calE_\beta}}=\rho_\beta$. The difference occurs at higher moments. For example, at the second moment, or equivalently, for the ensemble covariance inner product, one has
\begin{equation}
\langle A,B\rangle_{\widetilde{\calE_\beta}} = \left\langle \sqrt{D\rho_\beta} A\sqrt{D\rho_\beta}, \sqrt{D\rho_\beta} B\sqrt{D\rho_\beta} \right\rangle_\calE.
\end{equation}
This is exactly the thermal dressing we have discussed in the main text. However, for the normalized version, the expression is
\begin{multline}
\langle A,B\rangle_{\calE_\beta} \\ = \expt_{\psi\sim\calE} \frac{\langle \psi|\sqrt{D\rho_\beta}A^\dagger\sqrt{D\rho_\beta}|\psi\rangle \langle \psi|\sqrt{D\rho_\beta}B\sqrt{D\rho_\beta}|\psi\rangle}{\langle \psi|D\rho_\beta|\psi\rangle}.
\end{multline}
This expression does not admit a clear simplification due to the presence of $|\psi\rangle$ in the denominator.

We see that the unnormalized dressed ensemble enables more analytical calculations. One might worry that the unnormalized nature of the ensemble makes it less physical. We show that the ensemble variance norm of $\widetilde{\calE_\beta}$ is actually proportional to that of another normalized ensemble,
\begin{equation}
\calE_\beta^{(2)}:\begin{cases}
|\Phi\rangle = \frac{\sqrt{\rho_\beta} |\psi\rangle}{\sqrt{\langle \psi|\rho_\beta|\psi\rangle}}, \\
p(\Phi) =  \frac{\langle \psi|D\rho_\beta|\psi\rangle^2}{\|D\rho_\beta\|_\calE^2}  p_\calE(\psi).
\end{cases}
\end{equation}
It is straightforward to see that
\begin{equation}
\langle A,B\rangle_{\calE_\beta^{(2)}} = \frac{\langle A,B\rangle_{\widetilde{\calE_\beta}}}{\|D\rho_\beta\|_\calE^2}.
\end{equation}
Therefore, in a sense, we can see $\langle A,B\rangle_{\widetilde{\calE_\beta}}/\|D\rho_\beta\|_\calE^2$ as a ``physical'' ensemble inner product. The tradeoff is that
\begin{equation}
\rho_{\calE_\beta^{(2)}} = \frac{\sqrt{\rho_\beta}\calM_\calE[\rho_\beta]\sqrt{\rho_\beta}}{D\|\rho_\beta\|_\calE^2} \neq \rho_\beta.
\end{equation}
Therefore, while $\widetilde{\calE_\beta}$ can serve as a proxy for a finite-temperature ensemble, we see that no natural construction of a finite-temperature ensemble is both physically perfect and analytically tractable. We expect that finding a suitable ensemble for studying finite-temperature physics will be a non-trivial and potentially system-dependent problem.

\section{$n$-dependence of the $\nu(\tau)$ frontier}
\label{sec:n-dependence-nutau}

In this Appendix, we discuss in detail the role of the parameter $n$ in the $\nu_n(\tau)$ frontier. Recall that the definition of $\nu_n(\tau)$ is
\begin{equation}
\begin{array}{lll}
\nu_n(\tau) =&\text{maximize} \quad & \|
A\|_\text{RPS}^2, \\
&\text{subject to} & \|A\|^2 = 1,  \\
&& \mathrm{tr}[AH^k] = 0, \, 0\leq k \leq n, \\
&& \|[H,A]\|^2 \leq 1/\tau^2.
\end{array} \label{eq:nu-n-tau-def-appendix}
\end{equation}

We start with a simple observation: all powers of $H$ always commute with the Hamiltonian, and have multilocal structure if $H$ itself is few-body. The projection against the tower of powers of $H$ is introduced precisely to rule out such quantities. However, for any finite $n$, it is always possible to construct $K_{n+1}$ as the orthogonalization of $H^{n+1}$ against $\{I,H,\dots,H^n\}$ normalized to unit Frobenius norm, namely,
\begin{equation}
K_{n+1} = \frac{\calP_n[H^{n+1}]}{\|\calP_n[H^{n+1}]\|}. \label{eq:K-np1-definition}
\end{equation}
Note that $n+1$ has to be less than the degree of the minimal polynomial of $H$, otherwise $\calP_n[H^{n+1}]=0$. This will usually not be a concern, as such degree should generally be exponentially large in $L$ (the degree is precisely $D$ if the spectrum is non-degenerate), and we will not take $n$ beyond $\calO(L)$.

\begin{figure}
    \centering
    \includegraphics[width=0.8\linewidth]{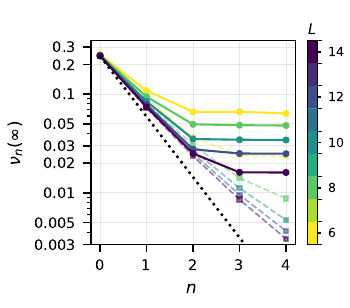}
    \caption{The dependence of $\nu_n(\infty)$ on $n$, illustrated for MFIM chains of different lengths as in Fig.~\ref{fig:nutau-mfim}. Solid lines show the numerically obtained frontier value $\nu_n(\infty)$, dashed lines show the squared RPS-norm of $K_{n+1}$, which lower-bounds $\nu_n(\infty)$. For each $L$, $\nu_n(\infty)$ is initially very close to the lower-bound line at small $n$, but then plateaus above a threshold $n_\ast$, while the RPS-norm of $K_{n+1}$ continues to decrease. The plateau value of $\nu_n(\infty)$ describes the simplest non-$K_{n+1}$ IOM and is what is shown in Fig.~\ref{fig:nutau-mfim}. As $L\to\infty$, $\|K_{n+1}\|_\text{RPS}^2$ converges to $r^{n+1}$ with $ r=(\|H\|_\text{RPS}/\|H\|)^{2}$ (black dotted line), and the $\nu_n(\infty)$ curve is expected to follow this line for any finite $n$.}
    \label{fig:nvar}
\end{figure}

The $K_{n+1}$ constructed would be a valid solution to the optimization problem in Eq.~\eqref{eq:nu-n-tau-def-appendix}. Therefore,
\begin{equation}
\nu_n(\infty) \geq \|K_{n+1}\|_\text{RPS}^2.
\end{equation}
If $H$ is $m$-local, then $K_{n+1}$ is no more than $m(n+1)$-local, so $\|K_{n+1}\|_\text{RPS}^2 \geq 3^{-m(n+1)}$. Therefore, for any fixed $n$, $\nu_n(\infty)$ is lower-bounded by an $L$-independent value, and cannot vanish in the thermodynamic limit. The interplay between $K_{n+1}$ and the simplistic picture of  $\nu(\infty) \sim \exp(-\alpha L)$ we gave in Section~\ref{sec:slow-operators} is demonstrated in Fig.~\ref{fig:nvar}, where we compared $\nu_n(\infty)$ with $\|K_{n+1}\|_\text{RPS}^2$ at different $n$ and $L$ for the same MFIM as in Fig.~\ref{fig:nutau-mfim}. For each $L$, at small $n$, we observe that $\nu_n(\infty) \approx \|K_{n+1}\|_\text{RPS}^2$. As $n$ is increased, $\|K_{n+1}\|_\text{RPS}$ decreases. $\nu_n(\infty)$ initially follows this decrease, yet after some threshold value $n_\ast$, it starts to plateau despite $\|K_{n+1}\|_\text{RPS}$ continuing to decrease. This implies that the simplest operator has become one that is independent from the powers of $H$, to which $\nu(\infty) \sim \exp(-\alpha L)$ can be expected to apply. Denoting this simplest operator as $A_\ast$, we would have to choose $n$ such that $\|K_{n+1}\|_\text{RPS} < \|A_\ast\|_\text{RPS}$ to have $A_\ast$ appear on the $\nu(\infty)$ frontier. We can show that in the thermodynamic limit,  $\|K_{n}\|_\text{RPS} \approx (\|H\|_\text{RPS}/\|H\|)^n$, assuming $H$ is traceless (see Appendix~\ref{sec:rps-Kn}); therefore, the threshold $n_\ast \approx \frac{\log \|A_\ast\|_\text{RPS}}{\log(\|H\|_\text{RPS}/\|H\|)}$. If $\|A_\ast\|_\text{RPS} \propto e^{-\alpha L}$, this implies $n_\ast$ is linear in $L$.

We can therefore pin down the notion of choosing $n$ to be ``large enough''. To make $\nu_n(\tau)$ independent of $K_{n+1}$, we need to let $n\gtrsim n_\ast$ so that we are in the plateau regime of Fig.~\ref{fig:nvar}, and the simplest operator we find would be the genuine $A_\ast$, whose RPS-norm can decrease exponentially with $L$. From Fig.~\ref{fig:nvar}, we see that $n=4$ serves this purpose at the system sizes we consider in Fig.~\ref{fig:nutau-mfim}. Note that while this discussion focuses on $\tau=\infty$, this choice of $n$ is applicable to all $\tau$, as the simplest non-$K_{n+1}$ operator at finite $\tau$ is simpler than the one at $\tau=\infty$; therefore, if the latter is already simpler than $K_{n+1}$, so is the former. However, as noted in Subsection~\ref{subsec:nutau-tfim}, it is not necessary to choose $n>n_\ast$ if one's goal is simply to obtain a thermalization bound that vanishes in the thermodynamic limit, as even when $\nu_n(\infty)\approx \|K_{n+1}\|_\text{RPS}^2$, the right-hand side still vanishes exponentially with respect to $n$; therefore, for any fixed accuracy $\nu_n(\infty) < \epsilon$, it would suffice to choose $n \sim \log \frac{1}{\epsilon}$.

\section{Properties of $\Onpow$} \label{sec:appendix-thermal-value}

In this Appendix, we will detail the properties of the polynomial ``thermal expectation value'' $\Onpow$. In \ref{sec:appendix-thermal-value-inf} and \ref{sec:appendix-thermal-value-finite-temp}, we offer semi-quantitative arguments that $\Onpow$ is an $n$-th degree polynomial approximation to $\Oth$. In \ref{subsec:convergence-of-Onpow}, irrespective of the relationship with $\Oth$, we show that $\Onpow$ converges well with respect to increasing $n$ under reasonable assumptions that argue that $\Onpow$ should be a good proxy for the thermal expectation value on its own.

\subsection{$\Onpow$ as an approximation to $\Oth$}
\label{sec:appendix-thermal-value-inf}

In Subsection~\ref{sec:thermal-expectation-value}, we argued that
\begin{equation}
O_\parallel^n = \sum_{k=0}^n c_k H^k \label{eq:Onpow-poly-expression}
\end{equation}
with the coefficients determined by
\begin{equation}
\langle O_\parallel^n-O, H^\ell\rangle = 0, \forall 0\leq \ell \leq n \label{eq:Npow-perp-O}
\end{equation}
will make $\langle \psi|O_\parallel^n|\psi\rangle$ a good approximation to the thermal expectation value $\Oth$. The argument goes as follows. Firstly, $O_\parallel^n$ commutes with the Hamiltonian. Therefore, it can be written as
\begin{align}
O_\parallel^n & = \sum_E f_{\npow}(E) |E\rangle \langle E|, \\ f_\npow(E) & = \sum_{k=0}^n c_k E^k. \label{eq:fnpow-polynomial}
\end{align}
Note that this does not require the spectrum of $H$ to be non-degenerate. Similarly, we write the diagonal part of $O$ as
\begin{equation}
O_\text{diag} = \sum_E f_{\text{diag}}(E)|E\rangle \langle E|,
\end{equation}
with $f_\text{diag}(E) = \langle E|O|E\rangle$. The microcanonical expectation value, $f_\text{MC}(E)$, is defined to be the local smoothing of $f_{\text{diag}}(E)$. It follows that for any smooth function $g(E)$, $\int f_\text{diag}(E) g(E)\mathrm dE = \int f_{\text{MC}}(E) g(E)\mathrm dE$. Then, Eq.~\eqref{eq:Npow-perp-O} implies that
\begin{equation}
\int \left[f_\npow(E)-f_\text{MC}(E)\right]E^\ell D(E)\mathrm dE = 0,\forall 0\leq \ell \leq n,
\end{equation}
where $D(E)$ is the spectral density, which should have the form $D(E)=e^{Ls(\epsilon)}$, where $\epsilon=E/L$ is the energy density. Assume that infinite temperature corresponds to zero energy density (which would be generally true if the Hamiltonian is traceless), such that $s^\prime(0)=0$, invoking the saddle point approximation in the limit of large $L$, one obtains $\left.\frac{\partial^\ell}{\partial \epsilon^\ell}\left[f_\npow(\epsilon)-f_\text{MC}(\epsilon)\right]\right|_{\epsilon=0} \approx 0$, which implies $f_\npow(\epsilon) \approx f_\text{MC}(\epsilon)+\calO(\epsilon^{n+1})$. However, at finite $L$, the error of the saddle-point approximation is on the order $1/L$, which implies that $\left.\frac{\partial^\ell}{\partial \epsilon^\ell}\left[f_\npow(\epsilon)-f_\text{MC}(\epsilon)\right]\right|_{\epsilon=0} = \calO(1/L)$, and therefore,
\begin{equation}
f_\npow(\epsilon) \approx f_\text{MC}(\epsilon)+\calO(\epsilon^{n+1}) + \calO(1/L).
\end{equation}

Formally, the above argument proving $f_\npow(\epsilon) \approx f_\text{MC}(\epsilon)$ requires the associated Taylor expansions to converge. However, the typical energy of a state in $\calE$ can be estimated with
\begin{equation}
\expt_{\psi\sim \calE} \langle \psi|H|\psi\rangle^2 = \|H\|_\calE^2.
\end{equation}
If $\calE$ is the RPS ensemble, $\|H\|_\text{RPS}^2 \sim \calO(L)$ for local Hamiltonians. Therefore, $|\langle \psi|H|\psi\rangle |\sim \calO(\sqrt L)$ for typical RPS,  so $\epsilon\sim \calO(L^{-1/2})$. Therefore, $|f_\npow(\epsilon)-f_\text{MC}(\epsilon)| \sim \calO(L^{-\min\{1,(n+1)/2\}})$ vanishes in the thermodynamic limit for any $n\geq 0$, without placing any requirement on the Taylor expansion convergence radius of the function $f_\text{MC}(\epsilon)$. Note, however, that this does not imply the thermalization Theorem is meaningful for any $n \geq 0$, as we also require $\nu_n(\tau)$ to be small, which places a separate constraint on $n$ (see Appendix~\ref{sec:n-dependence-nutau}).

However, if $\calE$ is the finite-temperature stabilizer product state ensemble (see Subsection~\ref{subsec:thermalization-at-finite-temperature}), the typical energy density would now be finite. Therefore, a finite convergence radius is now necessary to establish that $|f_\npow(\epsilon)-f_\text{MC}(\epsilon)|$ is small. Within the convergence radius, $|f_\npow(\epsilon)-f_\text{MC}(\epsilon)|$ can be made polynomially small in $L$ as long as $n$ is at least on the order of $\log L$. As was shown in Subsection~\ref{subsec:thermalization-at-finite-temperature}, $n = \Theta(L)$ is required in this case to obtain a non-trivial thermalization bound; therefore, such $n$ will always be large enough to make  $|f_\npow(\epsilon)-f_\text{MC}(\epsilon)|$ vanish.

\subsection{The finite-temperature $\Onpow$} \label{sec:appendix-thermal-value-finite-temp}

We now proceed to discuss the properties of the polynomial thermal expectation value at finite temperature, which was used in the thermally dressed RPS ensemble discussed in Subsection~\ref{subsec:thermalization-at-finite-temperature}.

We first generalize the argument we just made to finite temperature. If Eq.~\eqref{eq:Npow-perp-O} is generalized to $\langle O_{\parallel,\beta}^n-O, e^{-\beta H} H^\ell\rangle = 0$, the same argument gives
\begin{equation}
f_{\npow,\beta}(\epsilon) \approx f_\text{MC}(\epsilon)+\calO((\epsilon-\epsilon_\beta)^{n+1}) + \calO(1/L),
\end{equation}
where $\epsilon_\beta$ satisfies $s^\prime(\epsilon_\beta) = \beta$, such that it is the maximum of the new exponent $e^{L(s(\epsilon)-\beta\epsilon)}$. Notably, the argument does not require that $O_\parallel^n$ take the exact form in Eq.~\eqref{eq:Onpow-poly-expression}; instead, it suffices that $f_\npow$ is smooth with respect to $\epsilon$.

We see that the $\Onpow$ in the thermally dressed RPS ensemble formulation is exactly an expansion around $\epsilon_\beta$. To see this, recall that we introduced the finite-temperature projection $\calP_\beta$, which shall satisfy
\begin{equation}
\|\calP[O^\beta(t)]\|_\calE = \|\calP_\beta[O(t)]\|_{\widetilde{\calE_\beta}}.
\end{equation}
Consistent with the main text, we have suppressed the subscript $n$ of the projector $\calP$ here. It shall be understood that both $\calP$ and $\calP_\beta$ depend on $n$. Since $\|A\|_{\widetilde{\calE_\beta}} = \|A^\beta\|_\calE$, the definition is equivalent to
\begin{equation}
\calP[O^\beta(t)] = \left(\calP_\beta[O(t)]\right)^\beta. \label{eq:def-of-eneryg-projection-appendix}
\end{equation}
Since both the projection $\calP$ and the dressing $\cdot^\beta$ commute with time evolution, Eq.~\eqref{eq:def-of-eneryg-projection-appendix} would hold at all times if it holds at $t=0$. Therefore, the projection $\calP_\beta$ would be defined by
\begin{equation}
\left(\calP_\beta[O]\right)^\beta = \calP[O^\beta],
\end{equation}
and $\calQ_\beta[O] = O-\calP_\beta[O]$, accordingly.

To derive an explicit form of $\calQ_\beta[O]$, note that
\begin{equation}
\calQ[O^\beta] = \sum_i Q_i \langle Q_i,O^\beta\rangle = \sum_i Q_i \langle Q_i^\beta,O\rangle.
\end{equation}
Let $Q_i^{\bar\beta}$ denote the operator such that $\left(Q_i^{\bar\beta}\right)^\beta = Q_i$. Explicitly,
\begin{equation}
O^{\bar\beta} = \langle e^{-\beta H}\rangle e^{\frac{\beta H}{2}} O e^{\frac{\beta H}{2}}.
\end{equation}
Then, we see
\begin{equation}
\calQ_\beta[O] = \sum_i Q_i^{\bar\beta} \langle Q_i^\beta,O\rangle.
\end{equation}
Importantly,
\begin{equation}
\langle Q_j^\beta,\calQ_\beta[O]\rangle = \sum_i \langle Q_j^\beta, Q_i^{\bar\beta}\rangle \langle Q_i^\beta,O\rangle,
\end{equation}
with
\begin{equation}
\langle Q_j^\beta, Q_i^{\bar\beta}\rangle = \left\langle Q_j, \left(Q_i^{\bar\beta}\right)^\beta \right\rangle = \langle Q_j,Q_i\rangle = \delta_{ij},
\end{equation}
we have
\begin{equation}
\left\langle Q_j^\beta,\calQ_\beta[O] \right\rangle = \langle Q_j^\beta,O\rangle.
\end{equation}
That is, $\calQ_\beta[O]$ preserves the projection of $O$ onto the subspace spanned by the operators $\{Q_i^\beta\}$, although it is not an orthogonal projection onto this subspace. In this case, Eq.~\eqref{eq:Npow-perp-O} becomes
\begin{equation}
\langle \calQ_\beta[O]-O, e^{-\beta H} H^\ell\rangle = 0, \forall 0\leq \ell \leq n \label{eq:Npow-perp-O-finite-temp}.
\end{equation}
This establishes that $O_{\parallel,\beta}^n \equiv\calQ_\beta[O]$ is a Taylor expansion of the thermal expectation value around $\epsilon_\beta$.

The state ensemble concerned with this thermal expectation value is the finite-temperature dressed ensemble $\widetilde{\text{RPS}}_\beta$. Although we cannot make a rigorous argument as in the infinite-temperature case, it is reasonable to expect that most states in $\widetilde{\text{RPS}}_\beta$ also have energy density concentrated on a scale of $\calO(L^{-1/2})$ around $\epsilon_\beta$. On the other hand, a $n=\Theta(L)$ is required to produce a non-trivial thermalization bound. Therefore, the $n$ value chosen is more than enough to make the error term $\calO((\epsilon-\epsilon_\beta)^{n+1})$ vanish.

\subsection{Convergence of $\Onpow$} \label{subsec:convergence-of-Onpow}

Irrespective of the arguments offered above, we further argue that $\Onpow$ should by itself serve as a reasonable notion of the thermal expectation value, with or without explicitly establishing bounds on $|\Onpow-\Oth|$.

Firstly, as long as a tight bound $\nu$ can be established, $\Onpow$ would be an accurate description of the equilibrium value. If such value depends only on energy, it would be the definition of thermalization, and $\Onpow$ would also be the thermal expectation value by definition. The key difference here is that $\Onpow$ contains information about not just the total energy, but also the energy variance and higher cumulants of $H$. It is usually expected that, for physical initial states, such higher cumulants of $H$ do not have an effect on the equilibrium value in the thermodynamic limit. In this sense, $\Onpow$ must be equal to $\Oth$ by definition, as long as $\nu$ is demonstrably small. In fact, $\Onpow$ can even be considered a more fine-grained description of the thermal expectation value than traditional $\Oth$ measures. We have seen that the difference between $\Onpow$ and $\Oth$ always includes an $\calO(1/L)$ error term. This is inevitable, as different thermal ensembles already differ from each other on a similar, or even larger, scale. On the other hand, we have shown that the bound $\nu$ can scale as $e^{-\Theta(L)}$ in the best case, implying that $\Onpow$ predicts the equilibrium value far better than $\Oth$ by including higher cumulants.

We then prove a technical result that further establishes that $\Onpow$ is a physically reasonable quantity. We show that the value of $\Onpow$ is not sensitive to the choice of $n$ for most product initial states under some mild assumptions. This establishes that $\Onpow$ has a stable value in the range where $n$ is adequately chosen. Specifically,
\begin{proposition}
Assuming that the observable $O$ is normalized to $\|O\|=1$, and the Hamiltonian is a sum of local terms, with $\langle H^2\rangle = \Theta(L)$, then for $m > n$,
\begin{equation}
\expt_{\psi\sim\mathrm{RPS}}|\Onpow(\psi)-\langle O\rangle_{m\text{-}\mathrm{pow}}(\psi)|^2 \leq a e^{-bn}
\end{equation}
for some constants $a$,$b$. If one further assumes that $O$ is local, then the difference scales as $\calO(L^{-n-1})$.
\end{proposition}
\begin{proof}
Define $K_n$ similarly as in Eq.~\eqref{eq:K-np1-definition},
\begin{equation}
K_{n} = \frac{\calP_{n-1}[H^{n}]}{\|\calP_{n-1}[H^{n}]\|},
\end{equation}
where $\calP_{n-1}$ means orthogonalization against $\{I,H,\dots,H^{n-1}\}$. $\{K_0,\dots,K_n\}$ would then be an orthonormal basis of $\text{span}\{I,\dots,H^n\}$. It is proven in Appendix~\ref{sec:rps-Kn} that
\begin{equation}
\|K_n\|_\text{RPS} \approx (\|H\|_\text{RPS}/\|H\|)^n
\end{equation}
under reasonable assumptions. As a result, consider
\begin{equation}
O_\parallel^m - O_\parallel^n = \sum_{k=n+1}^m \langle K_k,O\rangle K_k,
\end{equation}
where we have assumed $m>n$. For a given initial state, this implies
\begin{multline}
    |\Onpow(\psi)-\langle O\rangle_{m\text{-}\mathrm{pow}}(\psi)|^2 \\ = \left(\sum_{k=n+1}^m \langle K_k,O\rangle \langle \psi|K_k|\psi\rangle\right)^2 \\
    \leq \left(\sum_{k=n+1}^m \langle K_k,O\rangle^2\right)\left(\sum_{k=n+1}^m \langle \psi|K_k|\psi\rangle^2\right). \label{eq:Onpow-Ompow-diff-1}
\end{multline}
As $\{K_k\}$ is an orthonormal basis, it follows that
\begin{equation}
    \sum_{k=n+1}^m \langle K_k,O\rangle^2 \leq \|O\|^2 = 1.
\end{equation}
Then, averaging Eq.~\eqref{eq:Onpow-Ompow-diff-1} gives
\begin{multline}
\expt_{\psi\sim\text{RPS}} |\Onpow(\psi)-\langle O\rangle_{m\text{-}\mathrm{pow}}(\psi)|^2 \\
\leq \sum_{k=n+1}^m \|K_k \|_\text{RPS}^2 \lesssim \frac{r^n}{r^{-1}-1},
\end{multline}
where we have denoted $r=\left(\frac{\|H\|_\text{RPS}}{\|H\|}\right)^2 < 1$. This shows that $\Onpow$ stabilizes exponentially with $n$.

If one further assumes that $O$ is local, we can establish that
\begin{equation}
\langle O,\calP_{k-1}[H^k]\rangle = \calO(1). \label{eq:O-PkH-overlap-1}
\end{equation}
As Appendix~\ref{sec:rps-Kn} established that $\|\calP_{k-1}[H^k]\| = \Theta(L^{k/2})$, this implies $\langle K_k,O\rangle = \calO(L^{-k/2})$, which gives
\begin{equation}
\|O_\parallel^n - O_\parallel^m\|_\text{RPS}^2 = \calO(L^{-n-1}).
\end{equation}

To prove Eq.~\eqref{eq:O-PkH-overlap-1}, first define a tower of operators $W_k$ as
\begin{equation}
\frac{e^{tH}}{\langle e^{tH}\rangle} = \sum_{k=0}^{\infty} \frac{t^k}{k!} W_k.
\end{equation}
Each $W_k$ would be a monic polynomial of $H$. Then,
\begin{align}
\langle O,W_k\rangle & = \left.\frac{\partial^k}{\partial t^k} \frac{\langle O,e^{tH}\rangle}{\langle e^{tH}\rangle} \right|_{t=0}  \nonumber \\
& = \left. \frac{\partial^k}{\partial t^k}\frac{\partial}{\partial s} \log\left\langle e^{sO+tH} \right\rangle \right|_{s=t=0} \nonumber \\
& = \langle OH^k\rangle_c = \calO(1). \label{eq:O-Wk-overlap-O1}
\end{align}
The last line comes from the fact that
\begin{equation}
\langle OH^k\rangle_c = \sum_{x_1,\dots,x_k} \langle O h_{x_1}\dots h_{x_k}\rangle_c,
\end{equation}
while the connected correlator vanishes if there is a subset of terms whose supports do not overlap with the rest of the terms. Therefore, the non-vanishing term must have all $h_{x_i}$ supported in an $\calO(1)$ region around the support of $O$, and the number of such terms is $\calO(1)$.

To convert $\calP_{k-1}[H^k]$ to $W_k$, note that by definition,
\begin{equation}
\sum_{k,\ell} \frac{t^{k}s^\ell}{k!\ell!}\langle W_k,W_\ell\rangle = \frac{\langle e^{(t+s)H}\rangle}{\langle e^{tH}\rangle \langle e^{sH}\rangle}.
\end{equation}
Meanwhile,
\begin{equation}
\langle e^{tH}\rangle = \exp\left(\sum_{k=1}^{\infty} \frac{t^k}{k!} \langle H^k\rangle_c\right).
\end{equation}
Therefore,
\begin{equation}
\frac{\langle e^{(t+s)H}\rangle}{\langle e^{tH}\rangle \langle e^{sH}\rangle} = \exp\left(\sum_{k,\ell=1}^{\infty} \frac{t^ks^\ell}{k!\ell!} \langle H^{k+\ell}\rangle_c\right).
\end{equation}
Because all $\langle H^{k+\ell}\rangle_c = \calO(L)$ due to the same overlapping-support argument as above, we can conclude that
\begin{equation}
\langle W_k,W_\ell\rangle = \calO(L^{\min\{k,\ell\}}),
\end{equation}
and in particular $\|W_k\|^2 = \Theta(L^k)$.
Therefore, for all $\ell <k$, $\langle W_k,W_\ell\rangle /\langle W_\ell,W_\ell\rangle = \calO(1)$. This implies that the coefficients required to orthogonalize $W_k$ against all lower-power polynomials $W_\ell$ is $\calO(1)$:
\begin{equation}
\calP_{n-1}[H^n] = W_n - \sum_{\ell=0}^{n-1}\calO(1)\times W_\ell.
\end{equation}
Combining this with Eq.~\eqref{eq:O-Wk-overlap-O1} establishes Eq.~\eqref{eq:O-PkH-overlap-1} and therefore concludes the proof.
\end{proof}

Note that the result in Appendix~\ref{sec:rps-Kn} holds rigorously only for fixed $n$ as $L\to\infty$; when $n$ scales together with $L$, the result should be understood semi-quantitatively. A similar reservation should be made for this Proposition.

\section{Thermalization in finite time}

\label{sec:thermalization-in-finite-time}

In this Appendix, we set up a generic formalism for thermalization in finite time. Consider an operator $O$, which admits the spectral decomposition,
\begin{equation}
O(t) = \sum_\lambda O_\lambda e^{i\lambda t}. \label{eq:O-spectral-decomp}
\end{equation}
The operators $O_\lambda$ satisfy $[H,O_\lambda]=\lambda O_\lambda$. As a result,
\begin{align}
\langle O_{\lambda},O_{\lambda^\prime}\rangle & = \|O_\lambda\|^2 \delta_{\lambda,\lambda^\prime}; \\
\sum_\lambda \|O_\lambda\|^2 & = \|O\|^2.
\end{align}

For finite-time thermalization, we have an envelope function $f(t)$, and consider the average
\begin{equation}
D_f = \int_{0}^{+\infty} \|O(t)\|_\calE^2 f(t) \mathrm dt.
\end{equation}
The common average on $[0,T]$ corresponds to choosing $f(t) = \chi([0,T])/T$, where $\chi$ is the characteristic function that gives value one for points in the given set and zero elsewhere.

For the moment, let us consider the more general object,
\begin{equation}
\tilde D_g = \int_{-\infty}^{+\infty} \|g(t) O(t)\|_\calE^2 \mathrm dt,
\end{equation}
which could upper bound $D_f$ by choosing adequate $g(t)$. We can apply Parseval's theorem to this expression, producing
\begin{equation}
\tilde D_g = \int_{-\infty}^{+\infty}\| O_{\mu,g} \|_\calE^2 \mathrm d\mu, \label{eq:Dg-parseval-result}
\end{equation}
where
\begin{equation}
O_{\mu,g} = \frac{1}{\sqrt{2\pi}}\int_{-\infty}^{+\infty} g(t) O(t) e^{-i\mu t} \mathrm dt = \sum_\lambda \tilde g(\mu-\lambda)O_\lambda,
\end{equation}
in which
\begin{equation}
\tilde g(\mu) = \frac{1}{\sqrt{2\pi}}\int_{-\infty}^{+\infty} g(t) e^{-i\mu t} \mathrm dt
\end{equation}
is the Fourier transform of $g$. Therefore, $O_{\mu,g}$ is the result of applying a spectral filter centered at $\mu$ with profile $\tilde g$ to $O$.

We wish the terms $O_{\mu,g}$ on the right-hand side of Eq.~\eqref{eq:Dg-parseval-result} to all have $\|[H,O_{\mu,g}] - \mu O_{\mu,g}\| / \|O_{\mu,g}\| \leq 1/\tau$. To this end, let $\tilde g(\mu) = \text{rect}(\mu\tau)$, where $\rect(x)=\Theta(1-|x|)$ is the characteristic function of the interval $[-1,1]$. The corresponding $g$ would be
\begin{equation}
g(t) = \frac{1}{\sqrt{2\pi}}\int_{-1/\tau}^{1/\tau} e^{i\mu t}\mathrm d\mu = \sqrt{\frac{2}{\pi}}\frac{\sin(t/\tau)}{t}.
\end{equation}
Then, we would have
\begin{equation}
O_{\mu,g} = \sum_{|\lambda-\mu|\leq 1/\tau} O_\lambda =: B_\mu
\end{equation}
By assumption, $\| B_\mu \|_\calE^2 \leq \nu \|B_\mu\|^2$. Therefore, we have
\begin{equation}
\tilde D_g \leq \nu \int_{-\infty}^{+\infty}\| B_\mu \|^2 \mathrm d\mu.
\end{equation}
Substituting
\begin{equation}
\|B_\mu\|^2 = \sum_{|\lambda-\mu|\leq 1/\tau} \|O_\lambda\|^2,
\end{equation}
we get
\begin{align}
\tilde D_g& \leq \nu \sum_{\lambda} \|O_\lambda\|^2 \int \rect\left(\tau(\lambda-\mu)\right)  d\mu \nonumber \\
& \leq \frac{2\nu}{\tau} \sum_{\lambda} \|O_\lambda\|^2 = \frac{2\nu}{\tau} \|O\|^2.
\end{align}
Plugging back the exact form of $g(t)$, we get
\begin{equation}
\int_{-\infty}^{+\infty} \frac{2\sin^2(t/\tau)}{\pi t^2/\tau} \|O(t)\|_\calE^2 \mathrm dt \leq 2\nu \|O\|^2.
\end{equation}
Restricting the integration over time to $[0,\infty)$ gives Eq.~\eqref{eq:O-t-thermalization-finite-time-bound}. Discarding half of the integration on the left makes this bound looser by a factor of $2$ than the infinite-time bound when we take $\tau\to\infty$. Noting that $\frac{2\sin^2(t/\tau)}{\pi t^2/\tau}$ is a decreasing function in $t$ for $t\leq \pi \tau$, for any $T \leq \pi \tau$, relaxing its value to $\frac{2\sin^2(T/\tau)}{\pi T^2/\tau}$ gives Eq.~\eqref{eq:O-t-thermalization-finite-time-bound-with-T}.

\section{Floquet formalism} \label{sec:floquet}

All of the formalism we developed for time-independent Hamiltonians can naturally be generalized to Floquet systems. Assume that the time-dependent Hamiltonian is $H(t)$, with $H(t+1)=H(t)$. Define the propagator
\begin{equation}
U(t_2,t_1) = \mathcal{T} \exp\left[-i \int_{t_1}^{t_2} H(t) \mathrm dt\right].
\end{equation}
The Floquet unitary $U$ is defined as $U=U(1,0)$. With this, the general propagator starting from $t_1=0$ has the form
\begin{equation}
U(t,0) = U(t,\lfloor t\rfloor) U^{\lfloor t\rfloor}.
\end{equation}
Due to periodicity, $U(t,\lfloor t \rfloor) = U(t-\lfloor t \rfloor,0)$. The operator evolution is given by
\begin{equation}
O(t) = U(t,0)^\dagger O U(t,0) = (U^\dagger)^{\lfloor t\rfloor} O(t-\lfloor t\rfloor ) U^{\lfloor t\rfloor}.
\end{equation}
Therefore,
\begin{multline}
\int_0^{+\infty} f(t) \|O(t)\|_\calE^2 \mathrm dt \\
= \int_0^1 \mathrm d\tau \sum_{t=0}^{+\infty} f(t+\tau) \left\| (U^\dagger)^{t} O(\tau) U^{t} \right\|_\calE^2.
\end{multline}
Therefore, a bound on the average fluctuation at stroboscopic times,
\begin{equation}
\sum_{t=0}^{+\infty} f(t) \left\| (U^\dagger)^{t} O U^{t} \right\|_\calE^2, \label{eq:bound-O-norm-discrete-times}
\end{equation}
can generalize to a bound averaging over continuous times. We will restrict ourselves to bounding Eq.~\eqref{eq:bound-O-norm-discrete-times} hereafter.

The superoperator $O \mapsto U^\dagger O U$ is a unitary superoperator with respect to the Hilbert-Schmidt inner product, as
\begin{equation}
\tr [(U^\dagger AU)^\dagger U^\dagger BU] = \tr[A^\dagger B].
\end{equation}
Therefore, any operator $O$ can similarly be decomposed into an orthogonal eigenbasis of it,
\begin{equation}
O = \sum_\lambda O_\lambda,
\end{equation}
such that
\begin{align}
U^\dagger O_\lambda U & = e^{i\lambda} O_\lambda, \\
\langle O_{\lambda},O_{\lambda^\prime}\rangle & = \|O_\lambda\|^2 \delta_{\lambda,\lambda^\prime}, \\
\sum_\lambda \|O_\lambda\|^2 & = \|O\|^2.
\end{align}
Note that due to the eigenvalue being defined as $e^{i\lambda}$, $\lambda$ is defined on the interval $[0,2\pi]$ with PBC. With the spectral decomposition, we have the time evolution
\begin{equation}
O(t) = (U^\dagger)^{t} O U^{t} = \sum_{\lambda} e^{i\lambda t}O_\lambda.
\end{equation}
All times $t$ should be understood as integers (\textit{i.e.}, stroboscopic) hereafter.

With this, the proof in Appendix~\ref{sec:thermalization-in-finite-time} can be repeated almost verbatim. Parseval's theorem gives
\begin{equation}
\sum_{t=-\infty}^{+\infty} \| g(t) O(t)\|_\calE^2 = \int_0^{2\pi} \frac{\mathrm d\mu}{2\pi} \left\| \sum_{t=-\infty}^{+\infty} e^{-i\mu t} g(t) O(t)  \right\|_\calE^2.
\end{equation}
Similarly define
\begin{equation}
\tilde g(\mu) = \frac{1}{\sqrt{2\pi}} \sum_{t=-\infty}^{+\infty} g(t) e^{-i\mu t},
\end{equation}
we would have
\begin{equation}
O_{\mu,g} = \frac{1}{\sqrt{2\pi}}\sum_{t=-\infty}^{+\infty} e^{-i\mu t} g(t) O(t) = \sum_{\lambda} \tilde g(\mu-\lambda) O_\lambda.
\end{equation}
Similar to Appendix~\ref{sec:thermalization-in-finite-time}, we choose $\tilde g(\mu) = \text{rect}(\mu\tau)$, corresponding to $g(t) = \sqrt{\frac{2}{\pi}}\frac{\sin(t/\tau)}{t}$. It follows that $B_\mu = O_{\mu,g} = \sum_{|\lambda-\mu|\leq 1/\tau} O_\lambda$. The rest of the proof follows as
\begin{equation}
\int_0^{2\pi} \|B_\mu\|^2\mathrm d\mu = \frac{2}{\tau} \|O\|^2,
\end{equation}
so
\begin{equation}
\sum_{t=-\infty}^{+\infty} \frac{\sin^2(t/\tau)}{\pi t^2/\tau} \|O(t)\|_\calE^2 \leq \|O\|^2 \cdot \sup_\mu\frac{\|B_\mu\|_\calE^2}{\|B_\mu\|^2}.
\end{equation}
Notably, since the eigenvalues of $U^\dagger \cdot U$ are $e^{i\lambda}$, this gives
\begin{multline}
\| U^\dagger B_\mu U - e^{i\mu} B_\mu \| = \sqrt{\sum_{|\lambda-\mu|\leq 1/\tau} \|O_\lambda\|^2 |e^{i\lambda}-e^{i\mu}|^2 } \\
\leq 2\sin \frac{1}{2\tau} \|B_\mu\|. \label{eq:Bmu-small-commutation-Floquet}
\end{multline}
When $\tau$ is large, $2\sin \frac{1}{2\tau} \sim \frac{1}{\tau}$, this becomes identical to the Hamiltonian case. In general, the Floquet analog of Eq.~\eqref{eq:finite-time-slowop-cond-2} would be
\begin{equation}
\left\| U^\dagger AU - e^{i\lambda}  A \right\| \leq 2\sin \frac{1}{2\tau}. \label{eq:UaU-bound-one-over-two-tau}
\end{equation}
Therefore, when defining the $\tau$ of an operator in the Floquet case, we would take
\begin{equation}
\tau = \frac{1}{2} \left(\arcsin\frac{\|[U,A]\|}{2\|A\|}\right)^{-1},
\end{equation}
in which we have taken $\lambda=0$ in line with Section~\ref{sec:slow-operators}.

To conclude, we offer a full statement of the thermalization bound in the Floquet case. Assuming that for all operators $A$ such that
\begin{align}
\|A\| & = 1,  \\
\|U^\dagger A U- e^{i\lambda} A \|& \leq 2\sin \frac{1}{2\tau} \text{ for some }\lambda , \\
\mathrm{tr}[AQ_i]& = 0 \text{ for all }i,
\end{align}
there is
\begin{equation}
\|A\|_\calE^2 \leq \nu,
\end{equation}
then,
\begin{multline}
\expt_{\psi\sim \calE}\sum_{t=-\infty}^{+\infty} \frac{\sin^2 (t/\tau)}{\pi t^2/\tau} \left[\langle \psi|O(t)|\psi\rangle - \langle \psi|\calQ[O]|\psi\rangle\right]^2 \\
\leq \nu \|O\|^2,
\end{multline}
where $\calQ[O]$ is the orthogonal projection of $O$ onto the subspace spanned by the IOMs $\{Q_i\}$. For generic Floquet unitaries, no orthogonalization is needed apart from making the operator traceless, corresponding to $\{Q_i\}=\{I\}$. The set $\{Q_i\}$ becomes relevant when, for example, we have a $U(1)$-conserving Floquet dynamics, in which case we will take $\{Q_i\}$ to be the powers of the total charge operator.

\section{Joint numerical range and the $\tau$-$\nu$ plane} \label{sec:jnr}

\begin{figure*}[!htbp]
    \centering
    \includegraphics[width=0.95\textwidth]{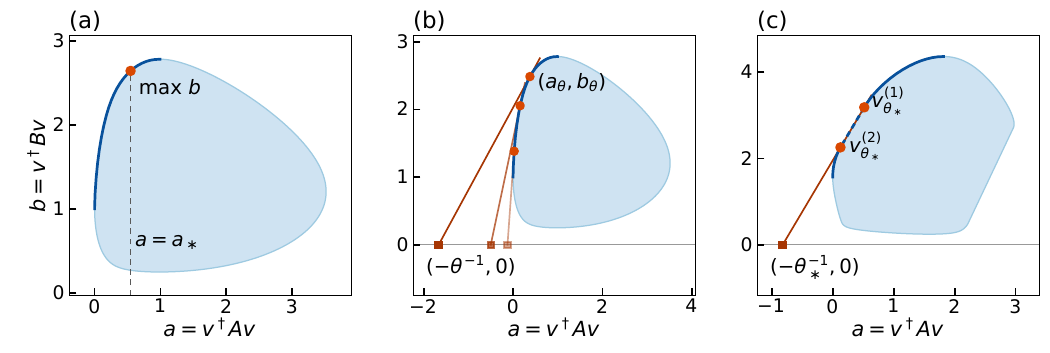}
  \caption{Illustration of the JNR and the algorithm to map its boundary. (a) Illustration of a JNR for two matrices $A$ and $B$,  as defined in Eq.~\eqref{eq:def-jnr-AB}. Maximizing $b$ subject to a constraint on $a$ is equivalent to finding the intersection of the upper boundary of the JNR (solid blue line) with a vertical line. (b) The top generalized eigenvector in Eq.~\eqref{eq:B-generalized-evp} corresponds to a point on the boundary, for which the tangent line intersects the $a$-axis on $(-\theta^{-1},0)$. As $\theta$ varies, one can reach all points on the segment of the boundary. (c) Degenerate generalized eigenvalue at a particular $\theta_\ast$ corresponds to a straight line segment of the boundary.}
    \label{fig:jnr}
\end{figure*}

Recall that the region in the $\tau$-$\nu$ plane occupied by all legitimate operators is given by
\begin{equation}
\calR = \left\{\left(1/\|[H,O]\|, \|O\|_\calE^2\right) \middle| \|O\|=1, O \perp \{Q_i\} \right\}.
\end{equation}
We will drop the constraint $O \perp \{Q_i\}$ for convenience, by implicitly understanding that $[H,\cdot]$ and $\calM_\calE$ are sandwiched between projectors into the subspace orthogonal to all $\{Q_i\}$. All discussions in this Appendix can be generalized to non-zero $\lambda$ (see Appendix~\ref{sec:sga}) by replacing $[H,\cdot]$ with $[H,\cdot]-\lambda$. For the sake of simplicity, we will restrict to the $\lambda=0$ case.

The construction $\calR$ is highly similar to the definition of the joint numerical range (JNR). Given two Hermitian operators $A$ and $B$, the JNR of them is defined as the subset of the plane
\begin{equation}
W(A,B) = \left\{ (v^\dagger Av,v^\dagger Bv) \middle| v^\dagger v=1 \right\}. \label{eq:def-jnr-AB}
\end{equation}
Therefore, $\calR$ is essentially $W([H,\cdot]^2,\calM_\calE)$, with a global transformation on the $\tau$ axis.

It is known that for any Hermitian operators $A$ and $B$, the JNR $W(A,B)$ is always a convex subset of the plane. This result is known as Brickman's theorem. A proof can be found in Refs.~\cite{barvinokCourseConvexity2002,quasicoherent2021brickman}.

The convexity ensures that finding the envelope of $W(A,B)$ via a (generalized) eigenvalue problem is always feasible. Specifically, consider the problem
\begin{equation}
\begin{array}{ll}
\text{maximize} \quad & b = v^\dagger  B v, \\
\text{subject to} & v^\dagger v = 1,  \\
& v^\dagger A v \leq a_\ast.
\end{array} \label{eq:A-B-optimization-array}
\end{equation}
In the JNR picture, this is equivalent to finding the intersection of the JNR's boundary with the vertical line $a=a_\ast$, as demonstrated in Fig.~\ref{fig:jnr}(a). We have chosen $A$ to be positive semidefinite and $B$ to be positive definite, as they stand in the actual $\tau$-$\nu$ problem, and we will assume such properties of $A$ and $B$ hereafter.

The problem can be recast as a generalized eigenvalue problem using Lagrange multipliers. To motivate this, recall that maximizing $v^\dagger Bv$ subject to $v^\dagger v=1$ is equivalent to finding the largest eigenvalue of $B$. This can be obtained from casting the optimization into a Lagrange multiplier problem, $\frac{\partial}{\partial v}\left[v^\dagger Bv - \lambda v^\dagger v\right]=0$, which yields $Bv = \lambda v$; therefore, the maximizing vector must be an eigenvector of $B$, and the nature of maximization immediately yields that it is the eigenvector corresponding to maximal eigenvalue. Similarly, the optimization in Eq.~\eqref{eq:A-B-optimization-array} can be cast into the Lagrange multiplier problem
\begin{equation}
\frac{\partial}{\partial v}\left[v^\dagger Bv - \lambda_1 v^\dagger Av - \lambda _2 v^\dagger v\right]=0,
\end{equation}
which yields
\begin{equation}
Bv = \lambda_1 Av + \lambda_2 v. \label{eq:B-Av-two-lambdas}
\end{equation}
This equation must be satisfied for any $v$ lying on the boundary of the JNR for some $\lambda_1$ and $\lambda_2$. Since there are two Lagrange multipliers, this is not explicitly an eigenvalue equation yet. We organize this into the generalized eigenvalue problem,
\begin{equation}
Bv = \mu(I+\theta A)v, \label{eq:B-generalized-evp}
\end{equation}
where $\theta$ is now a tunable parameter. Each $\theta$ defines a generalized eigenvalue problem, where the generalized eigenvalue $\mu$ can be solved. We will choose $\theta>0$, which ensures the positive definiteness of $I+\theta A$. We note that the way of casting Eq.~\eqref{eq:B-Av-two-lambdas} into a (generalized) eigenvalue problem is not unique: for example
$(B-\theta A)v = \mu v$ is also a legitimate reparametrization, which yields an eigenvalue problem for each choice of $\theta$. However, we will see that the choice in Eq.~\eqref{eq:B-generalized-evp} is superior in terms of numerical stability in our specific setting, as will be demonstrated below.

Denote a generalized eigenvalue and eigenvector pair at a given $\theta$ by $\mu_\theta$ and $v_\theta$, respectively. For the moment, assume that the eigenvalue is not degenerate. Let $a_\theta = v_\theta^\dagger A v_\theta$ and $b_\theta = v_\theta^\dagger B v_\theta$. The generalized eigenvalue equation implies
\begin{equation}
b_\theta = \mu_\theta(1+\theta a_\theta). \label{eq:a-b-relation-at-fixed-theta}
\end{equation}
We will show that points $(a_\theta,b_\theta)$ lying on the upper boundary of the JNR can always be identified by taking $\mu_\theta$ to be the largest generalized eigenvalue, at some $\theta$. To see this, assume for the moment that the top generalized eigenvalue is non-degenerate. Differentiating the generalized eigenvalue equation in a Hellmann-Feynman fashion gives
\begin{equation}
\mu_\theta^\prime(1+\theta a_\theta) + \mu_\theta a_\theta = 0,
\end{equation}
where we have used the prime to denote derivative with respect to $\theta$. Inserting this into a direct differentiation of Eq.~\eqref{eq:a-b-relation-at-fixed-theta} gives
\begin{equation}
b_\theta^\prime - \theta \mu_\theta a_\theta^\prime=0.
\end{equation}
Canceling $\mu_\theta$, this gives
\begin{equation}
\frac{b_\theta^\prime}{a_\theta^\prime} - \frac{b_\theta}{\theta^{-1}+a_\theta}=0.
\end{equation}
This implies that the tangent of the $(a_\theta,b_\theta)$ trajectory at this point, with equation given by $b-b_\theta = \frac{b_\theta^\prime}{a_\theta^\prime}(a-a_\theta)$, passes through the point $(-\theta^{-1},0)$. Therefore, as demonstrated in Fig.~\ref{fig:jnr}(b), for each point on the boundary of the JNR, we can identify the $\theta$ that produces it by taking its tangent and extracting the intercept with the $a$-axis. Convexity of the JNR ensures that the mapping from points on the boundary to such intercepts is injective. Therefore, the entire upper boundary can be mapped out by scanning the parameter $\theta$ and collecting the corresponding generalized eigenvectors.

We now revisit the non-degeneracy assumption. At a point $\theta=\theta_\ast$ where the largest generalized eigenvalue is degenerate, we would have at least two eigenvectors $v_{\theta_\ast}^{(1)}$ and $v_{\theta_\ast}^{(2)}$. The corresponding points $\left(a_{\theta_\ast}^{(1)}, b_{\theta_\ast}^{(1)}\right)$ and $\left(a_{\theta_\ast}^{(2)}, b_{\theta_\ast}^{(2)}\right)$ would generally be disjoint, and at $\theta=\theta_\ast-\epsilon$ and $\theta=\theta_\ast+\epsilon$ the corresponding $(a_\theta,b_\theta)$ point jumps, seemingly breaking the conclusion that the boundary of the JNR can be smoothly mapped out by varying $\theta$. To resolve this, notice that an arbitrary linear combination of degenerate generalized eigenvectors would also be a generalized eigenvector. Therefore, consider
\begin{equation}
v_{\theta_\ast}^{(\phi)} = \cos \phi v_{\theta_\ast}^{(1)} + \sin\phi v_{{\theta_\ast}}^{(2)},
\end{equation}
which is unit-norm if we take $\left(v_{\theta_\ast}^{(1)},v_{\theta_\ast}^{(2)}\right)$ to be orthonormal, one would have
\begin{align}
a_{\theta_\ast}^{(\phi)} & = \left(v_{\theta_\ast}^{(\phi)}\right)^\dagger A v_{\theta_\ast}^{(\phi)} \\& = \cos^2\phi a_{\theta_\ast}^{(1)} + \sin^2\phi a_{\theta_\ast}^{(2)} + \sin 2\phi  \re \left[\left(v_{\theta_\ast}^{(1)}\right)^\dagger A v_{\theta_\ast}^{(2)}\right]; \nonumber \\
b_{\theta_\ast}^{(\phi)} & = \left(v_{\theta_\ast}^{(\phi)}\right)^\dagger B v_{\theta_\ast}^{(\phi)} \\ &= \cos^2\phi b_{\theta_\ast}^{(1)} + \sin^2\phi b_{\theta_\ast}^{(2)} + \sin 2\phi  \re \left[\left(v_{\theta_\ast}^{(1)}\right)^\dagger B v_{\theta_\ast}^{(2)}\right] \nonumber.
\end{align}
Notably, since each $v_{\theta_\ast}^{(\phi)}$ is a generalized eigenvector, we find that all these points will lie on the line
\begin{equation}
b_{\theta_\ast}^{(\phi)} = \mu_{\theta_\ast}\left(1+\theta_\ast a_{\theta_\ast}^{(\phi)}\right).
\end{equation}
This means that the trajectory of $\left(a_{\theta_\ast}^{(\phi)}, b_{\theta_\ast}^{(\phi)}\right)$ is a line segment connecting $\left(a_{\theta_\ast}^{(1)}, b_{\theta_\ast}^{(1)}\right)$ and $\left(a_{\theta_\ast}^{(2)}, b_{\theta_\ast}^{(2)}\right)$. Therefore, degenerate eigenvalues are not a failure of the algorithm, but a genuine sign that the boundary of the JNR has a straight line segment.

Furthermore, it can be shown that the boundary must be everywhere differentiable provided that at no point is $v_\theta$ a simultaneous eigenvector of $A$ and $B$. Without loss of generality, assume that the non-differentiable point occurs at $a=b=0$. Then, a vector $v$ exists such that $v^\dagger A v = v^\dagger B v=0$. Due to the convexity of the region, the non-differentiable point must be isolated, and left- and right-derivatives at this point must exist. We can further linearly recombine $A$ and $B$ (\textit{i.e.}, apply a rotation in the $a$-$b$ plane) such that the two segments of the boundary have
\begin{equation}
\left.\frac{\mathrm db}{\mathrm da}\right|_{\pm} = \pm \tan \alpha,
\end{equation}
with $W$ lying on the $a\geq 0$ side. Due to convexity, this implies that $W$ must entirely lie in the region bounded by $a\geq 0$, $-\tan(\alpha) a \leq b \leq \tan(\alpha
)a.$ As a result, for any $|\beta| < \frac{\pi}{2}-\alpha$, the matrix
\begin{equation}
H_\beta = \cos\beta A + \sin\beta B
\end{equation}
must be positive-semidefinite. Furthermore, $v^\dagger H_\beta v=0$, which implies that $H_\beta v=0$ for any $\beta$. Therefore, one must have $Av=Bv=0$.

\section{The role of spectrum generating algebra}
\label{sec:sga}

\begin{figure*}[!htbp]
    \centering
    \includegraphics[width=\textwidth]{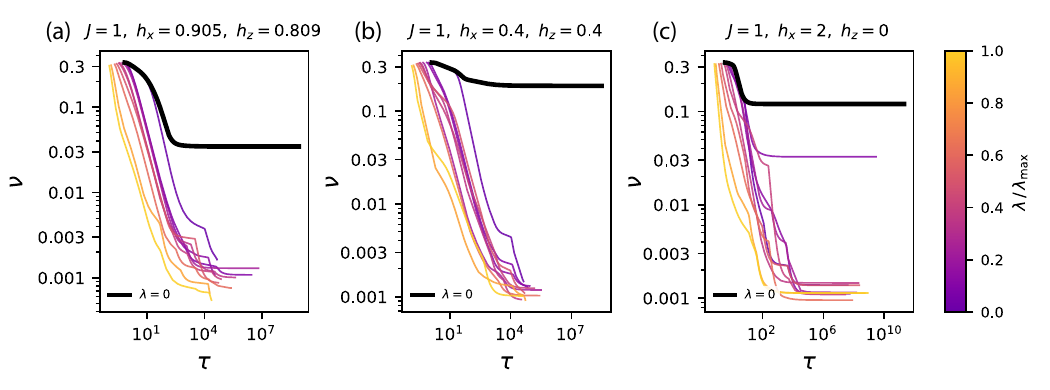}
  \caption{Numerical evidence for $\nu(\tau) = \nu^{(0)}(\tau)$. $\nu^{(0)}(\tau)$, plotted in thick line, is shown to dominate over all $\nu^{(\lambda)}(\tau)$ for all $\lambda\neq 0$ at all available $\tau$ across two MFIMs that are (a) chaotic, (b) prethermal, and (c) a TFIM, respectively. All are on $L=12$ PBC chains with $n=4$.}
    \label{fig:lambda}
\end{figure*}

In this Appendix, we will discuss in detail the role of (approximate) spectrum generating algebra operators in the thermalization bounds. We emphasize two facts: (1) absence of simple SGA governs equilibration while absence of simple IOM governs thermalization on top of that; (2) in most physical settings, it is very likely that non-thermalization is caused by simple IOM instead of simple SGA.

Recall that our infinite-time thermalization bound is obtained by expanding $O$ in the eigenbasis of $[H,\cdot]$, and exploiting the fact that any oscillatory function vanishes in infinite-time average. We have
\begin{equation}
\lim_{T\to\infty} \int_0^T \frac{\mathrm dt}{T} \|O(t)\|_\calE^2 = \sum_\lambda \|O_\lambda\|_\calE^2.
\end{equation}
In Subsection~\ref{subsec:thermalization-and-slow-operators}, we then proceed to bound $\|O_\lambda\|_\calE^2 \leq \nu(\infty) \|O_\lambda\|^2$, where $\nu(\infty)$ is the maximum possible value of $\|A\|_\calE^2/\|A\|^2$ for any $A$ that satisfies $[H,A]=\lambda A$ for some $\lambda$. For the sake of simplicity, we will not explicitly mention the orthogonalization against powers of $H$ when it does not play a significant role in the discussion.

As a more fine-grained description, define $\nu^{(\lambda)}(\infty)$ as the maximal value of $\|A\|_\calE^2/\|A\|^2$ for any $A$ that satisfies $[H,A]=\lambda A$ for \textit{a specific} $\lambda$. It is obvious from definition that
\begin{equation}
\nu(\infty) = \max_{\lambda} \nu^{(\lambda)}(\infty).
\end{equation}
While $\nu(\infty)$ characterizes the absence of simple IOMs and SGA, $\nu^{(0)}(\infty)$ specifically characterizes the absence of simple IOMs, and $\nu^{(\neq 0)}(\infty) = \max_{\lambda \neq 0} \nu^{(\lambda)}(\infty)$ characterizes the absence of simple SGA. Then, $\nu(\infty)$ is the larger of $\nu^{(0)}(\infty)$ and $\nu^{(\neq 0)}(\infty)$.

In fact, in the infinite-time case,  $\nu^{(0)}(\infty)$ and $\nu^{(\neq 0)}(\infty)$ play physically distinct roles. Let $O_0 = \overline{O(t)}$ be the long-time average of $O$ (we have used an overline to denote the average $\lim_{T\to\infty} \int_0^T \frac{\mathrm dt}{T} $). If we subtract it from $O$, we get
\begin{align}
\overline{\|O(t)-O_0\|_\calE^2} & = \sum_{\lambda \neq 0}\|O_\lambda\|_\calE^2 \nonumber \\
& \leq \sum_{\lambda \neq 0} \nu^{(\lambda)}(\infty) \|O_\lambda\|^2 \nonumber \\
& \leq \nu^{(\neq 0)}(\infty) \sum_{\lambda \neq 0}  \|O_\lambda\|^2 \nonumber \\
& \leq \nu^{(\neq 0)}(\infty) \|O-O_0\|^2. \label{eq:Ot-minus-O0-average-nu-SGA}
\end{align}
For a given initial state,
\begin{equation}
\langle \psi|O_0|\psi\rangle = \overline{\langle \psi|O(t)|\psi\rangle} = \mathrm{tr}\left[O\overline{|\psi(t)\rangle\langle\psi(t)|}\right].
\end{equation}
This is exactly the diagonal ensemble expectation value, denoted as $\langle O \rangle_\text{diag}(\psi)$, which is guaranteed to coincide with the equilibrium value as long as the system equilibrates, even if the system does not thermalize (recall from Section~\ref{sec:summary} that equilibration is a strictly weaker property than thermalization). Therefore, Eq.~\eqref{eq:Ot-minus-O0-average-nu-SGA} can be rewritten as
\begin{equation}
D_\infty(\calE,O;\langle O \rangle_\text{diag}) \leq \nu^{(\neq 0)}(\infty) \|O-O_0\|^2. \label{eq:sga-thm1-equilibration}
\end{equation}
Therefore, $\nu^{(\neq 0)}(\infty)$, or the presence of simple SGAs, characterizes whether the system \textit{equilibrates}. On the other hand, $\nu^{(0)}$ characterizes whether the diagonal ensemble coincides with the thermal expectation value. In fact,
\begin{multline}
\expt_{\psi\sim\calE} |\langle O\rangle_\text{diag}(\psi) - \Onpow(\psi)|^2 \\ = \|O_0 -O_\parallel^n\|_\calE^2 \leq \nu_n^{(0)}(\infty) \|O_0\|^2. \label{eq:sga-thm1-ensemble}
\end{multline}
Thermalization requires both (a) equilibration and (b) agreement between diagonal ensemble and thermal ensemble. The sum of the two bounds in Eq.~\eqref{eq:sga-thm1-equilibration} and Eq.~\eqref{eq:sga-thm1-ensemble} can be exactly upper bounded by $\max\{\nu_n^{(0)}(\infty),\nu^{(\neq 0)}(\infty)\}\|O\|^2$ (note that $\|O-O_0\|^2+\|O_0\|^2 = \|O\|^2$).

This separation between equilibration and ensemble equivalence is less clean in the finite-time version. However, we further demonstrate that in most physical settings, as far as non-thermalization is concerned, it is sufficient to focus on just the existence of simple approximate IOMs. To be more precise, carrying over the language at infinite time, we define $\nu^{(\lambda)}(\tau)$ to be the maximum $\|A\|_\calE^2 / \|A\|^2$ among all operators that satisfy $\|[H,A]-\lambda A\|/\|A\|\leq 1/\tau$. Under this definition, $\nu(\tau) = \max_\lambda \nu^{(\lambda)}(\tau)$. We claim that there is strong physical reason to believe $\nu(\tau) = \nu^{(0)}(\tau)$. Under this assumption, the thermalization bound can be expressed in terms of $\nu^{(0)}(\tau)$, \textit{i.e.}, the absence of simple IOMs, alone. We observe that this holds numerically across several models of distinct nature in Fig.~\ref{fig:lambda}. We further provide two physical justifications for this: (1) simple SGAs are unlikely to exist in many-body quantum systems, and (2) if they exist, they are often accompanied by IOMs.

For the first argument, notice that an SGA is defined as
\begin{equation}
[H,A] = \lambda A.
\end{equation}
This implies that if $|E\rangle$ is an eigenstate of $H$, then $A$ either annihilates the state, or brings it to another eigenstate of energy $E+\lambda$ (hence the name ``spectrum generating''). Then one of the following must be true: (a) For a considerable proportion of eigenstates in the system, the action of a simple operator creates an excitation that has a fixed energy on top of it. This is a very strong, free-particle-like structure that does not exist even in the Heisenberg spin chain, the prototypical integrable model. (b) $A|E\rangle$ vanishes for all but a small fraction of the eigenstates, which is the case for SGAs in models with quantum many-body scars~\cite{odeaTunnelsTowersQuantum2020,moudgalyaPairingHubbardModels2020,markPairingStatesTrue2020,serbynQuantumManybodyScars2021,renQuasisymmetryGroupsManyBody2021a,renDeformedSymmetryStructures2022,moudgalyaQuantumManybodyScars2022b,desaulesExtensiveMultipartiteEntanglement2022,dengUsingModelsStatic2023,moudgalyaExhaustiveCharacterizationQuantum2024a}. In this case, $A$ would have a small rank. It is proven (see Appendix~\ref{sec:rps-norm-rank}) that an operator whose rank does not scale exponentially with system size must have an exponentially small RPS-norm; therefore, such $A$ cannot be simple.

It is worth noting that there are interacting models where SGAs are known to exist. The most well-known example is Yang's $\eta$-pairing in the Hubbard model~\cite{yangPairingOffdiagonalLongrange1989,yangRemarksGeneralizationsSU2xSU21991,moudgalyaPairingHubbardModels2020,markPairingStatesTrue2020}, defined by
\begin{equation}
H_\text{Hubbard} = -t \sum_{\langle ij\rangle,\sigma} c_{i\sigma}^\dagger c_{j\sigma} + U \sum_i n_{i\uparrow} n_{i\downarrow}.
\end{equation}
The operator
\begin{equation}
\eta^+ = \sum_{i} (-1)^i c_{i\uparrow}^\dagger c_{i\downarrow}^\dagger
\end{equation}
exactly satisfies
\begin{equation}
\left[H_\text{Hubbard},\eta^+\right] = U \eta^+.
\end{equation}
In this case, $\eta^+$ is part of an $\mathrm{SU}(2)$ algebra of operators (hence the ``algebra'' in ``spectrum generating algebra''), with $\eta^-=(\eta^+)^\dagger$ and $\eta^z = \frac{1}{2}[\eta^+,\eta^-]$ also satisfying corresponding algebraic relations with the Hamiltonian. In particular, $[H,\eta^z]=0$, meaning that the SGA $\eta^+$ is accompanied by an IOM. Generally, whenever $[H,A]=\lambda A$, $AA^\dagger $and $A^\dagger A$ would commute with the Hamiltonian; and if $A$ is a simple operator, $AA^\dagger$ and $A^\dagger A$ will also be. Quantitatively (see Appendix~\ref{sec:norm-properties} for a proof),
\begin{equation}
\|A^\dagger A\|_\calE \geq \|A\|_\calE^2
\end{equation}
for normalized ensembles. If the SGA is approximate, the approximate IOM's conservedness can also be bounded by
\begin{multline}
\|[H,A^\dagger A]\| = \|A^\dagger \left([H,A]-\lambda A\right) + \left([H,A^\dagger]+\lambda A^\dagger\right)A\|\\
\leq 2 \sigma_{\max}(A) \|[H,A] - \lambda A\|,
 \end{multline}
 where $\sigma_{\max}$ is the largest singular value of $A$.
 Therefore, $A^\dagger A$ (or $AA^\dagger$) is likely to be a simple (approximate) IOM.
 
We should note that this argument is not mathematically rigorous in several ways: even as its argument stands, it does not prove that $\nu^{(\lambda)}(\tau)$ is less than $\nu^{(0)}(\tau)$. Moreover, it might happen that $\|A^\dagger A\|_\calE$ is large, but $A^\dagger A$ has large overlaps with the powers of $H$, and after orthogonalization it becomes a large operator. This requires $A^\dagger A$ to be very close to a low-order polynomial of $H$, which is physically unlikely; should such an operator exist, it would make an interesting case study. For most physically realistic situations, we settle for the conclusion that studying the approximate-IOM SSOs (\textit{i.e.}, $\nu^{(0)}(\tau)$) suffices in terms of determining thermalization.

\section{RPS-norm of low-rank operators}
\label{sec:rps-norm-rank}

In this Appendix, we prove the following theorem: for any operator $O$, one has
\begin{equation}
\|O\|_\text{RPS}^2 \leq \left(\frac{2}{d(d+1)}\right)^L (\tr |O|)^2, \label{eq:RPS-norm-bound-by-trace-norm}
\end{equation}
in which $|O| = \sqrt{O^\dagger O}$, so that
\begin{equation}
\tr |O| = \sum_i \sigma_i
\end{equation}
is the trace norm (also known as nuclear norm) of $O$ ($\{\sigma_i\}$ is the list of singular values of $O$.) Noting that
\begin{equation}
\|O\|^2 = \frac{1}{D} \sum_i \sigma_i^2,
\end{equation}
define the scale-free ratio
\begin{equation}
r_\text{eff}(O) = \frac{\left(\sum_i \sigma_i\right)^2}{\sum_i \sigma_i^2},
\end{equation}
we would have
\begin{equation}
\|O\|_\text{RPS}^2 \leq \left(\frac{2}{d+1}\right)^L r_\text{eff}(O) \|O\|^2.
\end{equation}
It is standard to show that
\begin{equation}
r_\text{eff}(O) \leq \text{rank}(O),
\end{equation}
with equality achieved if all non-zero singular values of $O$ are identical. Therefore, the theorem implies that for any operator $O$ whose rank does not scale exponentially with system size, $\|O\|_\text{RPS}/\|O\|$ must be exponentially small.

We will offer two proofs of Eq.~\eqref{eq:RPS-norm-bound-by-trace-norm}, which illustrate slightly different physics.

\textit{Proof 1.} Take the singular value decomposition of $O$ as
\begin{equation}
O = \sum_i \sigma_i |u_i\rangle \langle v_i|,
\end{equation}
where $\{|u_i\rangle\}$ and $\{|v_i\rangle\}$ are, respectively, orthonormal sets of states. Since any ensemble variance norm satisfies the triangle inequality, we have
\begin{equation}
\|O\|_\text{RPS} \leq \sum_i \sigma_i \||u_i\rangle \langle v_i| \|_\text{RPS} \leq \tr |O|\cdot \sup_{u,v} \||u\rangle \langle v| \|_\text{RPS}.
\end{equation}
Therefore, the task reduces to bounding the RPS-norm for rank-one operators. In fact, it can be further reduced to rank-one projectors, as
\begin{multline}
\||u\rangle \langle v| \|_\text{RPS}^2 = \expt_{\psi\sim \text{RPS}} |\langle \psi|u\rangle|^2 |\langle \psi|v\rangle|^2 \\
= \langle |u\rangle\langle u|,|v\rangle\langle v|\rangle_\text{RPS}
\leq \| |u\rangle \langle u|\|_\text{RPS}  \| |v\rangle \langle v|\|_\text{RPS} \\
\leq \left(\max\{\| |u\rangle \langle u|\|_\text{RPS},  \| |v\rangle \langle v|\|_\text{RPS}\}\right)^2.
\end{multline}
For rank-one projectors, using the same formulation as in Appendix~\ref{sec:gen-op-norm-as-evn},
\begin{equation}
\| |u\rangle\langle u| \|_\text{RPS}^2 = \langle u|\langle u| \rho^{(2)}_\text{RPS} |u\rangle|u\rangle,
\end{equation}
where $\rho^{(2)}_\text{RPS}$ is a density matrix defined on $\calH\otimes \calH$ by
\begin{equation}
\rho^{(2)}_\text{RPS} = \expt_{\psi\sim\text{RPS}} |\psi\rangle|\psi\rangle \langle\psi| \langle \psi|.
\end{equation}
Since RPS is the tensor product of $L$ copies of $d$-dimensional Haar random states, $\rho^{(2)}_\text{RPS}$ is the tensor product of $L$ copies of $\rho^{(2)}_\text{Haar}$ on a $d$-dimensional qudit. It is straightforward that
\begin{equation}
(\rho^{(2)}_\text{Haar})_{iji^\prime j^\prime} = \frac{1}{d(d+1)}\left[ \delta_{ii^\prime}\delta_{jj^\prime} + \delta_{ij^\prime} \delta_{ji^\prime}\right].
\end{equation}
Equivalently~\cite{cotlerEmergentQuantumState2023},
\begin{equation}
\rho^{(2)}_\text{Haar} = \frac{2}{d(d+1)}\Pi_\text{sym},
\end{equation}
where $\Pi_\text{sym}$ is the projector onto the subspace of $\calH \otimes \calH$ that is symmetric with respect to the swapping of the two components. In particular, the spectral norm
\begin{equation}
\left\| \rho^{(2)}_\text{Haar} \right\|_\infty = \frac{2}{d(d+1)}.
\end{equation}
Therefore, 
\begin{equation}
\left\| \rho^{(2)}_\text{RPS} \right\|_\infty = \left(\frac{2}{d(d+1)}\right)^L,
\end{equation}
and since $|u\rangle |u\rangle$ is normalized, we conclude that
\begin{equation}
\langle u|\langle u| \rho^{(2)}_\text{RPS} |u\rangle|u\rangle \leq \left(\frac{2}{d(d+1)}\right)^L.
\end{equation}
This completes the proof. $\square$

\textit{Proof 2.} Recall that for $O=\sum_P c_P P$,
\begin{equation}
\|O\|_\text{RPS}^2 = \sum_P |c_P|^2 (d+1)^{-|P|}.
\end{equation}
Consider a random subset $A$ of the lattice, such that each site is chosen with probability $\frac{1}{d+1}$. It is straightforward that
\begin{equation}
\|O\|_\text{RPS}^2 = \expt_A\left[\sum_P |c_P|^2 \times 1_{\text{supp}P\subset A}\right]. \label{eq:O-RPS-random-support-A}
\end{equation}
Specifically,
\begin{equation}
\expt_A = (d+1)^{-L} \sum_A d^{|\bar A|} \label{eq:expt-A-as-sum-over-A},
\end{equation}
where $\bar A$ is the complement of $A$. The factor $1_{\text{supp}P\subset A}$ also naturally appears in
\begin{equation}
\tr_{\bar A} O = d^{|\bar A|} \sum_{\substack{P\\\text{supp}P\subset A}} c_P P.
\end{equation}
Therefore,
\begin{equation}
\sum_P |c_P|^2 \times 1_{\text{supp}P\subset A} = d^{-|A|-2|\bar A|} \tr_A\left[(\tr_{\bar A} O)^\dagger \tr_{\bar A} O\right].
\end{equation}
This is proportional to the Frobenius norm of the reduced operator $\tr_{\bar A}O$. It is standard to show that
\begin{equation}
\tr[O^\dagger O] \leq (\tr|O|)^2
\end{equation}
for any operator $O$, so
\begin{equation}
\sum_P |c_P|^2 \times 1_{\text{supp}P\subset A} \leq d^{-L-|\bar A|} \left(\tr_A\left[|\tr_{\bar A} O|\right]\right)^2. \label{eq:suppA-sum-bound-by-polar-decomposition}
\end{equation}
We need to bound the quantity $\tr_A\left[|\tr_{\bar A} O|\right]$. Using the fact that
\begin{equation}
\tr | O | = \sup_X \frac{|\tr[OX]|}{\|X\|_\infty},
\end{equation}
where $\|\cdot \|_\infty$ is the spectral norm, \textit{i.e.}, the largest singular value. Then,
\begin{multline}
\tr_A\left[|\tr_{\bar A} O|\right] = \sup_{\text{supp} X\subset A} \frac{|\tr_A[X\tr_{\bar A}O]|}{\|X\|_\infty} \\
= \sup_{\text{supp} X\subset A} \frac{|\tr[(X\otimes I_{\bar A})O]|}{\|X\|_\infty} \leq \sup_{X} \frac{|\tr[XO]|}{\|X\|_\infty} = \tr |O|. \label{eq:trace-norm-under-partial-trace-bound}
\end{multline}
Combining Eqs.~\eqref{eq:O-RPS-random-support-A}, \eqref{eq:expt-A-as-sum-over-A}, \eqref{eq:suppA-sum-bound-by-polar-decomposition}, and \eqref{eq:trace-norm-under-partial-trace-bound}, we conclude that
\begin{align}
\|O\|_\text{RPS}^2 & \leq  [d(d+1)]^{-L} \sum_A (\tr |O|)^2 \nonumber \\
& = \left(\frac{2}{d(d+1)}\right)^L (\tr |O|)^2.
\end{align}
This is the desired result. $\square$

\section{RPS-norm of $K_n$}
\label{sec:rps-Kn}

In this Appendix, we discuss the RPS-norm of $K_n$ as defined in Eq.~\eqref{eq:K-np1-definition},
\begin{equation}
K_{n} = \frac{\calP_{n-1}[H^{n}]}{\|\calP_{n-1}[H^{n}]\|},
\end{equation}
where $\calP_{n-1}$ means orthogonalization against $\{I,H,\dots,H^{n-1}\}$, in the large-$L$ limit. Let the Hamiltonian have the form
\begin{equation}
H = \sum_x h_x,
\end{equation}
where $h_x$ is the energy density near site $x$. We assume that each $h_x$ is local, \textit{i.e.}, supported on $\calO(1)$ consecutive sites. Technically, we require the support size of $h_x$ and $\|h_x\|$ to both be uniformly bounded. We do not assume translational invariance, and only require that $\|H\|^2$ and $\|H\|^2_\text{RPS}$ scale as $\Theta(L)$. Without loss of generality, we also assume that each $h_x$ is traceless. By definition,
\begin{equation}
H^k = \sum_{x_1,\dots,x_k} h_{x_1} \dots h_{x_k}. \label{Eq:Hk-as-sum-of-local-terms}
\end{equation}
Let us define
\begin{equation}
H_c^k = \sum_{\substack{x_1,\dots,x_k\\\text{ non-overlapping}}} h_{x_1} \dots h_{x_k}.
\end{equation}
where non-overlapping means that the supports of $h_{x_i}$ are mutually disjoint. Note that we use $H_c^k$ as a single symbol, and it is not implied that $H_c^k$ is the $k$-th power of some operator. We will first show that $\|H_c^n\|_\text{RPS}/\|H_c^n\| \approx (\|H\|_\text{RPS}/\|H\|)^n$, and then show that $\calP_{n-1}[H^n]\approx H_c^n$.

Before delving into the discussions, we first notice a general pattern in evaluating the norms of the operators involved. By expanding norms in terms of inner products, all evaluations boil down to an expression of the form $\langle h_{x_1} \dots h_{x_m}\rangle$. Note that if the supports of two operators are disjoint, then the normalized trace factorizes as $\langle AB\rangle = \langle A \rangle \langle B\rangle$. Therefore, we can partition the set $\{x_1,\dots, x_{m}\}$ based on an equivalence relation defined by the overlapping support of the corresponding $h_x$. Then,
\begin{equation}
\langle h_{x_1}\dots h_{x_{m}}\rangle = \prod_{\{I_\alpha\}}\left\langle \prod_{i\in I_\alpha} h_{x_i}\right\rangle, \label{eq:hkplusell-decomposition-sum}
\end{equation}
where $\{I_\alpha\}_{\alpha=1}^P$ denotes the partition of $\{1,\dots,m\}$. We observe that if any $|I_\alpha|=1$, then the entire expression evaluates to zero, since each individual $h_x$ is traceless. Therefore, each $I_\alpha$ must contain at least two elements.

With this in mind, we can compute that $\|H_c^k\| \approx \|H\|^k \sim \Theta(L^{k/2})$. In fact, we have
\begin{equation}
\|H_c^k\|^2 = \sum_{\substack{x_1,\dots,x_k\\x_1^\prime,\dots,x_k^\prime}} \langle h_{x_1} \dots h_{x_k}, h_{x_1^\prime} \dots h_{x_k^\prime}\rangle. \label{eq:Hck-norm-full-sum}
\end{equation}
It is understood that $\{x_1,\dots,x_k\}$ and $\{x_1^\prime,\dots,x_k^\prime\}$ are both non-overlapping. We can see that a $\Theta(L^k)$ contribution comes from a factorized contribution,
\begin{equation}
\sum_{\substack{x_1,\dots,x_k\\x_1^\prime,\dots,x_k^\prime}} \langle h_{x_1} h_{x_1^\prime}\rangle  \dots \langle h_{x_k} h_{x_k^\prime}\rangle \label{eq:factorized-sum-Hck}.
\end{equation}
This factorization is valid if $x_i$ overlaps with $x_i^\prime$ and nothing else; equivalently, it corresponds to all contributions where all $|I_\alpha|=2$, in the notation of the previous paragraph. Up to an $\calO(1/L)$ relative correction, Eq.~\eqref{eq:factorized-sum-Hck} evaluates to $\|H\|^{2k} \sim \Theta(L^k)$. The rest of the contributions in Eq.~\eqref{eq:Hck-norm-full-sum} correspond to contributions where at least some $|I_\alpha|>2$. However, this means that $P \leq k-1$. Since each $I_\alpha$ cluster contains terms around an $\calO(1)$ span, factors of $L$ come only from choosing the center of each cluster, and therefore the sum scales at most as $\calO(L^P)$, which is bounded by $\calO(L^{k-1})$. Therefore, the leading order contribution to $\|H_c^k\|^2$ comes exactly from the terms as in Eq.~\eqref{eq:factorized-sum-Hck}, up to a permutation of the $k$ subscripts. We hereby conclude that
\begin{equation}
\|H_c^k\| = \sqrt{k!}\|H\|^{k} + \calO_r(L^{-1}). \label{eq:Hc-k-norm}
\end{equation}
We have used the notation $+\calO_r(L^{-1})$ to mean relative error with respect to the previous term. We further note that all arguments here remain intact by replacing the Hilbert-Schmidt inner product with the RPS inner product, as it also factorizes across non-overlapping patches. Therefore, we have also just shown that
\begin{equation}
\| H_c^k\|_\text{RPS} = \sqrt{k!}\|H\|_\text{RPS}^k + \calO_r(L^{-1}).
\end{equation}
This proves that
\begin{equation}
\frac{\|H_c^n\|_\text{RPS}}{\|H_c^n\|} = \left(\frac{\|H\|_\text{RPS}}{\|H\|}\right)^n + \calO_r(L^{-1}).
\end{equation}

With a similar argument, we have
\begin{align}
\langle H^{2k}\rangle & = (2k-1)!! \|H\|^{2k} + \calO_r(L^{-1}),  \label{eq:H-2k-expect} \\
\langle H^{2k+1}\rangle & = \calO(L^k).
\end{align}
The different combinatorial factor in Eq.~\eqref{eq:H-2k-expect} and Eq.~\eqref{eq:Hc-k-norm} is because $\langle H_c^k,H_c^k\rangle$ rules out self-overlap within the two $k$-th power blocks, respectively. This implies that the moments of $H/\|H\|$ are the same as $\mathcal{N}(0,1)$, the Gaussian distribution, up to $\calO(L^{-1/2})$ (more specifically, the odd moments' errors are $\calO(L^{-1/2})$, and the even moments $\calO(L^{-1})$). Therefore, take the (probabilist's) Hermite polynomials $\mathrm{He}_k$~\cite{weisstein_hermite_polynomial}, such that
\begin{equation}
\int \mathrm{He}_k(x) \mathrm{He}_\ell(x) \frac{e^{-\frac{x^2}{2}}}{\sqrt{2\pi}} \mathrm dx = k! \delta_{k\ell},
\end{equation}
then it follows that the set of operators defined by $\mathrm{He}_k\left(\frac{H}{\|H\|}\right)$ are approximately orthogonal up to error $\calO_r(L^{-1/2})$. More accurately, as the order of the error depends on the parity of the power, one has
\begin{multline}
\left\langle \mathrm{He}_k\left(\frac{H}{\|H\|}\right), \mathrm{He}_\ell\left(\frac{H}{\|H\|}\right) \right\rangle \\
= k! \delta_{k\ell} + \begin{cases}\calO(L^{-1/2}),& k+\ell\text{ is odd}; \\ \calO(L^{-1}),& k+\ell\text{ is even}.\end{cases} \label{eq:Hem-Hen-overlap-sizes}
\end{multline}
Since $\mathrm{He}_k$ are monic, it follows that
\begin{equation}
\frac{\calP_{n-1}[H^n]}{\|H\|^n } =\left[\mathrm{He}_n\left(\frac{H}{\|H\|}\right) - \sum_{k=0}^{n-1} a_{nk}  \mathrm{He}_k\left(\frac{H}{\|H\|}\right) \right],
\end{equation}
where $a_{nk}$ are Gram-Schmidt coefficients that, by virtue of Eq.~\eqref{eq:Hem-Hen-overlap-sizes}, have the order
\begin{equation}
a_{nk} = \begin{cases}\calO(L^{-1/2}),& n+k\text{ is odd}; \\ \calO(L^{-1}),& n+k\text{ is even}.\end{cases}
\end{equation}
This implies
\begin{equation}
\|\calP_{n-1}[H^n]\|^2 = n! \|H\|^{2n} + \calO_r(L^{-1}). \label{eq:Pn-1Hn-Fro-norm-final}
\end{equation}

Again using the pairing counting argument, we have
\begin{equation}
\langle H_c^k,H^k\rangle = k! \|H\|^{2k} + \calO_r(L^{-1}), \label{eq:Hc-Hk-overlap-1}
\end{equation}
and for $\ell < k$,
\begin{equation}
\langle H_c^k, H^\ell\rangle = \calO\left(L^{\lfloor\frac{k+\ell}{2}\rfloor}\right). \label{eq:Hc-Hk-overlap-2}
\end{equation}
Now consider $\|\calP_{n-1}[H^n]-H_c^n\|$. The square of it can be expanded as
\begin{multline}
\|\calP_{n-1}[H^n]-H_c^n\|^2 \\
= 2\left(n!\|H\|^{2n} +\calO_r(L^{-1}) - \langle \calP_{n-1}[H^n],H_c^n\rangle\right).
\end{multline}
On the other hand, by Eqs.~\eqref{eq:Hc-Hk-overlap-1} and \eqref{eq:Hc-Hk-overlap-2},
\begin{equation}
\langle \calP_{n-1}[H^n],H_c^n\rangle = n! \|H\|^{2n} + \calO_r(L^{-1}). \label{eq:Pn-1Hn-Hcn-overlap}
\end{equation}
This implies that
\begin{equation}
\|\calP_{n-1}[H^n]-H_c^n\|^2 = \calO(L^{n-1}).
\end{equation}
Using this, we evaluate
\begin{multline}
\|\calP_{n-1}[H^n]\|_\text{RPS}^2 = 2\langle \calP_{n-1}[H^n],H_c^n\rangle_\text{RPS} \\ - \|H_c^n\|_\text{RPS}^2 + \|\calP_{n-1}[H^n]-H_c^n\|_\text{RPS}^2.
\end{multline}
We note that the same argument leading to Eq.~\eqref{eq:Pn-1Hn-Hcn-overlap} produces
\begin{equation}
\langle \calP_{n-1}[H^n],H_c^n\rangle_\text{RPS} = n! \|H\|_{\text{RPS}}^{2n} + \calO_r(L^{-1}).
\end{equation}
This leads us to
\begin{equation}
\|\calP_{n-1}[H^n]\|_\text{RPS} = \sqrt{n!} \|H\|_\text{RPS}^n + \calO_r(L^{-1}).
\end{equation}
Combining with Eq.~\eqref{eq:Pn-1Hn-Fro-norm-final}, we conclude that
\begin{equation}
\|K_n\|_\text{RPS} = \left(\frac{\|H\|_\text{RPS}}{\|H\|}\right)^n + \calO_r(L^{-1}),
\end{equation}
where, to reiterate, ``$+\calO_r(L^{-1})$'' stands for relative error. Note that the result holds in the large-$L$ limit with $n$ being kept fixed. In the case where $n$ scales with $L$, the result serves as a semi-quantitative heuristic rather than a rigorous theorem.

\section{Solution of the Krylov chain model}
\label{sec:krylov-frontier}

In this Appendix, we supply the full technical solution of the $\nu(\tau)$ frontier in the Krylov chain model. Our methodology closely follows Ref.~\cite{parkerUniversalOperatorGrowth2019b}.

\paragraph{General theory of Krylov chain models.} Consider a subspace of the operator space with a basis $(O_n)_{n=0}^{+\infty}$ that is orthonormal under the Hilbert-Schmidt inner product, such that the Hamiltonian commutator $[H,\cdot]$ has matrix elements
\begin{equation}
\langle O_m,[H,O_n]\rangle = \begin{cases}
b_{n+1}, & m=n+1;\\
b_n, & m=n-1; \\
0, & \text{otherwise}.
\end{cases}
\end{equation}
The coefficients $(b_n)_{n=1}^{+\infty}$ are called Lanczos coefficients. We will denote $\calL = [H,\cdot]$, and that $\langle O_m,[H,O_n]\rangle = \calL_{mn}$. To solve the model, note that due to the matrix-element structure of $\calL$, $\calL^n[O_0]$ has to be a linear combination of $\{O_0,\dots,O_n\}$. Define the expansion coefficients as
\begin{equation}
\calL^n[O_0] = \sum_{k=0}^n \mu_{nk} \left(\prod_{\ell=1}^k b_\ell\right) O_k. \label{eq:def-Liouv-expansion-coef}
\end{equation}
The rationale for introducing the $b$-product will be clear shortly. By the definition of the model,
\begin{equation}
\calL[O_k] = b_{k+1} O_{k+1} + b_k O_{k-1}.
\end{equation}
This will hold for all $k\geq 0$ if we understand $b_0=0$. Applying this to Eq.~\eqref{eq:def-Liouv-expansion-coef}, we obtain a recursion relation of $\mu_{nk}$. Instead of explicitly writing down the relations, we abstract the notation with
\begin{equation}
Q_k = \left(\prod_{\ell=1}^k b_\ell\right) O_k,
\end{equation}
so that
\begin{equation}
\calL[Q_k] = Q_{k+1} + b_k^2 Q_{k-1}. \label{eq:Liouv-on-Qk-recursion}
\end{equation}
One can find another family of objects satisfying the same linear algebraic structure. Define $\calL^{(k)}$ as the $k \times k$ matrix obtained by restricting $\calL$ onto the basis $(O_0,\dots,O_{k-1})$, and let
\begin{equation}
q_k(z) = \det[z-\calL^{(k)}].
\end{equation}
$q_k(z)$ is a monic degree-$k$ polynomial. By convention, $q_0(z)=1$. The last row (column) of $\calL^{(k)}$ would contain only one non-zero element. Expanding the determinant with respect to that, one obtains a recursion relation
\begin{equation}
q_{k+1}(z) = z q_k(z) - b_k^2 q_{k-1}(z). \label{eq:q-recursion-relation}
\end{equation}
This can be rewritten as
\begin{equation}
zq_k(z) = q_{k+1}(z) + b_k^2 q_{k-1}(z).
\end{equation}
This has the exact same structure as Eq.~\eqref{eq:Liouv-on-Qk-recursion} if we identify $q_k \leftrightarrow Q_k$, and $z\times \leftrightarrow \calL$. As a result, the same set of coefficients as in Eq.~\eqref{eq:def-Liouv-expansion-coef} would also appear in the expansion
\begin{equation}
z^n = \sum_{k=0}^n \mu_{nk} q_k(z).
\end{equation}
The left-hand side should be understood as acting $(z\times)^n$ on $q_0(z)=1$. In particular, this implies that
\begin{equation}
q_k(\calL)[O_0] = Q_k. \label{eq:Qk-from-polynomial-action-of-L}
\end{equation}
By the orthogonality of $Q_k$, we also obtain
\begin{equation}
\int_{-\infty}^{+\infty} q_k(z) q_{k^\prime}(z) \rho(z) \mathrm dz = B_k^2 \delta_{kk^\prime},
\end{equation}
where
\begin{equation}
B_k = \|Q_k\| =\prod_{\ell=1}^k b_\ell,
\end{equation}
and the density $\rho(z)$ satisfies
\begin{equation}
\int_{-\infty}^{+\infty} f(z) \rho(z) \mathrm dz = \langle O_0,f(\calL)[O_0]\rangle.
\end{equation}
Its functional form can be obtained from
\begin{equation}
\rho(z) = \frac{1}{2\pi}\int_{-\infty}^{+\infty} e^{-iz t} \langle O_0,e^{i\calL t}[O_0]\rangle \mathrm dt. \label{eq:spectral-density-from-FT-autocorr}
\end{equation}
The autocorrelator $\langle O_0,e^{i\calL t}[O_0]\rangle$ can be evaluated by noting that
\begin{equation}
\langle O_0,\calL^n[O_0]\rangle = \mu_{n0}
\end{equation}
from Eq.~\eqref{eq:def-Liouv-expansion-coef}, so
\begin{equation}
\langle O_0,e^{i\calL t}[O_0]\rangle = \sum_{n=0}^{+\infty} \frac{\mu_{n0}}{n!}(it)^n. \label{eq:O0-autocorr-from-mu-genf}
\end{equation}

\paragraph{$\nu$ and $\tau$ in the Krylov chain model.} From Eq.~\eqref{eq:Qk-from-polynomial-action-of-L}, it is clear that any operator within this subspace can be represented as a function of $ \calL$ acting on $O_0$. In general, let this function be $\tilde f$, and define the operator as
\begin{equation}
O_f = \tilde f(\calL)[O_0].
\end{equation}
By Fourier transform, this would be equivalent to a real-time modulation of $O_0$,
\begin{equation}
O_f = \int_{-\infty}^{+\infty} f(t) O_0(t)\mathrm dt,
\end{equation}
where $O_0(t) = e^{i\calL t}[O_0]$, and
\begin{equation}
f(t) = \frac{1}{2\pi}\int_{-\infty}^{+\infty} \tilde f(\omega) e^{-i\omega t} \mathrm d\omega
\end{equation}
is the inverse Fourier transform of $\tilde f$.

We can then compute $\|O_f\|$ and $\|[H,O_f]\|$. By definition,
\begin{align}
\|O_f\|^2 & = \int_{-\infty}^{+\infty} |\tilde f(\omega)|^2 \rho(\omega) \mathrm d\omega, \\
\|[H,O_f]\|^2 & = \int_{-\infty}^{+\infty} \omega^2 |\tilde f(\omega)|^2 \rho(\omega) \mathrm d\omega.
\end{align}
We have renamed the argument $z$ to $\omega$ so as to emphasize its role as the spectral decomposition. This allows us to obtain $\tau$ from $\tilde f$ straightforwardly (assuming knowledge of $\rho$). To obtain $\nu$, we have to know the matrix element of the Pauli-size superoperator $\calS$ in the Krylov basis. We make the simplifying assumption that
\begin{equation}
\langle O_m,\calS[O_n]\rangle = (\alpha n+1) \delta_{mn},
\end{equation}
where $\alpha$ is some constant. We have set $O_0$ to have size $1$ instead of $0$, since $O_0$ cannot be the identity operator if it is the root of a Krylov subspace. Therefore, the RPS-norm kernel $\calM_\text{RPS} = 3^{-\calS}$ would have matrix elements
\begin{equation}
\langle O_m, \calM_\text{RPS}[O_n]\rangle = \delta_{mn} e^{-\mu n}/3,
\end{equation}
where $\mu=\alpha \log 3$. In terms of $\tilde f$, this can be expressed as
\begin{equation}
\|O_f\|_\text{RPS}^2 = \sum_{k=0}^{+\infty} \frac{e^{-\mu k}}{3 B_k^2}\left|\int_{-\infty}^{+\infty} \tilde f(\omega) q_k(\omega) \mathrm \rho(\omega) d\omega\right|^2.
\end{equation}
We see that all three quantities quadratic in $O_f$ are also quadratic in $\tilde f$. This is expected, as $\tilde f$ is a linear parametrization of $O_f$. We can hereby treat the $\nu(\tau)$ frontier as a quadratic optimization over $\tilde f$. Specifically, consider the space of functions defined by the inner product
\begin{equation}
\langle f,g\rangle = \int_{-\infty}^{+\infty} f(\omega)^\ast g(\omega) \rho(\omega) \mathrm d\omega,
\end{equation}
then we have
\begin{equation}
\nu_{O_f} = \frac{\langle \tilde f, \calM_\mu[\tilde f]\rangle}{3 \|\tilde f\|^2},\quad \tau_{O_f}^{-2} = \frac{\langle \tilde f,\omega^2 \tilde f\rangle}{\langle \tilde f,\tilde f\rangle},
\end{equation}
where the kernel $\calM_\mu$ is represented as
\begin{equation}
\calM_\mu[\tilde f] = \sum_{k=0}^{+\infty}\frac{e^{-\mu k}}{B_k^2}\langle q_k,\tilde f\rangle q_k.
\end{equation}
The generalized eigenvalue problem (see Appendix~\ref{sec:jnr}) becomes
\begin{equation}
\calM_\mu[\tilde f] = \lambda(1+\theta^2 \omega^2) \tilde f.
\end{equation}
We have renamed the eigenvalue as $\lambda$ to avoid notation collision, and denoted the Lagrange multiplier as $\theta^2$ for physical convenience.

Notably, consider the large-$\mu$ limit, in which the kernel $\calM$ reduces to
\begin{equation}
\calM_{\infty}[\tilde f] = \langle q_0,\tilde f\rangle q_0,
\end{equation}
\textit{i.e.}, a projector onto $q_0$. Since $q_0$ is the constant function, so is $\calM_\infty[\tilde f]$. In this case, the eigenfunction has the form
\begin{equation}
\tilde f_\theta(\omega) \propto \frac{1}{1+\theta^2 \omega^2}. \label{eq:ftheta-Lorentzian}
\end{equation}
Therefore, the family of operators on the frontier is parametrized by such functions $\tilde f_\theta$ of Lorentzian forms. The $\nu(\tau)$ frontier in this case would be exactly given by the following parametrized curve with respect to $\theta$,
\begin{align}
\nu_\theta & = \frac{\left[\int (1+\theta^2\omega^2)^{-1}\rho(\omega)\mathrm d\omega\right]^2}{3 \int (1+\theta^2\omega^2)^{-2}\rho(\omega)\mathrm d\omega}, \\
\tau_\theta & = \sqrt{\frac{\int (1+\theta^2\omega^2)^{-2}\rho(\omega)\mathrm d\omega}{\int \omega^2(1+\theta^2\omega^2)^{-2}\rho(\omega)\mathrm d\omega}}.
\end{align}

\paragraph{Numerical solution of the $\nu(\tau)$ frontier.} For general $\mu<\infty$, it is also possible to perform an approximate numerical solution of the $\nu(\tau)$ frontier. In general, a numerical solution has to make a truncation to the space. However, it is known that directly truncating the Krylov space yields very poor convergence~\cite{parkerUniversalOperatorGrowth2019b}. In general, consider expanding an operator as
\begin{equation}
O = \sum_{k=0}^{+\infty} c_k O_k.
\end{equation}
It then follows that
\begin{align}
\|O\|^2 & = \sum_k |c_k|^2 = \|c\|^2, \\
\|[H,O]\|^2 & = c^\dagger \calL^2 c, \\
\|O\|_\text{RPS}^2 & = \frac{1}{3} \sum_k e^{-\mu k} |c_k|^2.
\end{align}
The operator on the frontier should satisfy
\begin{equation}
\calM_\mu c = \lambda(1+\theta^2 \calL^2)c.
\end{equation}
To improve the convergence rate of truncation, consider defining $\tilde c = \sqrt{\calM_\mu} c/\lambda$. The eigenvalue equation then becomes
\begin{equation}
\sqrt{\calM_\mu} \calR \sqrt{\calM_\mu} \tilde c = \lambda \tilde c, \label{eq:optimization-MRM-tildec}
\end{equation}
where $\calR = (1+\theta^2 \calL^2)^{-1}$. Notably, the matrix elements of $\calR$ can be evaluated by
\begin{equation}
\calR_{mn} = \frac{1}{B_m B_n}\int_{-\infty}^{+\infty} \frac{q_m(\omega)q_n(\omega)}{1+\theta^2 \omega^2} \rho(\omega)\mathrm d\omega. \label{eq:Rmn-matrix-element-from-integral}
\end{equation}
We can use this to obtain a truncated version of the matrix $\calR$ and solve Eq.~\eqref{eq:optimization-MRM-tildec} within this truncated subspace. Let the truncation parameter be $N$, and define $\calR_N$ as the truncated $\calR$ with size $N\times N$ (so the indices range from $0$ to $N-1$), with $\tilde c_N$ the solution to the truncated eigenvalue problem. To recover $c$ from the solved $\tilde c$, we compute
\begin{equation}
c = \calR \sqrt{\calM_\mu} \tilde c_N. \label{eq:invert-c-from-tilde-c}
\end{equation}
This is equivalent to directly inverting $c =  \lambda \calM_\mu^{-1/2} \tilde c$ if the solution of $\tilde c$ is not truncated; however, since $\tilde c$ is truncated, using Eq.~\eqref{eq:invert-c-from-tilde-c} gives better accuracy. With this, we can compute
\begin{align}
\|O\|^2 & = \tilde c_N^\dagger \sqrt{\calM_\mu} \calR^2 \sqrt{\calM_\mu} \tilde c_N, \label{eq:ctilde-nu-tau-1} \\
\|[H,O]\|^2 & = \tilde c_N^\dagger \sqrt{\calM_\mu} \calR^2 \calL^2 \sqrt{\calM_\mu} \tilde c_N, \label{eq:ctilde-nu-tau-2} \\
3 \|O\|_\text{RPS}^2 & = \tilde c_N^\dagger \sqrt{\calM_\mu} \calR \calM_\mu \calR \sqrt{\calM_\mu} \tilde c_N. \label{eq:ctilde-nu-tau-3}
\end{align}
To compute $\|O\|$ and $\|[H,O]\|$, we use the same method behind Eq.~\eqref{eq:Rmn-matrix-element-from-integral} to obtain matrix elements:
\begin{align}
(\calR^2)_{mn} & =  \frac{1}{B_m B_n}\int_{-\infty}^{+\infty} \frac{q_m(\omega)q_n(\omega)}{(1+\theta^2 \omega^2)^2} \rho(\omega)\mathrm d\omega, \\
(\calR^2\calL^2)_{mn} & =  \frac{1}{B_m B_n}\int_{-\infty}^{+\infty} \frac{\omega^2 q_m(\omega)q_n(\omega)}{(1+\theta^2 \omega^2)^2} \rho(\omega)\mathrm d\omega.
\end{align}
This renders Eqs.~\eqref{eq:ctilde-nu-tau-1} and \eqref{eq:ctilde-nu-tau-2} finite computable sums. To compute $\|O\|_\text{RPS}$, note that it has the form
\begin{multline}
3 \|O\|_\text{RPS}^2  = \sum_{m,n=0}^{N-1} \Bigg[\left(\sqrt{\calM_\mu}\tilde c_N\right)_m^\ast \left(\sqrt{\calM_\mu}\tilde c_N\right)_n \\ \sum_{k=0}^{+\infty} e^{-\mu k} \calR_{mk}\calR_{kn} \Bigg].
\end{multline}
The sum over $k$ decays exponentially with $k$ and can be truncated at some $K$ with error exponentially small in $K$. Notably, if we choose $K=N-1$, it would be equivalent to replacing $\calR$ by $\calR_N$ in Eq.~\eqref{eq:ctilde-nu-tau-3}. This leads to
\begin{equation}
3\|O\|_\text{RPS}^2 \approx \lambda^2 \tilde c_N^\dagger \tilde c_N.
\end{equation}

\paragraph{The $\mathrm{su}(1,1)$ solvable family.} The Krylov chain is solvable if the Lanczos coefficients are chosen as
\begin{equation}
b_n = \sqrt{n(n+1)}. \label{eq:choice-of-bn}
\end{equation}
It has been known~\cite{parkerUniversalOperatorGrowth2019b} that with this choice of $b_n$, the coefficients $\mu_{nk}$ can be given by the following generating function:
\begin{equation}
\sum_{n=0}^{+\infty}\frac{\mu_{nk}}{n!} \tau^n = \frac{1}{k!}(\sec\tau)^2 (\tan\tau)^k.\label{eq:mu-moments-of-syk-krylov}
\end{equation}
The polynomials $q_k(z)$ themselves also have a generating function:
\begin{equation}
\sum_{k=0}^{+\infty} q_k(z) \frac{\tau^k}{k!} = \frac{e^{z\arctan\tau}}{1+\tau^2}.
\end{equation}
Inserting Eq.~\eqref{eq:mu-moments-of-syk-krylov} into Eq.~\eqref{eq:O0-autocorr-from-mu-genf}, we immediately see that
\begin{equation}
\langle O_0,O_0(t)\rangle = (\mathrm{sech}(t))^2.
\end{equation}
The associated spectral density is
\begin{equation}
\rho(\omega) = \frac{\omega}{2\sinh\frac{\pi\omega}{2}}.
\end{equation}

\section{Additional $\nu(\tau)$ numerics}
\label{sec:numerics}

In this Appendix, we present several pieces of numerical results and related comments that complement the discussions in Section~\ref{sec:slow-operators}. All code used and implementation details are available in the online repository~\cite{\CODE}.

\begin{figure*}[!tbp]
    \centering
    \includegraphics{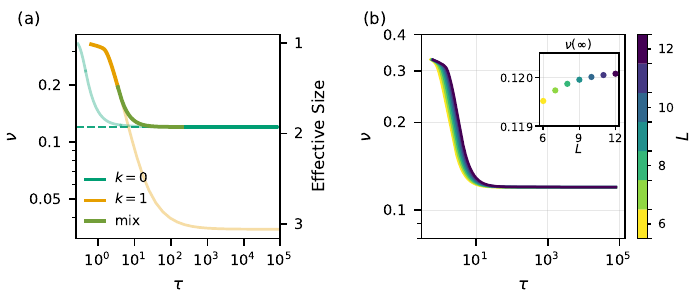}
    \caption{Structure of the $\nu(\tau)$ curve (\textit{cf.} Fig.~\ref{fig:nutau-mfim}) for an integrable spin chain. (a) The numerical $\nu(\tau)$ frontier for a TFIM $H=\sum_{i} JZ_i Z_{i+1} + h_x X_i$ with $J=1$ and $h_x=2$, on an $L=12$ chain with PBC, with $n=4$. The effective operator size is defined as $(-\log_3\nu)$. Color of the line indicates the momentum sector in which the operator lies, with $k=qL/2\pi$, where $q$ is the lattice momentum. The dimmed curves probe the simplest operator at a given $\tau$ within the momentum sector. (b) The same $\nu(\tau)$ frontier as in (a) for TFIM chains of different sizes $L$. Unlike the MFIM in Fig.~\ref{fig:nutau-mfim}, $\nu(\infty)$ does not meaningfully change with $L$.}
    \label{fig:tfim}
\end{figure*}

\begin{figure*}[!tbp]
    \centering
    \includegraphics{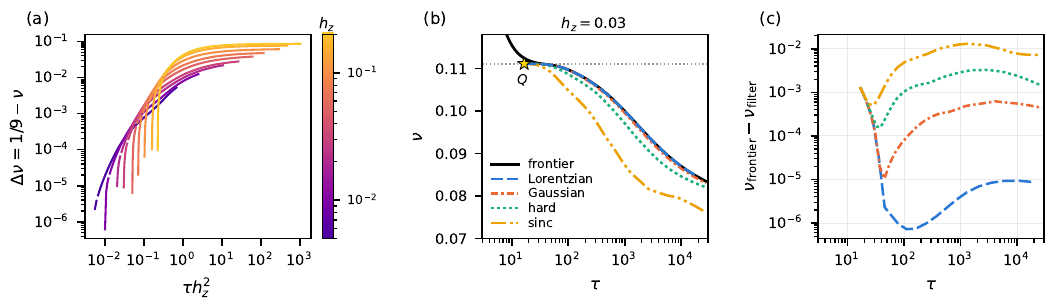}
    \caption{Additional details of the weak-$h_z$ MFIM. The setup is identical to Fig.~\ref{fig:thermdyn}(b). (a) $\Delta\nu = 1/9-\nu(\tau)$ roughly collapses across different models when plotted against $\tau h_z^2$, consistent with the Golden-rule prediction for the thermalization timescale $h_z^{-2}$. (b,c) Reproducing the $\nu(\tau)$ frontier by applying different filters (see Table~\ref{tab:filters}) to the base operator $Q$ (see Fig.~\ref{fig:thermdyn}(b)). The Lorentzian filter is shown to reproduce the frontier the best.}
    \label{fig:wib}
\end{figure*}

\begin{figure*}[!tbp]
    \centering
    \includegraphics{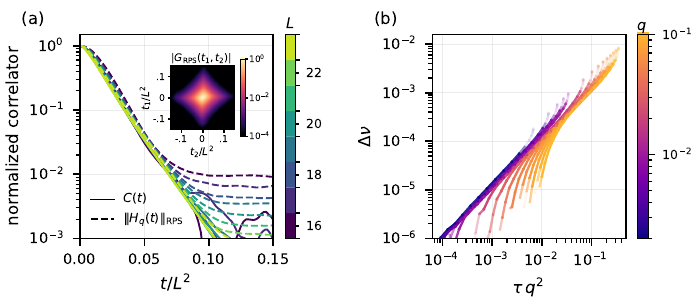}
    \caption{Additional details on the hydrodynamics in the MFIM with $J=1$, $h_x=-1.05$, $h_z=0.5$. (a) Direct comparison of the decay of the autocorrelator $C(t) = \langle H_q,H_q(t)\rangle$ and the decay of the RPS-norm of $H_q(t)$. The two are shown to decay with the same rate proportional to $L^{-2}$ for a finite amount of time. The inset shows the two-time correlator $G_\text{RPS}(t_1,t_2)=\langle H_q(t_1), H_q(t_2)\rangle_\text{RPS}$, corroborating the fact that the correlator lives in the region where both $|t_1|$ and $|t_2|$ lie below the Thouless timescale. (b) The $\nu(\tau)$ frontier within an ansatz of operators carrying a fixed momentum $q$ on an infinitely long chain. The vertical axis shows $\Delta \nu = \nu|_{q=0}-\nu$, the deviation of the frontier from the $q=0$ baseline. The deviation is shown to happen on a timescale proportional to $q^{-2}$. The large-$\tau$ end of each curve is dimmed due to that part of the curve being unstable against increasing the truncation range $M$ and is therefore prone to finite-size effects.}
    \label{fig:hydro}
\end{figure*}

\begin{figure*}[!tbp]
    \centering
    \includegraphics{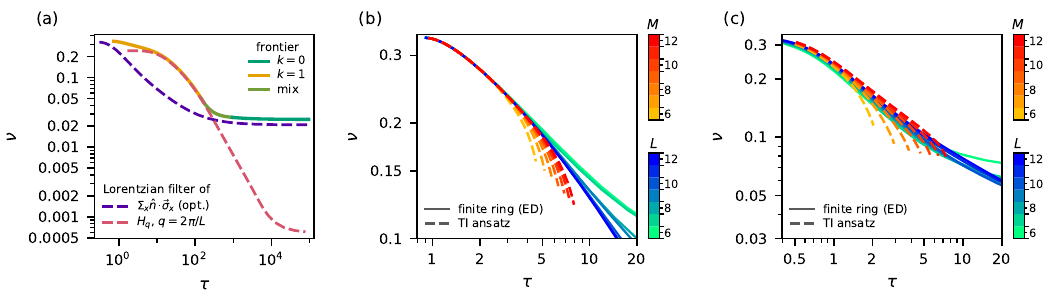}
    \caption{(a) Reproducing the $\nu(\tau)$ frontier of an MFIM with Lorentzian filtering. Setup is identical to Fig.~\ref{fig:nutau-mfim}(b). The intermediate-$\tau$ regime is very well reproduced by filtering the longest-wavelength hydrodynamic mode $H_q$. The large-$\tau$ regime can be roughly reproduced by filtering a global magnetization operator in the direction $\hat n = (0.71,0,-0.7)$. (b,c) Comparison of the $\nu(\tau)$ frontier computed on a finite-size PBC chain and that on an infinitely long chain within an ansatz of translationally invariant sums of local operators, for (b) a KIM, same as in Fig.~\ref{fig:thermdyn}(a), and (c) an MFIM, same as in Fig.~\ref{fig:nutau-mfim}. As the size of the PBC chain $L$ and the support of local operators $M$ increase, the two families of $\nu(\tau)$ tend to match, suggesting that the finite-$L$ curves' behavior in the intermediate-$\tau$ regime survives the thermodynamic limit.}
    \label{fig:chaotic}
\end{figure*}

\begin{figure*}[!tbp]
    \centering
    \includegraphics{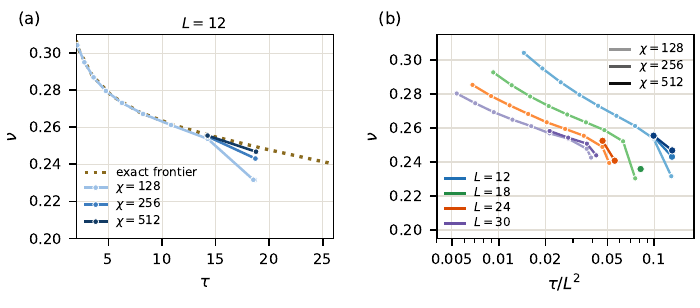}
    \caption{The $\nu(\tau)$ frontier obtained by solving the generalized eigenvalue problem with DMRG for open-boundary MFIM chains with $J=1$, $h_x=-1.05$, $h_z=0.5$. (a) The DMRG frontier of different bond dimensions compared against the ED frontier on an $L=12$ chain. The DMRG frontier virtually agrees with the ED frontier for $\tau\lesssim 15$, but the required bond dimension grows steeply and the algorithm stalls at $\tau\lesssim 20$. (b) The DMRG frontier obtained at different system sizes and bond dimensions. The computationally accessible bond dimension is not able to reach the regime where the SSO is given by the modulated $H_q$ (\textit{cf.} Fig.~\ref{fig:chaotic}(a)) and $\tau\sim L^2$.}
    \label{fig:dmrg}
\end{figure*}

\subsection{$\nu(\tau)$ frontier of an integrable model}

In Fig.~\ref{fig:tfim}, we present the numerically computed $\nu(\tau)$ frontier on an integrable spin chain, namely the TFIM. Compared to Fig.~\ref{fig:nutau-mfim}, we see that the structure of the $\nu(\tau)$ curves looks generally similar; however, the large-$\tau$ regime of the $\nu(\tau)$ curve plateaus quickly in the integrable model, and the plateau value $\nu(\infty)$ is largely constant as $L$ is varied. This corresponds to the fact that the integrable model possesses genuine IOMs that are orthogonal to powers of $H$, whose operator size remains constant in chains of different lengths (up to some minor finite-size corrections).

\subsection{Thermalization timescale of the weak-$h_z$ MFIM}

In Fig.~\ref{fig:wib}(a), we confirm that the $\nu(\tau)$ frontier reflects the Golden-rule thermalization timescale $h_z^{-2}$ in the weak-$h_z$ MFIM discussed in Subsection~\ref{subsec:nutau-and-thermalization-timescale}. The $h_z=0$ model is a TFIM, whose $\nu(\infty)=1/9$, corresponding to the current operator $Q=\sum_{x} Y_x Z_{x+1} - Z_x Y_{x+1}$, whose operator size is exactly $2$. We plot $\Delta\nu(\tau) = 1/9-\nu(\tau)$, tracking the rate at which $\nu$ deviates from the plateau value given by $Q$. The $\Delta \nu$ curves show an approximate collapse upon rescaling $\tau$ by $h_z^2$, corresponding to the expected thermalization timescale.

\subsection{Alternative spectral filtering}

In the main text, we have adopted the Lorentzian spectral filter to construct slow operators. Here, we formulate a general theory for such spectral filtering, and compare the Lorentzian against several other filters. For a general operator $Q$, define the spectral filter as
\begin{equation}
Q_f = \int_{-\infty}^{+\infty} f(t) Q(t) \mathrm dt. \label{eq:Qf-filter-definition}
\end{equation}
If we take the spectral decomposition of $Q$ into the eigenbasis of $[H,\cdot]$, $Q=\sum_\lambda Q_\lambda$, Eq.~\eqref{eq:Qf-filter-definition} can be cast equivalently as
\begin{equation}
Q_f = \sum_\lambda \tilde f(\lambda) Q_\lambda, \label{eq:Qf-filter-definition-lambda-space}
\end{equation}
in which
\begin{equation}
\tilde f(\omega) = \int_{-\infty}^{+\infty} e^{i\omega t} f(t) \mathrm dt.
\end{equation}
We normalize the function $f$ such that $\tilde f(0)=1$. Such an operator-filtering construction has been used in several contexts~\cite{hastingsQuasiadiabaticContinuationQuantum2005,hastingsSpectralGapExponential2006,hastingsLocalityQuantumSystems2010,nachtergaeleQuasilocalityBoundsQuantum2019,kimSlowestLocalOperators2015a,cakanApproximatingLongTime2021a,bachmannGaussianFiltersQuantum2025,chenEfficientExactNoncommutative2025a,lloydQuantumThermalState2026}. We list four different families of filters in Table~\ref{tab:filters}. In Fig.~\ref{fig:wib}(b,c), we compared how well each filter reproduces the $\nu(\tau)$ frontier in the weak-$h_z$ MFIM. The Lorentzian is demonstrated to approach the frontier the best.

\begin{table}[!h]
    \centering
    \begin{tabular}{|c|c|c|} \toprule
        Filter (Parameter) & $f(t)$ & $\tilde f(\omega)$ \\ \toprule
        Lorentzian ($\theta$) & $\frac{e^{-\frac{|t|}{\theta}}}{2\theta}$& $\frac{1}{1+\theta^2\omega^2}$\\ \midrule
        Gaussian ($s$) & $\frac{e^{-\frac{t^2}{s^2}}}{\sqrt{\pi}s}$ & $e^{-\frac{1}{4}s^2\omega^2}$ \\ \midrule
        hard ($\Omega$) & $\frac{\sin \Omega t}{\pi t}$ & $\Theta(\Omega-|\omega|)$ \\ \midrule
        sinc ($T$) & $\frac{\Theta(t)\Theta(T-t)}{T}$ & $e^{\frac{i\omega T}{2}}\mathrm{sinc}\left(\frac{\omega T}{2}\right)$ \\ \bottomrule
    \end{tabular}
    \caption{Four families of spectral filters used to construct slow operators as in Eqs.~\eqref{eq:Qf-filter-definition} and \eqref{eq:Qf-filter-definition-lambda-space}.}
    \label{tab:filters}
\end{table}

\subsection{Additional details of hydrodynamics} 

In Fig.~\ref{fig:hydro}(a), we plotted the decay of $\|H_q(t)\|_\text{RPS}$ (where $q$ is the smallest wave vector) for the same MFIM as in Fig.~\ref{fig:thermdyn}(c). We found that $\|H_q(t)\|_\text{RPS}$ decays at the same rate as the autocorrelator $C(t) = \langle H_q, H_q(t)\rangle$ before finite-size effects kick in. This further corroborates the results in Fig.~\ref{fig:thermdyn}(c), as $C(t)$ is expected to decay with a rate $Dq^2$ and thus produces the Thouless time. We also compute the two-time correlator $\langle H_q(t_1), H_q(t_2)\rangle_\text{RPS}$, as
\begin{multline}
\left\| \int f(t) H_q(t) \mathrm dt \right\|_\text{RPS}^2 \\ = \iint f(t_1)^\ast f(t_2) \langle H_q(t_1), H_q(t_2)\rangle_\text{RPS} \mathrm dt_1 \mathrm dt_2.
\end{multline}
It is shown that $\langle H_q(t_1), H_q(t_2)\rangle_\text{RPS}$ lives in a region where both $|t_1|$ and $|t_2|$ are on the order of the Thouless time. Therefore, the filtered operator is simple if and only if $f(t)$ is supported on a window no bigger than the Thouless timescale.

In Fig.~\ref{fig:chaotic}(a), the exact $\nu(\tau)$ frontier of an MFIM is compared against the Lorentzian spectral filtering of $H_q$. The filtered-$H_q$ curve tracks the frontier very well from $\tau\approx 10$ up to the point where the frontier is no longer a $k=1$ operator.

\subsection{Ansatz based on Ref.~\cite{kimSlowestLocalOperators2015a}} We also compare our formulation to a previous work on slow local operators~\cite{kimSlowestLocalOperators2015a}. Ref.~\cite{kimSlowestLocalOperators2015a} proposed the following ansatz on an infinitely long spin chain:
\begin{equation}
O = \sum_{x=-\infty}^{+\infty} e^{iqx} T^x[O_M], \label{eq:PRE15-ansatz}
\end{equation}
where $O_M$ is a traceless operator supported on $M$ consecutive sites and $T$ is the lattice translation operator. Such an operator is guaranteed to carry the given momentum $q$. Within such an ansatz, we can compute the $\nu(\tau)$ frontier exactly with exact diagonalization for small $M$. Note that for any $q\neq 0$, any operator in this ansatz is also automatically orthogonal to powers of $H$, therefore no orthogonalization is needed apart from the requirement that $O_M$ is traceless. For finite $M$, this ansatz only covers a subset of all possible operators; we use stability-in-$M$ as a proxy that the $\nu(\tau)$ frontier we probe within this restricted ansatz is a faithful representation of the true frontier. Ref.~\cite{kimSlowestLocalOperators2015a} computes the minimum of $\|[H,O]\|/\|O\|$ within the ansatz, which exactly corresponds to the maximum value of $\tau$ per our definition among all admissible operators within this ansatz. Therefore, our computation of the $\nu(\tau)$ curve within this ansatz probes a finer structure: we obtain a curve whose endpoint $\tau$ value coincides with the result in Ref.~\cite{kimSlowestLocalOperators2015a}.

In Fig.~\ref{fig:hydro}(b), we use this ansatz at small $q$ to probe the hydrodynamic slow mode in an infinite chaotic spin chain. For small $q$, the frontier is expected to behave similarly as in Fig.~\ref{fig:thermdyn}(b): at $q\to 0$, the $\nu(\tau)$ curve displays an integrable-system-like plateau, since we do not orthogonalize against $H$, and $O=H$ is simple and exactly conserved in the $q\to 0$ limit; at small $q$, the curve is expected to stay close to the $q=0$ plateau up to some thermalization time $\tau_\ast$, and then start to deviate downwards. In Fig.~\ref{fig:hydro}(b), we confirm that the timescale on which the deviation $\Delta \nu = \nu|_{q=0} - \nu$ develops indeed scales as $\tau_\ast \sim q^{-2}$, consistent with the hydrodynamic prediction. We note that for finite $M$, $\tau_\text{max}$ would scale as $q^{-1}$ instead of $q^{-2}$ due to finite-size effects.

In Fig.~\ref{fig:chaotic}(b,c), we also compare the $\nu(\tau)$ frontier obtained within the ansatz at $q=0$ against the finite-$L$ frontier. Fig.~\ref{fig:chaotic}(b) shows this comparison for a Floquet KIM, in which the two families of curves are observed to converge to a common limit as $L$ and $M$ are increased, suggesting that the moderate-$\tau$ regime of the finite-$L$ curves is representative of the behavior of this model in the thermodynamic limit. In Fig.~\ref{fig:chaotic}(c), this comparison is shown for an MFIM. A rough agreement also exists, although less clean than in the Floquet case.

\subsection{Solving for $\nu(\tau)$ with DMRG} As the $\nu(\tau)$ frontier can be obtained via a generalized eigenvalue problem, it is tempting to use density matrix renormalization group (DMRG)~\cite{schollwockDensitymatrixRenormalizationGroup2005a,schollwockDensitymatrixRenormalizationGroup2011b,fishmanITensorSoftwareLibrary2022,ITensor-r0.3} to solve for it in beyond-ED-size systems. Implementation details are available in the online repository~\cite{\CODE}.

However, we found that DMRG does not perform well in finding the slow mode beyond the small-$\tau$ regime. In Fig.~\ref{fig:dmrg}(a), it is shown that the $\nu(\tau)$ frontier for an MFIM obtained from DMRG tracks the exact frontier very well in an $L=12$ system for small $\tau$; however, at around $\tau=15$, the algorithm stalls, and the bond dimension $\chi$ required to reproduce the ED result rises steeply. In the numerical experiments we perform on this system, DMRG is unable to converge to $\tau=20$ up to $\chi=768$. Fig.~\ref{fig:dmrg}(b) shows the DMRG frontier for several $L$. It is shown that there is a $\tau$-barrier at each $L$, below which $\chi=128$ converges decently, and beyond which the bond dimension requirement blows up. We attribute this to the fact that the SSO, understood as a Lorentzian modulation of some local operator, requires a highly entangled tail to stay nearly conserved, which cannot be captured with a matrix product operator ansatz. For the MFIM, the DMRG wall sits below the $\tau$ at which the SSO becomes the modulated hydrodynamic operator and displays a $\tau\sim L^2$ collapse.

\bibliographystyle{apsrev4-2-etal}
\bibliography{MPS-CB}

\end{document}